\documentclass[11pt]{article}
\usepackage{amsmath,amsthm,amssymb}
\usepackage{graphicx,psfrag,epsf}
\usepackage{enumerate}

\usepackage{mathdots}
\usepackage[latin1]{inputenc}
\usepackage{natbib}
\usepackage{url}
\usepackage{subfigure} % crea sottotabelle e sottofigure
\usepackage{ctable}
\usepackage{color}
\usepackage{colortbl}
\usepackage[format=hang, indention=-1.4cm, font=footnotesize, labelfont=bf]{caption}
\usepackage{lscape}
\usepackage{bbm}
\usepackage{todonotes}
\usepackage{bbold}
\usepackage{enumitem}

\usepackage{hyperref}

\def\bx{\boldsymbol{x}}
\def\bX{\boldsymbol{X}}

\def\bS{\boldsymbol{S}}

\def\bu{\boldsymbol{u}}
\def\bU{\boldsymbol{U}}

\def\bI{\boldsymbol{I}}
\def\bmu{\boldsymbol{\mu}}
\def\btheta{\boldsymbol{\theta}}
\def\bSigma{\boldsymbol{\Sigma}}
\def\bOmega{\boldsymbol{\Omega}}

\def\bPsi{\boldsymbol{\Psi}}
\def\bzero{\boldsymbol{0}}

\newcommand{\real}{I\hspace{-1.0mm}R}
\newcommand{\indic}{\text{1\hspace{-1.0mm}I}}

\theoremstyle{definition}
\newtheorem{defi}{Definition}[section]

\newtheorem{theorem}{Theorem}[section]
\newtheorem{cor}{Corollary}[section]
\newcommand{\blind}{0}

\def\permille{\ensuremath{{}^\text{o}\mkern-5mu/\mkern-3mu_\text{oo}}}

\begin{document}

\def\spacingset#1{\renewcommand{\baselinestretch}%
{#1}\small\normalsize} \spacingset{1}

%%%%%%%%%%%%%%%%%%%%%%%%%%%%%%%%%%%%%%%%%%%%%%%%%%%%%%%%%%%%%%%%%%%%%%%%%%%%%%

\if0\blind
{
  \title{\bf The multivariate tail-inflated normal distribution \\ and its application in
finance}
  \author{Antonio Punzo\\
    \small Dipartimento di Economia e Impresa\\
		\small Universit\`{a} di Catania\\ \\
    Luca Bagnato\\
    \small Dipartimento di Scienze Economiche e Sociali\\
		\small Universit\`{a} Cattolica del Sacro Cuore}
  \maketitle
} \fi

%\if1\blind
%{
  %\bigskip
  %\bigskip
  %\bigskip
  %\begin{center}
    %{\LARGE\bf The multivariate tail-inflated normal distribution}
%\end{center}
  %\medskip
%} \fi

\bigskip
\begin{abstract}
This paper introduces the multivariate tail-inflated normal (MTIN) distribution, an elliptical heavy-tails generalization of the multivariate normal (MN).
% having heavier tails for the occurrence of mild outliers.
The MTIN belongs to the family of MN scale mixtures by choosing a convenient continuous uniform as mixing distribution.
Moreover, it has a closed-form for the probability density function characterized by only one additional ``inflation'' parameter, with respect to the nested MN, governing the tail-weight.
The first four moments are also computed; interestingly, they always exist and the excess kurtosis can assume any positive value.
The method of moments and maximum likelihood (ML) are considered for estimation.
As concerns the latter, a direct approach, as well as a variant of the EM algorithm, are illustrated.
The existence of the ML estimates is also evaluated.
Since the inflation parameter is estimated from the data, robust estimates of the mean vector and covariance matrix of the nested MN distribution are automatically obtained by down-weighting.
Simulations are performed to compare the estimation methods/algorithms and to investigate the ability of AIC and BIC to select among a set of candidate elliptical models. 
For illustrative purposes, the MTIN distribution is finally fitted to multivariate financial data where its usefulness is also shown in comparison with other well-established multivariate elliptical distributions.
\end{abstract}

\noindent%

%\spacingset{1.45} % DON'T change the spacing!

\section{Introduction}
\label{sec:intro}

Statistical inference dealing with continuous multivariate data is commonly focused on the multivariate normal (MN) distribution, with mean vector $\bmu$ and covariance matrix $\bSigma$, due to its computational and theoretical convenience.
However, for many real phenomena, the tails of this distribution are lighter than required.
This is often due to the presence of mild outliers \citep[][p.~4]{Ritt:Robu:2015}. 
Outliers are defined as ``mild'', with respect to the MN distribution (which is the reference distribution), when they do not really deviate from the MN and are not strongly outlying, but rather they produce an overall distribution that is too heavy-tailed to be modeled by the MN (\citealp{Mazz:Punz:Mixt:2017}, \citealp{morris19}, and \citealp{PunTor20}); for a discussion about the concept of reference (or target) distribution, see \citet{Davi:Gath:Thei:1993}.
Therefore, mild outliers can be modeled by means of a heavy-tailed elliptically symmetric (or elliptically contoured or, simply, elliptical) distribution embedding the MN distribution as a special case \citep[][p.~79]{Ritt:Robu:2015}.

A classical way to define such a larger model is by means of the MN scale mixture (also known as MN variance mixture); it is a finite/continuous mixture of MN distributions on $\bSigma$ obtained via a convenient discrete/continuous mixing distribution with positive support.
The MN scale mixture simply alters the tail behavior of the MN distribution while leaving the resultant distribution elliptically contoured.
This is made possible by the additional parameters of the mixing distribution which, in the scale mixture, govern the deviation from normality in terms of tail weight.
Regardless of the mixing distribution considered, expressing a multivariate elliptical distribution as a MN scale mixture enables (among other): natural/efficient computation of the maximum likelihood (ML) estimates via expectation-maximization (EM) based algorithms (\citealp{Lang:Sins:Norm:1993} and \citealp{Bala:Leiv:Sanh:Vilc:Esti:2009}), robust ML estimation of the parameters $\bmu$ and $\bSigma$ based on down-weighting of mild outliers (\citealp{Peel:McLa:Robu:2000}, \citealp{Punz:McNi:Robu:2016}, and \citealp{Farc:Punz:Robu:2019}), and efficient Monte Carlo calculations (\citealp{Rell:Tech:1970}, \citealp{Andr:Mall:Scal:1974}, \citealp{Efro:Olsh:HowB:1978} and \citealp{Choy:Chan:Scal:2008}).
Moreover, robustness studies have often used MN scale mixtures for simulation and in the analysis of outlier models; see \citet{West:Outl:1984,West:Onsc:1987} and the references therein.

The MN scale mixture family includes many classical distributions, such as the multivariate $t$ (M$t$; \citealp{Lang:Litt:Tayl:Robu:1989} and \citealp{Kotz:Nada:Mult:2004}), obtained by using a convenient gamma as mixing distribution, the multivariate contaminated normal (MCN) of \citet[][see also \citealp{Aitk:Wils:Mixt:1980}]{Tuke:Asur:1960}, defined under a convenient Bernoulli mixing distribution (see, e.g., \citealp{Yama:Robu:2004} and \citealp{Punz:Blos:McNi:High:2019}), and the multivariate symmetric generalized hyperbolic (MSGH; \citealp{McNe:Frey:Embr:Quan:2005}) obtained if a generalized inverse Gaussian is considered as mixing distribution. 

In the present paper, we add a new member to the MN scale mixture family: the multivariate tail-inflated normal (MTIN) distribution.
It is obtained by using a convenient continuous uniform as mixing distribution.
%The paper proceeds as follows.
In Section~\ref{subsec:pdf} we show that the proposed distribution has a closed-form probability density function (pdf) which is characterized by a single additional parameter, with respect to $\bmu$ and $\bSigma$ of the nested MN, governing the tail weight.
We also make explicit the representations of the MTIN distribution as MN scale mixture (Section~\ref{subsubsec:Multivariate normal scale mixtures}) and as elliptical distribution (Section~\ref{subsubsec:Elliptical distributions}).
In Section~\ref{sec:moments} we provide the first four moments of the MTIN distribution, which are those of practical interest (see also \appendixname~\ref{app:Teorema sui momenti della MTIN}).
We estimate the parameters by the method of moments (\appendixname~\ref{sec:Method of moments}) and maximum likelihood (ML; Section~\ref{sec:ML}); for the latter, we illustrate two alternative algorithms/methods to obtain these estimates (a direct approach in Section~\ref{subsec:Direct approach} and an ECME algorithm in Section~\ref{subsec:Alternative ECME algorithm}).
Always for the ML method, in Section~\ref{sec:Existence of MLE} we discuss the existence of the estimates.
%One of them, the EM algorithm, has the advantage to have closed-form updates of all the parameters during the M-step (Section~\ref{subsec:EM algorithm}).
The ECME algorithm allows us to highlights how the ML estimates of $\bmu$ and $\bSigma$ are robust in the sense that mild outliers are down-weighted in the computation of these parameters (see Section~\ref{sec:Some notes on robustness} for details). 
Two simulated data analyses are designed to compare the estimation methods discussed above in terms of parameter recovery and computational time required (Section~\ref{subsec:Quality of estimates and computational times}), and to evaluate the appropriateness of AIC and BIC in selecting among a set of candidate elliptical models that can be considered as natural competitors of our MTIN distribution (Section~\ref{subsec:Model selection: AIC versus BIC}).
At last, we illustrate the proposed model by analyzing financial multivariate data related to four of the publicly owned companies considered by the Dow Jones index (Section~\ref{sec:Real data analysis}); here, we also compare the performance of the MTIN distribution with other heavy-tailed elliptical distributions which are well-known in the financial literature.
We conclude the paper with a brief discussion in Section~\ref{sec:Discussion}.

\section{Multivariate tail-inflated normal distribution}
\label{sec:MTIN}

\subsection{Probability density function}
\label{subsec:pdf}

\begin{defi}[Probability density function]\label{defi1}
A $d$-dimensional random vector $\bX$ is said to have a $d$-variate tail-inflated normal distribution with mean vector $\bmu \in \real^d$, $d\times d$ scale matrix $\bSigma$, and inflation parameter $\theta \in \left(0,1\right)$, in symbols $\bX\sim \mathcal{TIN}_d\left(\bmu,\bSigma,\theta\right)$, if its pdf is given by  
\begin{equation}
f_{\text{TIN}}\left(\bx;\bmu,\bSigma,\theta\right)  =
\frac{2\pi^{-\frac{d}{2}}\left|\bSigma \right|^{-\frac{1}{2}}}{\theta \left[\delta\left(\bx;\bmu,\bSigma\right)\right]^{\left(\frac{d}{2}+1\right)}}    
\left[ \Gamma\left(\frac{d}{2}+1,(1-\theta) \frac{\delta\left(\bx;\bmu,\bSigma\right)}{2} \right) - \Gamma\left(\frac{d}{2}+1,\frac{\delta\left(\bx;\bmu,\bSigma\right)}{2} \right) \right],  
\label{eq:MTIN pdf}
\end{equation}
where $\Gamma\left(\cdot,\cdot\right)$ is the upper incomplete gamma function, $\left|\cdot\right|$ is the determinant, and \linebreak $\delta\left(\bx;\bmu,\bSigma\right)=\left(\bx-\bmu\right)'\bSigma^{-1}\left(\bx-\bmu\right)$ denotes the squared Mahalanobis distance between $\bx$ and $\bmu$ (with $\bSigma$ as the covariance matrix).
\end{defi}
Alternative formulations of \eqref{eq:MTIN pdf} are given in \appendixname~\ref{app:Alternative formulation of the pdf}.

Theorem~\ref{theo:normalcase} proves that the limiting form of \eqref{eq:MTIN pdf} as $\theta\rightarrow 0$ is the pdf of the $d$-variate normal distribution with mean vector $\bmu$ and covariance matrix $\bSigma$; in symbols $\bX\sim \mathcal{N}_d\left(\bmu,\bSigma\right)$.
\begin{theorem}[Limiting normal form]
\label{theo:normalcase}
The pdf $f_{\text{N}}\left(\bx;\bmu,\bSigma\right)$ of $\bX\sim \mathcal{N}_d\left(\bmu,\bSigma\right)$ can be obtained as a special case of \eqref{eq:MTIN pdf} when $\theta\rightarrow 0$; in formula
\begin{equation}
\lim_{\theta\rightarrow 0} f_{\text{TIN}}\left(\bx;\bmu,\bSigma,\theta\right) = f_{\text{N}}\left(\bx;\bmu,\bSigma\right).
\label{eq:limiting normal case}
\end{equation}
\end{theorem}
%%%%%%%%%%%
%% PROOF %%
%%%%%%%%%%%
\begin{proof}
To prove the theorem we work on the limit, as $\theta\rightarrow 0$, of 
\begin{equation}
\frac{1}{\theta}   \displaystyle \left[ \Gamma\left(\frac{d}{2}+1,(1-\theta) \frac{\delta\left(\bx;\bmu,\bSigma\right)}{2} \right) - \Gamma\left(\frac{d}{2}+1,\frac{\delta\left(\bx;\bmu,\bSigma\right)}{2} \right) \right] ,
\label{eq:limit0}
\end{equation}
which is the part of \eqref{eq:MTIN pdf} depending on $\theta$.
The limit of \eqref{eq:limit0}, as $\theta \rightarrow 0$, is of the form zero over zero.
Using L'Hospital's rule, such a limit becomes
\begin{equation}
\lim_{\theta\rightarrow 0}
\label{eq:limit2}\displaystyle 
(1-\theta)^{\frac{d}{2}} 
\left[\frac{\delta\left(\bx;\bmu,\bSigma\right)}{2}\right]^{\left(\frac{d}{2}+1\right)}
\exp\left(-\frac{ \delta \left(\bx;\bmu,\bSigma\right) }{2}\right)
= \left[\frac{\delta\left(\bx;\bmu,\bSigma\right)}{2}\right]^{\left(\frac{d}{2}+1\right)} \exp\left(-\frac{ \delta \left(\bx;\bmu,\bSigma\right) }{2}\right).
\end{equation}
When \eqref{eq:limit2} is substituted in \eqref{eq:MTIN pdf} instead of \eqref{eq:limit0}, $f_{\text{N}}\left(\bx;\bmu,\bSigma\right)$ is obtained.
\end{proof}

\subsection{Representations}
\label{subsec:Representations}

In this section we show that, if $\bX\sim \mathcal{TIN}_d\left(\bmu,\bSigma,\theta\right)$, then its distribution has a twofold representation as MN scale mixture (Section~\ref{subsubsec:Multivariate normal scale mixtures}) and as multivariate elliptically symmetric distribution (Section~\ref{subsubsec:Elliptical distributions}).
The properties of these families are so implicitly inherited by our model. 

%%%%%%%%%%%%%%%%%%%%%%%%%%%%%%%%%%%
% Representation as scale mixture %
%%%%%%%%%%%%%%%%%%%%%%%%%%%%%%%%%%%

\subsubsection{Multivariate normal scale mixture representation}
\label{subsubsec:Multivariate normal scale mixtures}

For robustness sake, one of the most common ways to generalize the MN distribution is represented by the MN scale mixture (MNSM), also called MN variance mixture (MNVM), with pdf 
\begin{equation}
f_{\text{MNSM}}\left(\bx;\bmu,\bSigma,\btheta\right)=
\int_{S_h \subseteq \real_{>0}} f_{\text{N}}\left(\bx;\bmu,\bSigma/w\right)h\left(w;\btheta\right)dw,
\label{eq:MN mixture model}
\end{equation}
where 
%$f_{\text{N}}\left(\bx;\bmu,\bSigma\right)$ denotes the pdf of $\bX\sim \mathcal{N}_d\left(\bmu,\bSigma\right)$ and 
$h\left(w;\btheta\right)$ is the mixing probability density (or mass) function --- with support $S_h \subseteq \real_{>0}$ --- depending on the parameter(s) $\btheta$. 
%\textcolor{red}{The consistency property ensures that any marginal distribution of a random vector whose distribution belongs to a specific elliptical family also belongs to the family.
%Elliptical distributions with this property must be a mixture of normal distributions \citep{Kano:Conc:1994}.} \todo[fancyline,size=\tiny,color=green!40]{Rivedere sta frase alla luce di ciò che diciamo complessivamente.}
The pdf in \eqref{eq:MN mixture model} is unimodal, elliptically symmetric, and guarantees tails heavier than those of the MN distribution (see, e.g., \citealp{Barn:Kent.Sore:Inte:Norm:1982}, \citealp[][Section~2.6]{Fang:Kotz:Ng:Symm:2013}, \citealp{Yama:Robu:2004} and \citealp{McLa:Peel:fini:2000}, Section~7.4). 
The tail weight of $f_{\text{MNSM}}$ is governed by $\btheta$. 

Examples of distributions belonging to the MNSM family are: the multivariate contaminated normal (MCN), the multivariate $t$ (M$t$), the multivariate symmetric generalized hyperbolic (MSGH), the multivariate symmetric hyperbolic (MSH), the multivariate symmetric variance-gamma (MSVG), and the multivariate symmetric normal inverse Gaussian (SNIG); for details about these special cases, as well as about the properties of the MNSM family, see \citet[][Section~3.2.1]{McNe:Frey:Embr:Quan:2005}. 

In Theorem~\ref{theo:MN scale mixture} we show that $\bX\sim \mathcal{TIN}_d\left(\bmu,\bSigma,\theta\right)$ can be represented as a MNSM by considering a convenient (continuous) uniform as mixing distribution.
%It also includes some well-known models; for example, if $W\sim\text{gamma}\left(\nu/2,\nu/2\right)$, then we obtain the $t$ distribution with location parameter $\bmu$, positive definite inner product matrix $\bSigma$, and $\nu$ degrees of freedom.
%The following theorem shows as the TIN distribution derives from a normal scale mixture.
\begin{theorem}[Scale mixture representation]\label{theo:MN scale mixture}
The pdf in \eqref{eq:MTIN pdf} can be obtained as a special case of the pdf in \eqref{eq:MN mixture model} by considering a uniform on $\left(1-\theta,1\right)$ as mixing pdf.
In formula,
\begin{equation}
f_{\text{TIN}}\left(\bx;\bmu,\bSigma,\theta\right)
= \int_{1-\theta}^1 f_{\text{N}}\left(\bx;\bmu,\bSigma/w\right) h_{\text{U}}\left(w;1-\theta,1\right) dw,
\label{eq:den2}
\end{equation}
where $h_{\text{U}}\left(w;a,b\right)$ denotes the pdf of a uniform on $\left(a,b\right)$.
\end{theorem}
\begin{proof}
The pdf in \eqref{eq:den2} can be written as
\begin{equation}
f_{\text{TIN}}\left(\bx;\bmu,\bSigma,\theta\right) = \frac{1}{\theta \left(2 \pi\right)^{\frac{d}{2}}} \left|\bSigma \right|^{-\frac{1}{2}} \int_{1-\theta}^1 w^{\frac{d}{2}} \exp\left[-\frac{w}{2}{\delta\left(\bx;\bmu,\bSigma\right)}\right] dw.
\label{eq:withintegral}
\end{equation}
By noting that
\begin{equation}
\int_{1-\theta}^1 w^{\frac{d}{2}} \exp\left[-\frac{w}{2}{\delta\left(\bx;\bmu,\bSigma\right)}\right]dw = \left[\frac{2}{\delta\left(\bx;\bmu,\bSigma\right)}\right]^{\left(\frac{d}{2}+1\right)} \left[ \Gamma\left(\frac{d}{2}+1,(1-\theta) \frac{\delta\left(\bx;\bmu,\bSigma\right)}{2} \right) - \Gamma\left(\frac{d}{2}+1,  \frac{\delta\left(\bx;\bmu,\bSigma\right)}{2} \right) \right],
\label{eq:uppergamma}
\end{equation}
the proof of the theorem is straightforward. 
\end{proof}
A more direct interpretation of Theorem~\ref{theo:MN scale mixture} can be given by the hierarchical representation of $\bX\sim \mathcal{TIN}_d\left(\bmu,\bSigma,\theta\right)$ as
\begin{equation}
\begin{array}{rcl}
W          & \sim & \mathcal{U}\left(1-\theta,1\right) \\
\bX | W=w  & \sim & \mathcal{N}_d\left(\bmu,\bSigma/w\right), 
\end{array}
\label{eq:hierarchy}
\end{equation}
where $\mathcal{U}\left(1-\theta,1\right)$ denotes a uniform distribution on $\left(1-\theta,1\right)$.
This alternative way to see the MTIN distribution is useful for random generation and for the implementation of EM-based algorithms (as we will better appreciate in Section~\ref{sec:ML}).

%%%%%%%%%%%%%%%%%%%%%%%%%%%%%
% Elliptical representation %
%%%%%%%%%%%%%%%%%%%%%%%%%%%%%

\subsubsection{Elliptical representation}
\label{subsubsec:Elliptical distributions}

As well-documented in the literature, the class of elliptical distributions has several good properties (\citealp{Camb:Onth:1981}, \citealp[][Chapter~2.5]{Fang:Kotz:Ng:Symm:2013} and \citealp[][p.~357]{Rach:Hoec:Fabo:Foca:Prob:2010}) and the class of MNSMs is contained therein (\citealp[][p.~247]{Bing:Kies:Rudi:Semi:2002}, \citealp[][p.~61]{McNe:Frey:Embr:Quan:2005} and \citealp[][p.~48]{Fang:Kotz:Ng:Symm:2013}).  
Therefore, the MTIN distribution is also elliptical and Theorem~\ref{theo: Elliptical representation} presents its elliptical representation.  

\begin{theorem}[Elliptical representation]\label{theo: Elliptical representation}
If $\bX\sim \mathcal{TIN}_d\left(\bmu,\bSigma,\theta\right)$, then $\bX$ is elliptically distributed with the following stochastic representation
\begin{equation}
\bX = \bmu + \boldsymbol{\Lambda} \bU \displaystyle \frac{T}{\sqrt{W}},
\label{eq:ell}
\end{equation}
where $\boldsymbol{\Lambda}$ is a $d\times d$ matrix such that $\boldsymbol{\Lambda}\boldsymbol{\Lambda}' = \bSigma$, $\bU$ is a $d$-variate random vector uniformly distributed on the unit hypersphere with $d-1$ topological dimensions $\left\{ \bu \in \real^d : ||\bu|| = 1 \right\}$, $T$ is a random variable having a $\chi$ distribution with $d$ degrees of freedom, and $W \sim \mathcal{U}\left(1-\theta,1\right)$.
Note that $\bU$, $T$, and $W$ are independent. 
\end{theorem}
\begin{proof}
To find the distribution of $\bX$ as defined by \eqref{eq:ell} we need to derive the pdf $f_R\left(\cdot;\theta\right)$ of the random variable
$$
R = \displaystyle \frac{T}{\sqrt{W}},
$$
and then use the relationship \citep[see][]{Bagn:Punz:Zoia:Them:2016}
\begin{equation}
  f_{\bX}\left(\bx;\bmu,\bSigma, \theta\right) = \left|\bSigma\right|^{-\frac{1}{2}} \frac{\Gamma\left(\frac{d}{2}\right)}{2 \pi^{d/2}} \sqrt{\delta\left(\bx;\bmu,\bSigma\right)}^{-(d-1)} f_{R} \left(\sqrt{\delta\left(\bx;\bmu,\bSigma\right)};\theta\right).%\quad t=\left(\bx-\bmu\right)'\bSigma^{-1}\left(\bx-\bmu\right).
  \label{eq:dim}
\end{equation}
The random variable $V=1/\sqrt{W}$ has density
\begin{equation}
  f_{V}\left(v;\theta\right) = \frac{2}{\theta v^3} \indic_{\left(1,\frac{1}{1-\theta}\right)} \left(v\right),
  \label{eq:denW}
\end{equation}
where $\indic_A\left(\cdot\right)$ is the indicator function on the set $A$.
By definition, the density of the product of the independent random variables $T$ and $V$ is given by
\begin{eqnarray}
  f_{R}\left(r;\theta\right) &=& \int_{\real} f_{V}\left(v;\theta\right) f_{T}\left(\frac{r}{v}\right) \frac{1}{|v|} dv \nonumber \\
  &=& \frac{2^{1-\frac{d}{2}} r^{d-1}}{\theta \Gamma\left( \frac{d}{2}\right)} \int_{1-\theta}^1 v^{\frac{d}{2}} \exp\left(-\frac{v}{2}{r^2}\right)dv.
  \label{eq:denR}
\end{eqnarray}
By using \eqref{eq:denR} in \eqref{eq:dim}, we obtain the pdf of the MTIN distribution in \eqref{eq:MTIN pdf}.
%The pdf in \eqref{eq:MTIN pdf} is obtained by substituting \eqref{eq:denR} in \eqref{eq:dim}.
%The pdf obtained by substituting \eqref{eq:denR} in \eqref{eq:dim} coincides with \eqref{eq:MTIN pdf}, that is the TIN distribution.
\end{proof}

%\subsection{Consistency property}
%\label{subsec: Consistency property}
%
%\begin{cor}[Consistency property]
%If $\bX\sim \mathcal{TIN}_d\left(\bmu,\bSigma,\theta\right)$, then any marginal distribution of $\bX$ also belongs to the TIN family.
%\end{cor}
%\begin{proof}
  %The proof is direct if the Theorem~1 in \citet{Kano:Conc:1994} is applied.
  %In particular, the stochastic representation in Theorem~\ref{theo: Elliptical representation} assures that the consistency property is satisfied. \todo[fancyline,size=\tiny,color=green!40]{Perché? Si può spiegare meglio?}
%\end{proof}

\subsection{Moments}
\label{sec:moments}

In this section we provide some of the moments of the MTIN distribution.
In addition to the mean, we give the first two moments of higher even order, i.e.~the covariance matrix and kurtosis, which are helpful in assessing the influence of mild outliers on the distribution \citep[][p.~307]{Rach:Hoec:Fabo:Foca:Prob:2010}; indeed, these are the moments of practical interest for people using multivariate heavy-tailed elliptical distributions.
As measure of kurtosis, as usual, we consider
% we use a natural multivariate extension of the classical notion in the univariate case, that is 
\begin{equation*}
\mbox{Kurt}\left(\bX\right) = E\left\{\left[\left(\bX - \bmu \right)' \bSigma^{-1}\left(\bX - \bmu \right)\right]^2\right\},
\label{eq:kurt}
\end{equation*}
where $\bX$ is a random vector having mean $\bmu$ and covariance matrix $\bSigma$ \citep{Mard:Meas:1970}.
\begin{theorem}[MTIN: mean vector, covariance matrix and kurtosis]\label{theo: MTIN moments}
If $\bX\sim \mathcal{TIN}_d\left(\bmu,\bSigma,\theta\right)$, then 
\begin{eqnarray}
\mbox{E}\left(\bX\right) &=& \bmu,   \label{eq:mean}  \\  
\text{Var}\left(\bX\right) &=& v\left(\theta\right) \bSigma, \label{eq:variance} \\
%\text{Skew}\left(\bX\right) &=& 0, \label{eq:skewness} \\
\mbox{Kurt}\left(\bX\right) &=& k\left(\theta\right) d\left(d+2\right), \label{eq:kurtosis}
\end{eqnarray}
where 
\begin{eqnarray}
v\left(\theta\right) &=& \displaystyle - \frac{\log(1-\theta)}{\theta}, \label{eq:variance factor} \\
k\left(\theta\right) &=& \displaystyle \frac{\theta^2}{(1-\theta)\log^2(1-\theta)}. \label{eq:kurtosis factor} 
\end{eqnarray}
\end{theorem}
\begin{proof}
See \appendixname~\ref{app:Teorema sui momenti della MTIN}.
\end{proof}

Based on \eqref{eq:variance}, the covariance matrix $\text{Var}\left(\bX\right)$ of the MTIN distribution is proportional to the covariance matrix $\bSigma$ of the nested MN distribution, with proportionality factor $v\left(\theta\right)$ depending on $\theta$.
In particular, the function $v\left(\theta\right): \left(0,1\right) \rightarrow \left(1,\infty\right)$, given in \eqref{eq:variance factor} and graphically represented via a solid line in \figurename~\ref{fig:Multiplicative Factors}, is increasing in $\theta$.
\begin{figure}[!ht]
\centering
\resizebox{0.65\textwidth}{!}{
\includegraphics{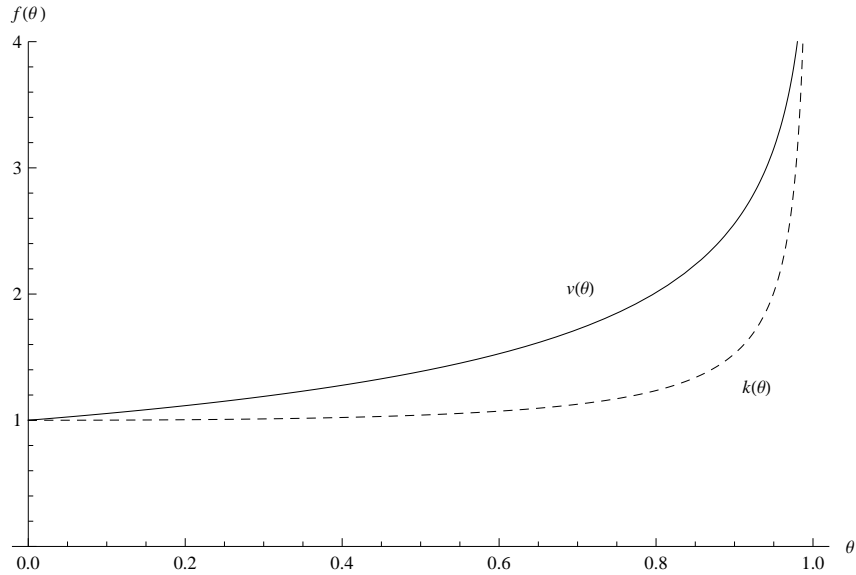} %[width=0.99\textwidth]
}
\caption{
Multiplicative factors $v\left(\theta\right)$ (solid line) and $k\left(\theta\right)$ (dashed line).
% which respectively characterize, with respect to the MN distribution, the covariance matrix and the kurtosis of the MTIN distribution.
}
\label{fig:Multiplicative Factors}
\end{figure}
Because of the values $v\left(\theta\right)$ can assume, $\text{Var}\left(\bX\right)$ is an inflated version of $\bSigma$, without limits about the degree of inflation that tends to $\infty$ when $\theta \rightarrow 1$.
However, if $\theta \rightarrow 0$, then $v\left(\theta\right) \rightarrow 1$; so, under this case, $\text{Var}\left(\bX\right)$ tends to $\bSigma$.

According to \eqref{eq:kurtosis factor}, $\text{Kurt}\left(\bX\right)$ is proportional to the kurtosis $d\left(d+2\right)$ of the (nested) MN distribution, with proportionality factor $k\left(\theta\right)$ depending on $\theta$.
In particular, the function $k\left(\theta\right): \left(0,1\right) \rightarrow \left(1,\infty\right)$, given in \eqref{eq:kurtosis factor} and graphically represented by a dashed line in \figurename~\ref{fig:Multiplicative Factors}, is increasing in $\theta$.
Therefore, $\text{Kurt}\left(\bX\right)$ is greater than $d\left(d+2\right)$ meaning that the MTIN distribution allows for leptokurtosis.
An upper bound for $\text{Kurt}\left(\bX\right)$ does not exist since $k\left(\theta\right) \rightarrow \infty$ when $\theta \rightarrow 1$.
However, if $\theta \rightarrow 0$, then $k\left(\theta\right) \rightarrow 1$; so, under this case, $\text{Kurt}\left(\bX\right)$ tends to $d\left(d+2\right)$.
%, which is the kurtosis of the MN distribution.
There is another aspect to be noted from \figurename~\ref{fig:Multiplicative Factors}.
While $v\left(\theta\right)$ is a smoothly increasing function of $\theta$, $k\left(\theta\right)$ is approximately flat on the left, increasing suddenly only when $\theta$ is close to 1.
From a practical point of view, this means that we need an high $\theta$-value to have a relevant excess kurtosis.

\section{Maximum likelihood estimation}
\label{sec:ML}

Several estimators of the parameters $\bPsi=\left\{\bmu,\bSigma,\theta\right\}$ of the MTIN distribution may be considered. 
Among them, the maximum likelihood (ML) estimator is the most attractive because of its asymptotic properties \citep[see, e.g.,][p.~206]{Pfan:Hamb:Para:1994}.
Naturally, other estimation methods can be used; \appendixname~\ref{sec:Method of moments} gives details about the method of moments (MM).
% are postponed to \appendixname~\ref{sec:Method of moments}.} 
Below we discuss two possible ways to obtain the ML estimates of $\bPsi$; for their comparison, see Section~\ref{subsec:Quality of estimates and computational times}.

\subsection{Direct approach}
\label{subsec:Direct approach}

Given a random sample $\bx_1,\ldots,\bx_n$ (observed data) from $\bX\sim \mathcal{TIN}_d\left(\bmu,\bSigma,\theta\right)$, the ML estimation method is based on the maximization of the (observed-data) log-likelihood function 
%.
%From \eqref{eq:MTIN pdf}, the log-likelihood function for the MTIN distribution is
\begin{eqnarray}
l\left(\bPsi\right)
%&=&\sum_{i=1}^n\log\left[f_{\text{TIN}}\left(\bx_i;\bmu,\bSigma,\theta\right)\right] \nonumber\\
&=& n\log\left(2\right)-\frac{nd}{2}\log\left(\pi\right)-\frac{n}{2}\log \left|\bSigma\right| - n\log\left(\theta\right) - \left(\frac{d}{2}+1\right)\sum_{i=1}^n \log\left[\delta\left(\bx_i;\bmu,\bSigma\right)\right] \nonumber\\
&& + \sum_{i=1}^n \log \left[ \Gamma\left(\frac{d}{2}+1,(1-\theta) \frac{\delta\left(\bx_i;\bmu,\bSigma\right)}{2} \right) - \Gamma\left(\frac{d}{2}+1,\frac{\delta\left(\bx_i;\bmu,\bSigma\right)}{2} \right) \right], 
\label{eq:MTIN observed-data log-likelihood}
\end{eqnarray} 
which does not admit a closed form solution.
% for the ML estimate of $\bPsi$. 
Furthermore, maximizing \eqref{eq:MTIN observed-data log-likelihood} with respect to $\bPsi$ is a constrained problem due to $\bSigma$ and $\theta$.
To make the maximization problem unconstrained, we apply the following reparameterization for $\bSigma$ and $\theta$.
According to the Cholesky decomposition we write
$$
\bSigma = \bOmega' \bOmega ,
$$
where $\bOmega$ is an upper triangular matrix with real-valued entries.
As concerns $\theta$ we write
$$
\theta=\frac{\exp(\gamma)}{1+\exp(\gamma)},
$$
where $\gamma \in \real$.
According to the above parametrization, the log-likelihood function in \eqref{eq:MTIN observed-data log-likelihood} can be re-written as
\begin{eqnarray}
l\left(\bPsi^*\right)&=& n\log\left(2\right)-\frac{nd}{2}\log\left(\pi\right)-\frac{n}{2} \log \left| \bOmega' \bOmega \right| -n \left\{\gamma - \log\left[1+\exp(\gamma)\right]\right\} \nonumber \\
&& - \left(\frac{d}{2}+1\right)\sum_{i=1}^n \log\left[\delta\left(\bx_i;\bmu,\bOmega' \bOmega \right)\right] \nonumber\\
&& + \sum_{i=1}^n \log\left\{ \Gamma\left[\frac{d}{2}+1,\frac{\delta\left(\bx_i;\bmu,\bOmega' \bOmega \right)}{2 \left[1+\exp(\gamma)\right]} \right] - \Gamma\left[\frac{d}{2}+1,\frac{\delta\left(\bx_i;\bmu,\bOmega' \bOmega \right)}{2} \right] \right\},
\label{eq:MTIN observed-data log-likelihood rep}
\end{eqnarray}
where $\bPsi^*=\left\{\bmu,\bOmega, \gamma \right\}$.
%where $\bPsi^*=\left\{\bmu,\bOmega,\gamma\right\}$. 
Details about the first order partial derivatives of $l\left(\bPsi^*\right)$ in \eqref{eq:MTIN observed-data log-likelihood rep}, with respect to $\bmu$, $\bOmega $ and $\gamma$, are given in \appendixname~\ref{app:Partial derivatives}.
Operationally, we perform unconstrained maximization of $l\left(\bPsi^*\right)$ with respect to $\bPsi^*$ via the general-purpose optimizer \texttt{optim()} for \textsf{R} \citep{R:2018}, included in the \textbf{stats} package.
Different algorithms, such as the Nelder-Mead or BFGS, can be used for maximization.
They can be passed to \texttt{optim()} via the argument \texttt{method}.
Once the maximum is determined, the estimates for $\bSigma$ and $\theta$ are simply obtained by back-transformations.
For a comparison, in terms of parameter recovery and computational times, between Nelder-Mead and BFGS algorithms applied to the direct ML estimation of $\bPsi$, see Section~\ref{subsec:Quality of estimates and computational times}. 
%The BFGS algorithm, passed to \texttt{optim()} via the argument \texttt{method}, is used for maximization.}

\subsection{ECME algorithm}
\label{subsec:Alternative ECME algorithm}

To circumvent the complexity of the direct ML approach discussed in Section~\ref{subsec:Direct approach}, we could consider the application of the expectation-maximization (EM) algorithm \citep{Demp:Lair:Rubi:Maxi:1977}, which is the classical approach to find ML estimates for distributions belonging to the MNSM family.
However, for the MTIN, the M-step is not well-defined (for the motivations we give in \appendixname~\ref{app:application of the EM algorithm}), and the EM algorithm fails to converge.
This is not a serious problem because ML estimates of $\bPsi$ can be still found, much more efficiently, by using the expectation-conditional maximization either (ECME) algorithm \citep{Liu:Rubi:TheE:1994}.     
The higher efficiency of the ECME algorithm, with respect to the EM algorithm, has been also discussed with reference to another member of the MNSM family, the M$t$ distribution (see \citealp{Liu:Rubi:TheE:1994,Liu:Rubi:MLes:1995} and \citealp[][Section~5.8]{McLa:Kris:TheE:2007}).

The ECME algorithm is an extension of the expectation-conditional maximum (ECM) algorithm \citep{Meng:Rubin:Maxi:1993} which, in turn, is an extension of the EM algorithm \citep[see][Chapter~5, for details]{McLa:Kris:TheE:2007}. 
The ECM algorithm replaces the M-step of the EM algorithm by a number of computationally simpler conditional maximization (CM) steps.
The ECME algorithm generalizes the ECM algorithm by conditionally maximizing on some or all of the CM-steps the incomplete-data log-likelihood.
As for the EM and ECM algorithms, the ECME algorithm monotonically increases the likelihood and reliably converges to a stationary point of the likelihood function (see \citealp{Meng:VanD:TheE:1997} and \citealp[][Chapter~5]{McLa:Kris:TheE:2007}).
Moreover, \citet{Liu:Rubi:TheE:1994} found the ECME algorithm to be nearly always faster than both the EM and ECM algorithms in
terms of number of iterations, and that it can be faster in total computer time by orders of magnitude \citep[see also][Chapter~5]{McLa:Kris:TheE:2007}.

For the application of any variant of the EM algorithm, in the light of the hierarchical representation of the MTIN distribution given in \eqref{eq:hierarchy}, it is convenient to view the observed data as incomplete.
The complete data are $\left(\bx_1',w_1\right)',\ldots,\left(\bx_n',w_n\right)'$, where the missing variables $w_1,\ldots,w_n$ are defined so that 
\begin{equation}
\bX_i | W_i=w_i \sim \mathcal{N}_d\left(\bmu,\bSigma/w_i\right),
\label{eq:complete data X}
\end{equation}
independently for $i = 1, \ldots , n$, and
\begin{equation}
W_1,\ldots,W_n \stackrel{\text{i.i.d.}}{\sim} \mathcal{U}\left(1-\theta,1\right).
\label{eq:complete data W}
\end{equation}
Because of this conditional structure, the complete-data likelihood function $L_c\left(\bPsi\right)$ can be factored into the product of the conditional densities of $\bX_i$ given the $w_i$ and the marginal densities of $W_i$, i.e.
\begin{eqnarray}
L_c\left(\bPsi\right) &= & \prod_{i=1}^n f_{\text{N}}\left(\bx_i;\bmu,\bSigma/w_i\right) h_{\text{U}}\left(w_i;1-\theta,1\right) \nonumber \\
&=& \prod_{i=1}^n (2\pi)^{-\frac{d}{2}} w_i^{\frac{d}{2}}\left|\bSigma\right|^{-\frac{1}{2}} \exp\left[-\frac{w_i}{2}\delta\left(\bx_i;\bmu,\bSigma\right)\right] \frac{1}{\theta} \indic_{\left(1-\theta,1\right)}\left(w_i\right).
\label{eq:MTIN complete-data likelihood}
\end{eqnarray}
Accordingly, the complete-data log-likelihood function can be written as
\begin{equation}
l_c\left(\bPsi\right) = \log\left[L_c\left(\bPsi\right)\right] = l_{1c}\left(\bmu,\bSigma\right) + l_{2c}\left(\theta\right) ,
\label{eq:MTIN complete-data log-likelihood}
\end{equation}
where
\begin{equation}
l_{1c}\left(\bmu,\bSigma\right) = - \displaystyle \frac{nd}{2} \log(2\pi) - \frac{n}{2} \log\left|\bSigma\right| + \frac{d}{2} \sum_{i=1}^n \log\left(w_i\right) - \frac{1}{2} \sum_{i=1}^n w_i \delta\left(\bx_i;\bmu,\bSigma\right),
\label{eq:MTIN complete-data log-likelihood mu and Sigma}
\end{equation}
and
\begin{equation} 
l_{2c}\left(\theta\right) = - n \log\left(\theta\right) + \sum_{i=1}^n \log\left[\indic_{\left(1-\theta,1\right)}\left(w_i\right)\right].
\label{eq:MTIN complete-data log-likelihood theta}
\end{equation}

The ECME algorithm iterates between three steps, one E-step and two CM-steps, until convergence.
The two CM-steps arise from the partition of $\bPsi$ as $\left\{\bPsi_1,\Psi_2\right\}$, where $\bPsi_1=\left\{\bmu,\bSigma\right\}$ and $\Psi_2=\theta$.
These steps, for the generic $\left(r+1\right)$th iteration of the algorithm, $r=1,2,\ldots$, are detailed below. 

%%%%%%%%%%
% E-step %
%%%%%%%%%%

\subsubsection{E-step}
\label{subsec:E-step}

The E-step requires the calculation of 
\begin{equation}
Q\left(\bPsi|\bPsi^{(r)}\right) = Q_1\left(\bmu,\bSigma|\bPsi^{(r)}\right) + Q_2\left(\theta|\bPsi^{(r)}\right),
\label{eq:Q}
\end{equation}
the conditional expectation of $l_c\left(\bPsi\right)$ given the observed data $\bx_1,\ldots,\bx_n$, using the current fit $\bPsi^{(r)}$ for $\bPsi$.
In \eqref{eq:Q} the two terms on the right-hand side are ordered as the two terms on the right-hand side of \eqref{eq:MTIN complete-data log-likelihood}.
%It can be seen from \eqref{eq:MTIN complete-data log-likelihood mu and Sigma} and \eqref{eq:MTIN complete-data log-likelihood theta} that 
As well-explained in \citet[][p.~82]{McNe:Frey:Embr:Quan:2005}, in order to compute $Q\left(\bPsi|\bPsi^{(r)}\right)$ we need to replace any function $m\left(W_i\right)$ of the latent mixing variables which arise in \eqref{eq:MTIN complete-data log-likelihood mu and Sigma} and \eqref{eq:MTIN complete-data log-likelihood theta} by the quantities $E_{\bPsi^{(r)}}\left[m\left(W_i\right)|\bX_i=\bx_i\right]$, where the expectation, as it can be noted by the subscript, is taken using the current fit $\bPsi^{(r)}$ for $\bPsi$, $i=1,\ldots,n$.
To calculate these expectations we can observe that the conditional pdf of $W_i|\bX_i=\bx_i$ satisfies $f\left(w_i|\bx_i;\bPsi\right) \propto  f\left(w_i,\bx_i;\bPsi\right)$, up to some constant of proportionality.
In detail,
\begin{eqnarray}
f\left(w_i|\bx_i;\bPsi\right) & \propto & f\left(w_i,\bx_i;\bPsi\right) \nonumber \\
&\propto& w_i^{\frac{d}{2}+1-1} \exp\left[-\frac{\delta\left(\bx_i;\bmu,\bSigma\right)}{2}w_i\right] \indic_{\left(1-\theta,1\right)}\left(w_i\right)\nonumber \\
&\propto& \frac{1}{\eta\left(\bx_i;\bPsi\right)} f_{\text{G}}\left(w_i;\frac{d}{2}+1,\frac{\delta\left(\bx_i;\bmu,\bSigma\right)}{2}\right)\indic_{\left(1-\theta,1\right)}\left(w_i\right), 
\label{eq:posterior}
\end{eqnarray}  
where 
$$
f_{\text{G}}\left(w;\alpha,\beta\right)=\frac{\beta^{\alpha}}{\Gamma\left(\alpha\right)}w^{\alpha-1}\exp\left(-\beta w\right)
$$ 
denotes the pdf of a gamma distribution with parameters $\alpha>0$ and $\beta>0$ and 
$$
\eta\left(\bx_i;\bPsi\right) = \frac{1}{\Gamma\left(\frac{d}{2}+1\right)}\left[\Gamma\left(\frac{d}{2}+1,(1-\theta)\frac{\delta\left(\bx_i;\bmu,\bSigma\right)}{2}\right)-
\Gamma\left(\frac{d}{2}+1,\frac{\delta\left(\bx_i;\bmu,\bSigma\right)}{2}\right)\right].
$$
This means that $W_i|\bX_i=\bx_i$ has a doubly-truncated gamma distribution \citep{Coff:Mull:Prop:2000}, on the interval $\left(1-\theta,1\right)$, with parameters $d/2+1$ and $\delta\left(\bx_i;\bmu,\bSigma\right)/2$, whose pdf is given in \eqref{eq:posterior}; in symbols   
\begin{equation}
W_i|\bX_i=\bx_i \sim \mathcal{DTG}_{\left(1-\theta,1\right)}\left(\frac{d}{2}+1,\frac{\delta\left(\bx_i;\bmu,\bSigma\right)}{2}\right).
\label{eq:Truncated Gamma}
\end{equation}

Now, the functions $m\left(W_i\right)$ arising in \eqref{eq:MTIN complete-data log-likelihood mu and Sigma} and \eqref{eq:MTIN complete-data log-likelihood theta} are $m_1\left(w\right)=w$, $m_2\left(w\right)=\log\left(w\right)$ and $m_3\left(w\right)=\log\left[\indic_{\left(1-\theta,1\right)}\left(w\right)\right]$. 
However, there is no need to compute the expectation of $m_2$ since the term $\log\left(W_i\right)$ is not related with the parameters; moreover, there is no need to compute the expectation of $m_3$ because we do not use $Q_2\left(\theta|\bPsi^{(r)}\right)$ to update $\theta$.  
%for the motivation we will give in Section~\ref{subsec:M-step}.
Thanks to \eqref{eq:Truncated Gamma} we have that
\begin{equation}
E_{\bPsi^{(r)}}\left(W_i|\bX_i=\bx_i\right) = w_i^{(r)},
\label{eq:wi expectation}
\end{equation}
where 
\begin{align}
w_i^{(r)} = &\frac{2\:  \left[\Gamma\left(\frac{d}{2}+2,\left(1-\theta^{(r)}\right)\frac{\delta\left(\bx_i;\bmu^{(r)},\bSigma^{(r)}\right)}{2}\right)-
\Gamma\left(\frac{d}{2}+2,\frac{\delta\left(\bx_i;\bmu^{(r)},\bSigma^{(r)}\right)}{2}\right)\right] }{\Gamma\left(\frac{d}{2}+1\right)\eta\left(\bx_i;\bPsi^{(r)}\right)\delta\left(\bx_i;\bmu^{(r)},\bSigma^{(r)}\right)} \nonumber\\
%&\times\left[\Gamma\left(\frac{d}{2}+2,\left(1-\theta^{(r)}\right)\frac{\delta\left(\bx_i;\bmu^{(r)},\bSigma^{(r)}\right)}{2}\right)-
%\Gamma\left(\frac{d}{2}+2,\frac{\delta\left(\bx_i;\bmu^{(r)},\bSigma^{(r)}\right)}{2}\right)\right] \nonumber \\
=& \frac{2}{\delta\left(\bx_i;\bmu^{(r)},\bSigma^{(r)}\right)} 
\frac{
\left[\Gamma\left(\frac{d}{2}+2,\left(1-\theta^{(r)}\right)\frac{\delta\left(\bx_i;\bmu^{(r)},\bSigma^{(r)}\right)}{2}\right)-
\Gamma\left(\frac{d}{2}+2,\frac{\delta\left(\bx_i;\bmu^{(r)},\bSigma^{(r)}\right)}{2}\right)\right]
}{
\left[\Gamma\left(\frac{d}{2}+1,\left(1-\theta^{(r)}\right)\frac{\delta\left(\bx_i;\bmu^{(r)},\bSigma^{(r)}\right)}{2}\right)-
\Gamma\left(\frac{d}{2}+1,\frac{\delta\left(\bx_i;\bmu^{(r)},\bSigma^{(r)}\right)}{2}\right)\right]
}.
\label{eq:E step 1}
\end{align}
Then, by substituting $w_i$ with $w_i^{(r)}$ in the last term on the right-hand side of \eqref{eq:MTIN complete-data log-likelihood mu and Sigma}, we obtain 
$$
Q_1\left(\bmu,\bSigma|\bPsi^{(r)}\right) = - \frac{n}{2} \log\left|\bSigma\right| - \frac{1}{2} \sum_{i=1}^n w_i^{(r)} \delta\left(\bx_i;\bmu,\bSigma\right),
$$
where the terms constant with respect to $\bmu$ and $\bSigma$ are dropped.

%%%%%%%%%%%%%%%
%% CM-step 1 %%
%%%%%%%%%%%%%%%

\subsubsection{CM-step 1}
\label{subsec:CM-step 1}

The first CM-step, at the same iteration, requires the calculation of $\bPsi_1^{\left(r+1\right)}$ as the value of $\bPsi_1$ that maximizes $Q_1\left(\bmu,\bSigma|\bPsi^{(r)}\right)$, with respect to $\bmu$ and $\bSigma$, with $\Psi_2$ fixed at $\Psi_2^{(r)}$.
%According to the right-hand side of \eqref{eq:MTIN complete-data log-likelihood}, $Q_1\left(\bmu,\bSigma|\bPsi^{(r)}\right)$ and $Q_2\left(\theta|\bPsi^{(r)}\right)$ can be maximized separately with respect to the parameters they involve.
Theis maximization is easily implemented if we note that $Q_1$ is a weighted log-likelihood, with weights $w_1^{(r)},\ldots,w_n^{(r)}$, of $n$ independent observations $\bx_1,\ldots,\bx_n$ from $\mathcal{N}_d\left(\bmu,\bSigma/w_1^{(r)}\right),\ldots,\mathcal{N}_d\left(\bmu,\bSigma/w_n^{(r)}\right)$, respectively.
So, the updates for $\bmu$ and $\bSigma$ are 
%This yields the following updates for     
%So, updating $\bmu$ and $\bSigma$ is equivalent to compute the following weighted sample mean and sample covariance matrix
\begin{equation}
	\bmu^{(r+1)} = \displaystyle\frac{1}{\displaystyle\sum_{i=1}^nw_i^{(r)}}\sum_{i=1}^n w_i^{(r)}\bx_i 
	\label{eq:M-step mu} 
\end{equation}
and
\begin{equation}
	\bSigma^{(r+1)} = \frac{1}{n}\sum_{i=1}^n w_i^{(r)}\left(\bx_i-\bmu^{\left(r+1\right)}\right)\left(\bx_i-\bmu^{\left(r+1\right)}\right)'. 
	\label{eq:M-step Sigma}
\end{equation}

%%%%%%%%%%%%%%%
%% CM-step 2 %%
%%%%%%%%%%%%%%%

\subsubsection{CM-step 2}
\label{subsec:CM-step 2}

In the second CM-step, at the same iteration, $\Psi_2=\theta$ is chosen to maximize the observed-data log-likelihood function $l\left(\bPsi\right)$, as given by \eqref{eq:MTIN observed-data log-likelihood}, with $\bPsi_1$ fixed at $\bPsi_1^{(r+1)}$.
In detail, from \eqref{eq:MTIN observed-data log-likelihood} we have
\begin{align}
l\left(\left\{\bPsi_1^{(r+1)},\theta\right\}\right)=& n\log\left(2\right)-\frac{nd}{2}\log\left(\pi\right)-\frac{n}{2}\log \left|\bSigma^{(r+1)}\right| - n\log\left(\theta\right) - \left(\frac{d}{2}+1\right)\sum_{i=1}^n \log\left[\delta\left(\bx_i;\bmu^{(r+1)},\bSigma^{(r+1)}\right)\right] \nonumber\\
& + \sum_{i=1}^n \log\left[ \Gamma\left(\frac{d}{2}+1,(1-\theta) \frac{\delta\left(\bx_i;\bmu^{(r+1)},\bSigma^{(r+1)}\right)}{2} \right) - \Gamma\left(\frac{d}{2}+1,\frac{\delta\left(\bx_i;\bmu^{(r+1)},\bSigma^{(r+1)}\right)}{2} \right) \right]. 
\label{eq:ECME CM-step 2 theta}
\end{align} 
Thus, the second CM-step of the ECME algorithm chooses $\theta^{(r+1)}$ to maximize \eqref{eq:ECME CM-step 2 theta} with $\bmu=\bmu^{(r+1)}$ and $\bSigma=\bSigma^{(r+1)}$. 
This implies that $\theta^{(r+1)}$ is a solution of the equation
\begin{equation}
-\frac{n}{\theta} + \left(1-\theta\right)^{\frac{d}{2}}\sum_{i=1}^n 
\frac{\left[\frac{\delta\left(\bx_i;\bmu^{(r+1)},\bSigma^{(r+1)}\right)}{2}\right]^{\frac{d}{2}+1} \exp\left[-\frac{\left(1-\theta\right)\delta\left(\bx_i;\bmu^{(r+1)},\bSigma^{(r+1)}\right)}{2}\right]}
{\Gamma\left(\frac{d}{2}+1,(1-\theta) \frac{\delta\left(\bx_i;\bmu^{(r+1)},\bSigma^{(r+1)}\right)}{2} \right) - \Gamma\left(\frac{d}{2}+1,\frac{\delta\left(\bx_i;\bmu^{(r+1)},\bSigma^{(r+1)}\right)}{2} \right)}=0.
\label{eq:ECME dertheta}
\end{equation}
As a closed form solution for the root of \eqref{eq:ECME dertheta} is not analytically available, we use the {\tt uniroot()} function in the \textbf{stats} package of \textsf{R} to perform the numerical one-dimensional search of $\theta \in (0,1)$.
%As a closed form solution of ... is not analytically available, we use the {\tt uniroot()} function in the \textbf{stats} package of \textsf{R} to perform numerical search of the maximum of the previous expression.

% Studiare i pesi al variare di theta e al variare della Mahalanobis

% Krishnan & McLachan: 59 

%For v < cm, ML estimation of p is
%robust in the sense that observations with large Mahalanobis distances are downweighted.
%This can be clearly seen from the form of the equation (2.48) to be derived for the MLE of
%p. As u decreases, the degree of downweighting of outliers increases.

% McLachan & Peel: page 228

% Watanabe & Yamaguthi: page 59

%Both models downweight observations with large d2. However, the
%curve of the weights is quite different for the two modelswith the
%multivariate t-model producing relatively smoothly declining weights with
%increasing d2, and the contaminated normal model tending to concentrate
%on the low weights in a few outlying observations.

\section{Existence of the ML estimates} %for $\bmu$ and $\bSigma$}
\label{sec:Existence of MLE}

In this section we demonstrate the existence of the ML estimates for $\bmu$, $\bSigma$, and $\theta$.
With this aim, we consider the solutions from the ECME algorithm of Section~\ref{subsec:Alternative ECME algorithm}.
In line with the CM-steps of the algorithm, we first demonstrate the existence of $\bmu$ and $\bSigma$, given $\theta$ (Theorem~\ref{theo:existence MLE of mu and Sigma}), and then the existence of $\theta$ given $\bmu$ and $\bSigma$ (Theorem~\ref{theo:existence MLE of theta}).

\begin{theorem}\label{theo:existence MLE of mu and Sigma}
Given $\theta$ and a sample $\bx_1,\ldots,\bx_n$, if $n>d\left(\frac{d}{2}+1\right)$, then there exists $(\widehat{\bmu},\widehat{\bSigma})\in\real^d \times \mathcal{M}^+_{d\times d}$, being $\mathcal{M}^+_{d\times d}$ the set of positive-definite symmetric $d\times d$ matrices, such that
\begin{equation}
l\left(\widehat{\bmu},\widehat{\bSigma},\theta\right) \geq l\left(\bmu, \bSigma ,\theta\right), \quad \text{for every} \ (\bmu , \bSigma) \in \real^d \times \mathcal{M}^+_{d\times d}.
%\prod_{i=1}^n f_{\text{TIN}}\left(\bx_i;\widehat{\bmu},\widehat{\bSigma},\theta^*\right)  \geq \prod_{i=1}^n f_{\text{TIN}}\left(\bx_i; \bmu , \bSigma ,\theta^*\right) , \qquad \text{for every} \ (\bmu , \bSigma) \in \mathcal{S}.
\label{eq:existence}
\end{equation}
\end{theorem}
\begin{proof}
The pdf of $\bX\sim \mathcal{TIN}_d\left(\bmu,\bSigma,\theta\right)$ can be written as
$$
f_{\text{TIN}}\left(\bx;\bmu,\bSigma,\theta\right)  =
 \left|\bSigma \right|^{-\frac{1}{2}} g\left[\delta\left(\bx;\bmu,\bSigma\right)\right],
$$
where
$$
g(t)= \frac{2\pi^{-\frac{d}{2}}}{\theta t^{\left(\frac{d}{2}+1\right)}}    
\left[ \Gamma\left(\frac{d}{2}+1,(1-\theta) \frac{t}{2} \right) - \Gamma\left(\frac{d}{2}+1,\frac{t}{2} \right) \right]. 
$$
To prove the theorem we use Proposition~2.4 in \citet{Cues:Trim:2008}.
In particular, it is sufficient to verify that the following assumptions are valid:
\begin{enumerate}[start=1,label={(G\arabic*)}] %\bfseries (G\arabic*)
	\item\label{item:G1} there exists a strictly decreasing sequence $\left\{t_n\right\}$ which converges to zero, such that $g(t_n)<g(t_{n+1})$ for every $n$;
	\item\label{item:Gp} if $d>1$, then there exists $\gamma>d/2$ such that $\displaystyle\lim_{r\rightarrow \infty} r^\gamma g(r)=0$;
	\item\label{item:G2} $g(\cdot)$ is continuous on $\real^+$.
\end{enumerate}
While \ref{item:G2} is easily verifiable, \ref{item:G1} and \ref{item:Gp} need some attention.
To verify \ref{item:G1} we can consider the first derivative of $g(\cdot)$, i.e.
$$
g'(t)=-\frac{2\pi^{-\frac{d}{2}}}{\theta t^{\left(\frac{d}{2}+2\right)}}    \left[ \Gamma\left(\frac{d}{2}+2,(1-\theta) \frac{t}{2} \right) - \Gamma\left(\frac{d}{2}+2,\frac{t}{2} \right) \right].
$$
It is straightforward to verify that $g'(t)<0$ for all $t \in \real^+$; so, \ref{item:G1} is satisfied for any strictly decreasing sequence $\left\{t_n\right\}$ which converges to zero.

In order to verify \ref{item:Gp}, consider $\gamma=\frac{d}{2}+1>\frac{d}{2}$.
Then
%To verify \ref{item:Gp}, let $\gamma=\frac{d}{2}+1$ which leads to
\begin{eqnarray}
r^{\frac{d}{2}+1} g(r) &=& \frac{2\pi^{-\frac{d}{2}}}{\theta}    
\left[ \Gamma\left(\frac{d}{2}+1,(1-\theta) \frac{t}{2} \right) - \Gamma\left(\frac{d}{2}+1,\frac{t}{2} \right) \right] \nonumber \\
&=& \frac{2\pi^{-\frac{d}{2}}}{\theta} \left(\int_0^{\frac{r}{2}} \exp(-x) x^{\frac{d}{2}+1}dx-\int^{(1-\theta)\frac{r}{2}}_0 \exp(-x) x^{\frac{d}{2}+1}dx\right).
\label{eq:Gp}
\end{eqnarray}
When $r\rightarrow \infty$, both the integrals in \eqref{eq:Gp} tend to $\Gamma\left(\frac{d}{2}+2\right)$, so that $\lim_{r\rightarrow \infty} r^\gamma g(r)=0$.
Then, \ref{item:Gp} is verified for $\gamma=\frac{d}{2}+1$. 
\end{proof}

%%%%%%%%%%%%%%%
%% Theorem 2 %%
%%%%%%%%%%%%%%%
\begin{theorem}\label{theo:existence MLE of theta}
Given $\bmu$, $\bSigma$, and a sample $\bx_1,\ldots,\bx_n$, then there exists $\widehat{\theta} \in \left(0,1\right)$ such that 
\begin{equation}
l\left(\bmu,\bSigma,\widehat{\theta}\right) \geq l\left(\bmu,\bSigma,\theta\right), \quad \text{for every $\theta \in \left(0,1\right)$}.
\label{eq:existencetheta}
\end{equation}
\end{theorem}
\begin{proof}
According to \eqref{eq:withintegral}, and given $\bmu$ and $\bSigma$, the log-likelihood function results
\begin{equation}
l\left(\bmu,\bSigma,\theta\right) = l\left(\theta\right) = -\frac{nd}{2}\log\left(2\pi\right)-\frac{n}{2}\log \left|\bSigma\right|  + \sum_{i=1}^n \log \left\{ \frac{1}{\theta}  \int_{1-\theta}^1 w^{\frac{d}{2}} \exp\left[-\frac{w}{2}{\delta\left(\bx_i;\bmu,\bSigma\right)}\right] dw \right\}.
\label{eq:liktheta}
\end{equation}
%where $l\left(\theta\right)$ stands for $l\left(\bmu,\bSigma,\theta\right)$ since $\bmu$ and $\bSigma$ are fixed.
It is straightforward to realize that $l\left(\theta\right)$ is continuous in $(0,1)$.
%Note that $l\left(\theta\right)$ is continuous in $(0,1)$.
If the limits of $l\left(\theta\right)$, when $\theta$ tends to the boundaries $0$ and $1$ of the support, are both finite, then the log-likelihood function has a maximum and the theorem is proved.
When $\theta \rightarrow 1$ we have 
%The limit of $l\left(\theta\right)$ when $\theta$ tends to $1$ is easy to calculate and results to be
\begin{equation*}
\lim_{\theta \rightarrow 1} l\left(\theta\right) = -\frac{nd}{2}\log\left(2\pi\right)-\frac{n}{2}\log \left|\bSigma\right|  + \sum_{i=1}^n \log \left\{   \int_{0}^1 w^{\frac{d}{2}} \exp\left[-\frac{w}{2}{\delta\left(\bx_i;\bmu,\bSigma\right)}\right] dw \right\}.
%\label{eq:liklim1}
\end{equation*}
When $\theta \rightarrow 0$, 
%we obtain 
%The limit of $l\left(\theta\right)$ when $\theta$ tends to $0$ can be calculated by 
using the results in \eqref{eq:limit0}--\eqref{eq:limit2}, we obtain
%The limit is
\begin{equation}
\lim_{\theta \rightarrow 0} l\left(\theta\right) = -\frac{nd}{2}\log\left(2\pi\right)-\frac{n}{2}\log \left|\bSigma\right|  -\frac{1}{2} \sum_{i=1}^n  \delta\left(\bx_i;\bmu,\bSigma\right).
\label{eq:liklim0}
\end{equation} 
The limit in \eqref{eq:liklim0} is constant and corresponds to the log-likelihood function of a $d$-variate normal distribution computed in $\bmu$ and $\bSigma$.  
\end{proof}

\section{Some notes on robustness}
\label{sec:Some notes on robustness}

The MTIN model allows to obtain improved (in terms of robustness) ML estimates of $\bmu$ and $\bSigma$, with respect to those provided by the nested (reference) MN model, in the presence of mild outliers.
In detail, the influence of the observations $\bx_i$ is reduced (down-weighted) as the squared Mahalanobis distance $\delta\left(\bx_i;\bmu,\bSigma\right)=\delta_i$ increases.
This is in line with the $M$-estimation \citep{Maro:Robu:1976} which uses a decreasing weighting function $w\left(\delta_i\right) : \left(0,\infty\right) \rightarrow \left(0,\infty\right)$ to down-weight the observations $\bx_i$ with large $\delta_i$ values.
	To be more precise, according to \eqref{eq:M-step mu} and \eqref{eq:M-step Sigma}, $\bmu^{(r+1)}$ and $\bSigma^{(r+1)}$ can be respectively viewed, since $\theta$ is estimated from the data by ML, as an adaptively weighted sample mean and sample covariance matrix, in the sense used by \citet{Hogg:Adap:1974}, with weights $w_i^{(r)}$ given in \eqref{eq:E step 1}.
This approach, in addition to be a type of $M$-estimation, follows \citet{Box:Samp:1980} and \citet{Box:Tiao:Baye:2011} in embedding the reference MN model in a larger model with one or more parameters (here $\theta$) that afford protection against non-normality.

For each possible value of the dimension $d$, consider the weights $w_i^{(r)}$ in \eqref{eq:E step 1} as a bivariate (weighting) function $w$ of the squared Mahalanobis distance $\delta\geq 0$ and of the inflation parameter $\theta\in\left(0,1\right)$, i.e.
\begin{equation}
w\left(\delta,\theta;d\right) = 
\frac{2
\left[\Gamma\left(\frac{d}{2}+2,\left(1-\theta\right)\frac{\delta}{2}\right)-
\Gamma\left(\frac{d}{2}+2,\frac{\delta}{2}\right)\right]
}{\delta
\left[\Gamma\left(\frac{d}{2}+1,\left(1-\theta\right)\frac{\delta}{2}\right)-
\Gamma\left(\frac{d}{2}+1,\frac{\delta}{2}\right)\right]
}.
\label{eq:weight function}
\end{equation}
It is straightforward to show that the weighting function in \eqref{eq:weight function} is positive. 
An example of graphical representation of $w\left(\delta,\theta;d\right)$, in the case $d=3$, is provided in \figurename~\ref{fig:weight 3D}. %\todo[fancyline,size=\tiny,color=green!40]{Riusciamo a dire che la funzione è positiva e qual è il massimo valore che assume?}
\begin{figure}[!ht]
\centering
\resizebox{0.5\textwidth}{!}{
\includegraphics{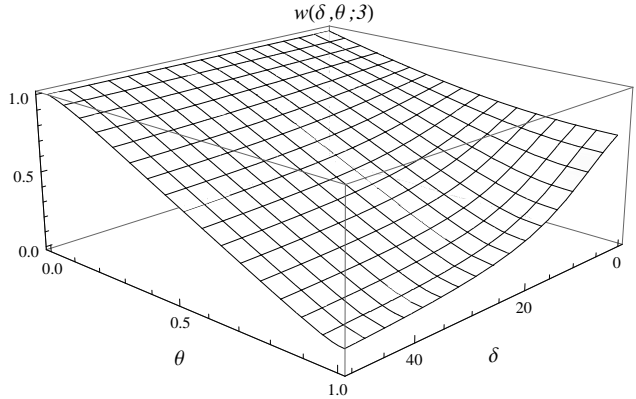} %[width=0.99\textwidth]
}
\caption{
Graphical representation of $w\left(\delta,\theta;d=3\right)$.
}
\label{fig:weight 3D}
\end{figure}
As concerns the limiting behavior of $w\left(\delta,\theta;d\right)$ as $\delta \rightarrow \infty$, by using the fundamental theorem of calculus and l'H$\hat{\mbox{o}}$spital's rule we obtain
% the limits of $w\left(\delta,\theta;d\right)$ as $\delta\rightarrow \infty$ and as $\delta\rightarrow 0$ are respectively given by
\begin{equation}
\lim_{\delta\rightarrow \infty} w\left(\delta,\theta;d\right) = 1-\theta.
\label{eq:delta inf}
\end{equation}
%and
%\begin{equation}
%\lim_{\delta\rightarrow 0} w\left(\delta,\theta;d\right) = \frac{\left(1-\theta\right)^{\frac{d}{2}+3}-1}{\left(1-\theta\right)^{\frac{d}{2}+2}-1}.
%\label{eq:delta 0}
%\end{equation}
From \eqref{eq:delta inf} we note that, when observations are very far from the bulk of the data ($\delta\rightarrow \infty$) their weight in the computation of $\bmu$ and $\bSigma$ depends by $\theta$.
If $\theta$ is close to zero (i.e.~if the MTIN approaches the MN distribution), then the weight is close to one; indeed, this is what we expect under the normal case.
%If $\theta \rightarrow 0$ (the MN case), then $w\left(\delta,\theta;d\right) \rightarrow 1$; this is what we expect under the normal case.
In the other cases, the weight linearly decreases as $\theta$ increases, i.e.~as the MTIN departs from the MN distribution.
As concerns the limiting behavior of $w\left(\delta,\theta;d\right)$ as $\theta\rightarrow 0$, by using the mathematics considered for the limit in \eqref{eq:delta inf}, we have
%Through the fundamental theorem of calculus and l'H$\hat{\mbox{o}}$spital's rule, the limits are
\begin{equation}
\lim_{\theta\rightarrow 0} w\left(\delta,\theta;d\right) = 1.
\label{eq:theta 0}
\end{equation}
From \eqref{eq:theta 0} we can observe that if $\theta\rightarrow 0$, then $w\left(\delta,\theta;d\right)\rightarrow 1$ regardless of $\delta$, and this happens because we are operationally working with a MN distribution (see Theorem~\ref{theo:normalcase}).

\section{Simulation studies}
\label{sec:Simulation studies}

In this section, we investigate various aspects related to our model through simulation studies.
While the simulation study of Section~\ref{subsec:Quality of estimates and computational times} is considered to evaluate parameter recovery of the proposed estimation procedures, in addition to the computational times required by these procedures, the simulation study of Section~\ref{subsec:Model selection: AIC versus BIC} aims to compare classical model selection criteria in selecting among a set of well-established elliptical distributions, with differing number of parameters, including the MTIN.
The whole analysis is conducted in \textsf{R} and the code for density evaluation, random number generation, and fitting for the MTIN distribution is available, in the form of an \textsf{R} package named \textbf{mtin}, from \url{http://docenti.unict.it/punzo/Rpackages.htm}.

\subsection{Parameter recovery and computational times}
\label{subsec:Quality of estimates and computational times}

In this section we compare two methods to estimate the parameters of the MTIN distribution, namely the method of moments (MM), discussed in \appendixname~\ref{sec:Method of moments}, and ML, discussed in Section~\ref{sec:ML}.
These methods are compared with respect to parameter recovery and to the computational time required.
% by our \textsf{R} codes to obtain the estimates.
We use three approaches to compute ML estimates: direct approach with the Nelder-Mead algorithm, direct approach with the BFGS algorithm, and ECME algorithm (cf.~Section~\ref{sec:ML}).
All the algorithms used to obtain ML estimates are initialized by the solution provided by the method of moments.

In this study we consider three experimental factors: the dimension ($d\in \left\{2,3,5\right\}$), the sample size ($n\in \left\{200,500,1000\right\}$), and the inflation parameter ($\theta\in \left\{0.6,0.7,0.8,0.9\right\}$). %of the MTIN used to generate the data 
The values of $\theta$ are chosen unbalanced on the right simply to have scenarios with more excess kurtosis (refer to the considerations made at the end of Section~\ref{sec:moments}).

For each combination of $n$, $d$ and $\theta$, we sample one hundred datasets from a MTIN distribution with a zero mean vector ($\bmu=\bzero$) and an identity scale matrix ($\bSigma=\bI$), for a total of $3\times 3\times 4\times 100=3600$ datasets. 
On each generated dataset, we fit the MTIN distribution with the four approaches cited above.
Computation is performed on a Windows 10 PC, with Intel i7-8550U CPU, 16.0 GB RAM, using \textsf{R} 64-bit, and the elapsed time (in seconds) is computed via the \texttt{system.time()} function of the \textbf{base} package.
Parallel computing, using 4 cores, is considered.

%%%%%%%%%%%%%%%%%%%%%%%%
%% Parameter Recovery %%
%%%%%%%%%%%%%%%%%%%%%%%%

We start evaluating parameter recovery by focusing on $\theta$; we limit the investigation to this parameter because, in analogy with other normal scale mixtures, the parameter(s) governing the tail-weight is (are) the most difficult to be estimated.
However, although not reported here, the ranking of the estimation methods we obtain with respect to $\theta$ are roughly preserved when we move to $\bmu$ and $\bSigma$. 

In \figurename~\ref{fig:Theta.Bias} we report the box-plots of the differences $(\hat{\theta}-\theta)$, for bias evaluation, while in \figurename~\ref{fig:Theta.MSE} we report the box-plots of the squared differences $(\hat{\theta}-\theta)^2$, for mean square error (MSE) evaluation.
Each of the nine plots in these figures refers to a particular pair $\left(n,d\right)$ and shows $4 \times 3 = 12$ box-plots each summarizing (for every $\theta$ and used method) the behavior of the considered differences with respect to the available 100 replications.
As expected, the differences under evaluation improve as $n$ increases.
Interestingly enough, the differences improves when either $d$ or $\theta$ increase.
Regardless of the scenario and difference considered, BFGS and ECME algorithms for ML estimation work comparably and represent the best approaches.
The worst approach is MM.
The very similar behavior between BFGS and ECME algorithms can be further corroborated by looking at \tablename~\ref{tab:log-likelihoods}, where log-likelihood values are averaged over the 100 replications of each triplet $\left(n,d,\theta\right)$. 
Here we can note how these algorithms return exactly the same average results, which are always slightly better than those provided by the Nelder-Mead algorithm, and this is true regardless of the triplet $\left(n,d,\theta\right)$ considered.
To be more precise, looking at the whole set of obtained results (not reported here for the sake of space), BFGS and ECME algorithms provide, on each fitted model, the same maximized log-likelihood value.
% the same simulation results not reported here for the sake of space, show that the two algorithm provide the same maximized log-likelihood on each fitted model.

\begin{figure}[!ht]
\centering
\subfigure[$d=2$ and $n=200$ \label{fig:biasd2n200}]
{\includegraphics[width=0.328\textwidth]{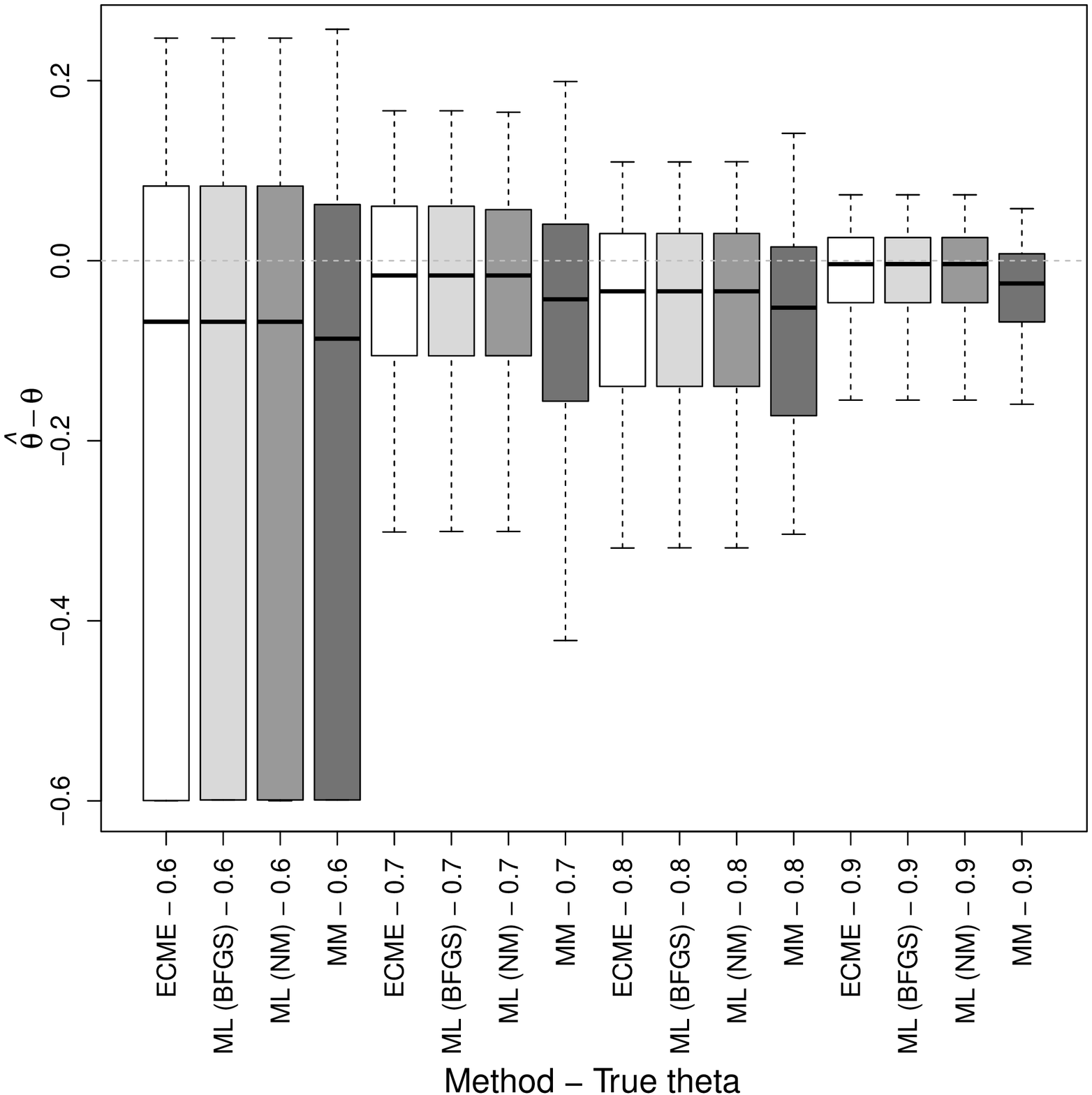}}
\subfigure[$d=2$ and $n=500$\label{fig:biasd2n500}]
{\includegraphics[width=0.328\textwidth]{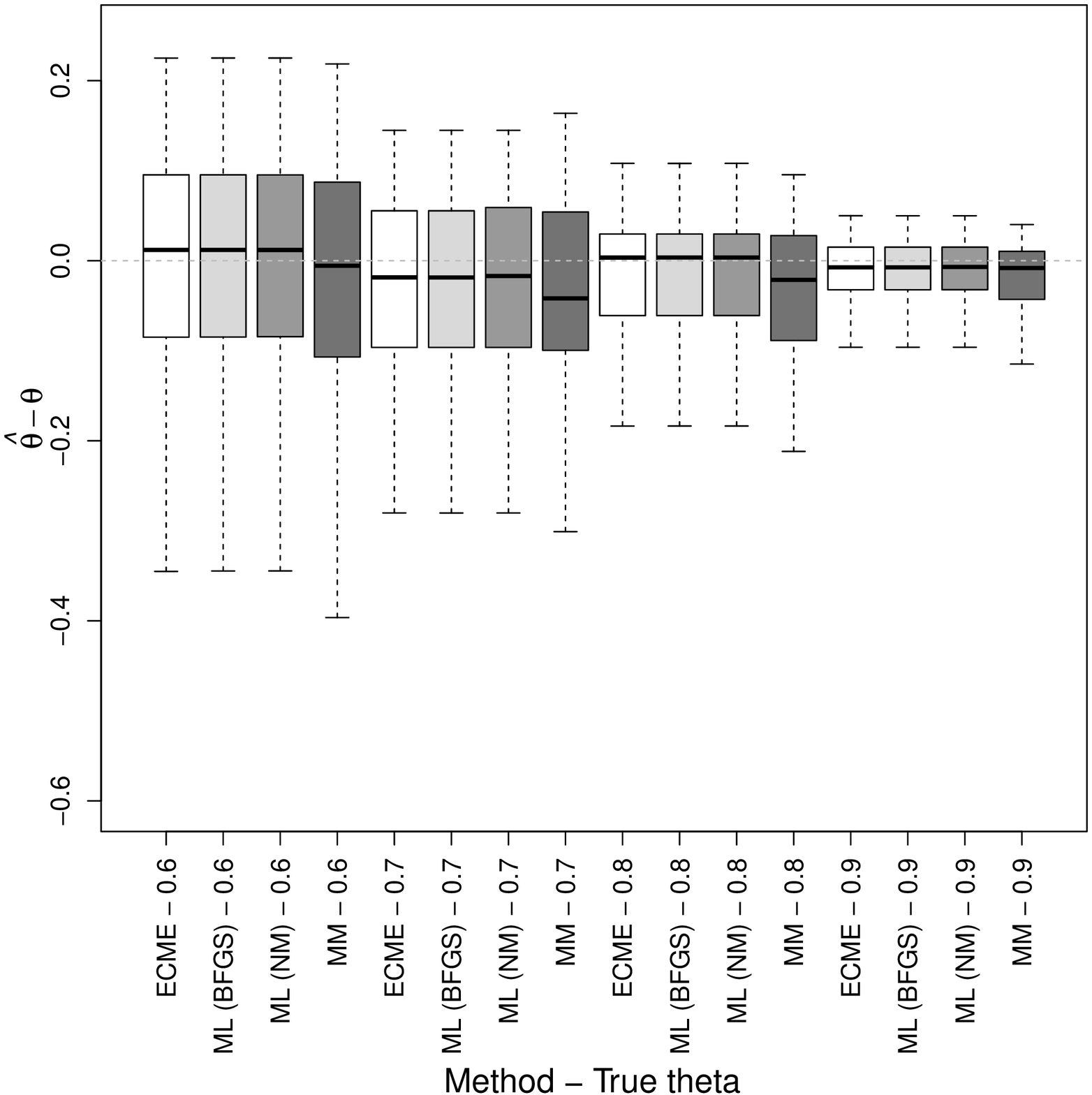}}
\subfigure[$d=2$ and $n=1000$\label{fig:biasd2n1000}]
{\includegraphics[width=0.328\textwidth]{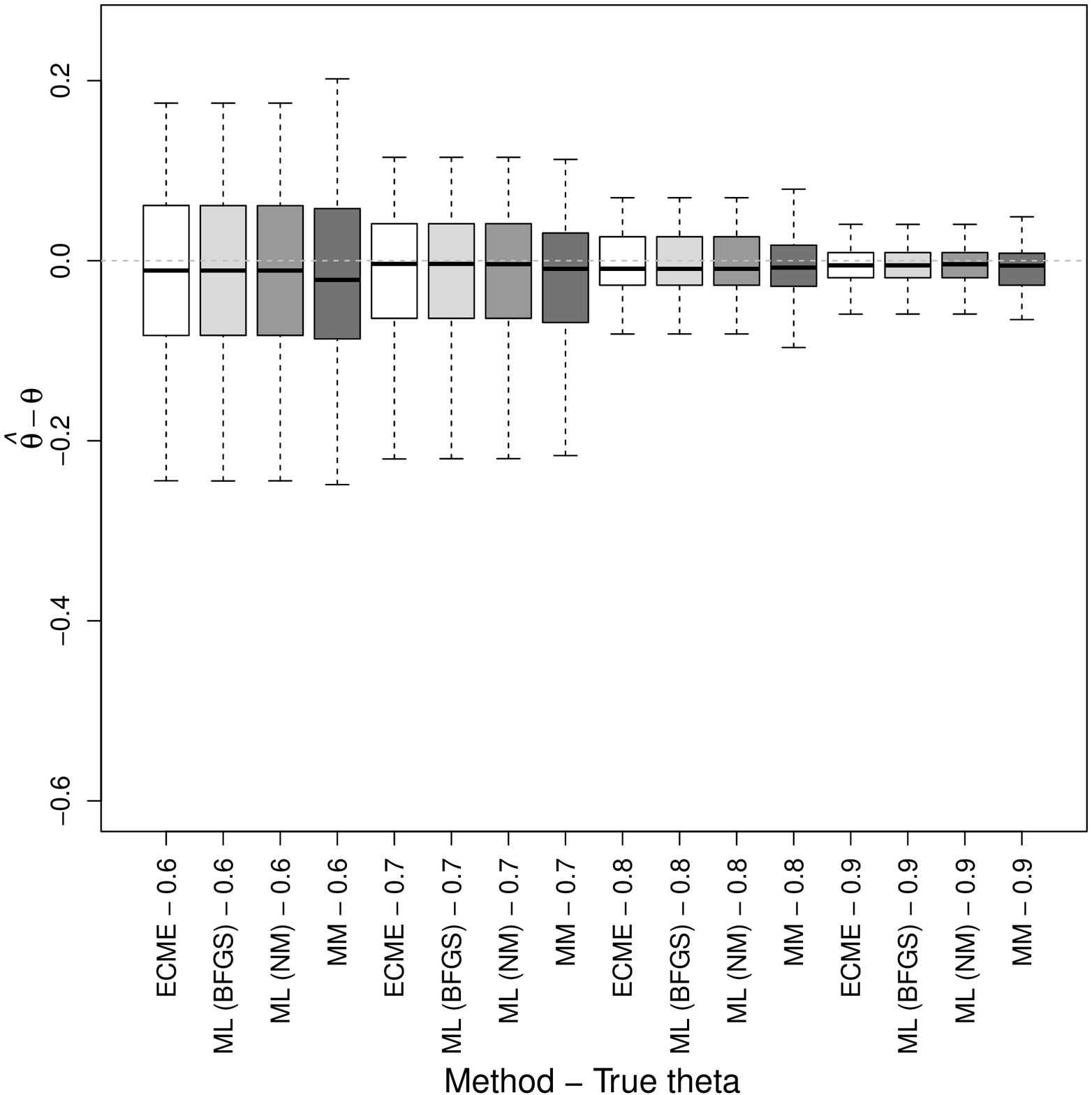}}

\subfigure[$d=3$ and $n=200$ \label{fig:biasd3n200}]
{\includegraphics[width=0.328\textwidth]{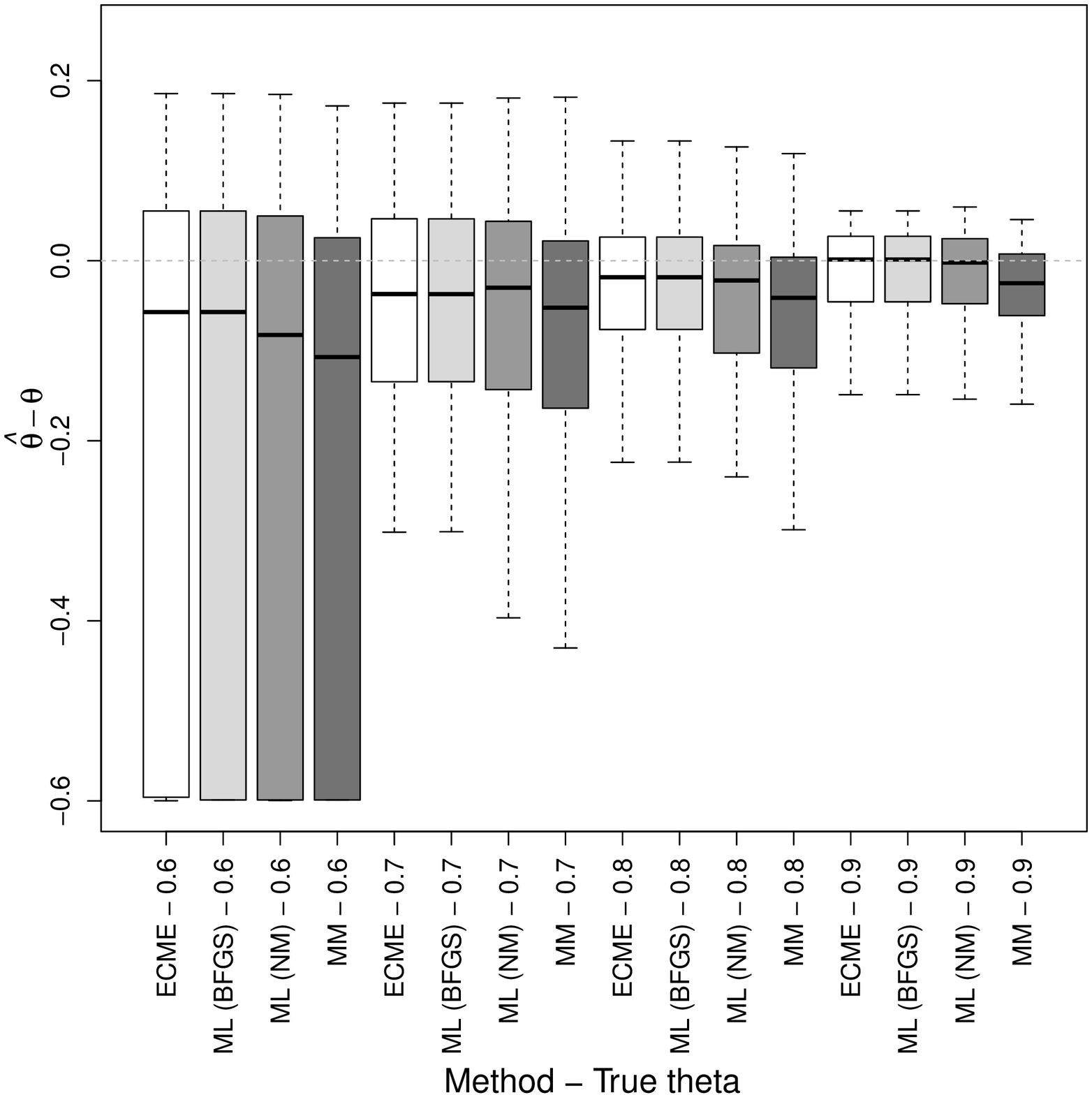}}
\subfigure[$d=3$ and $n=500$\label{fig:biasd3n500}]
{\includegraphics[width=0.328\textwidth]{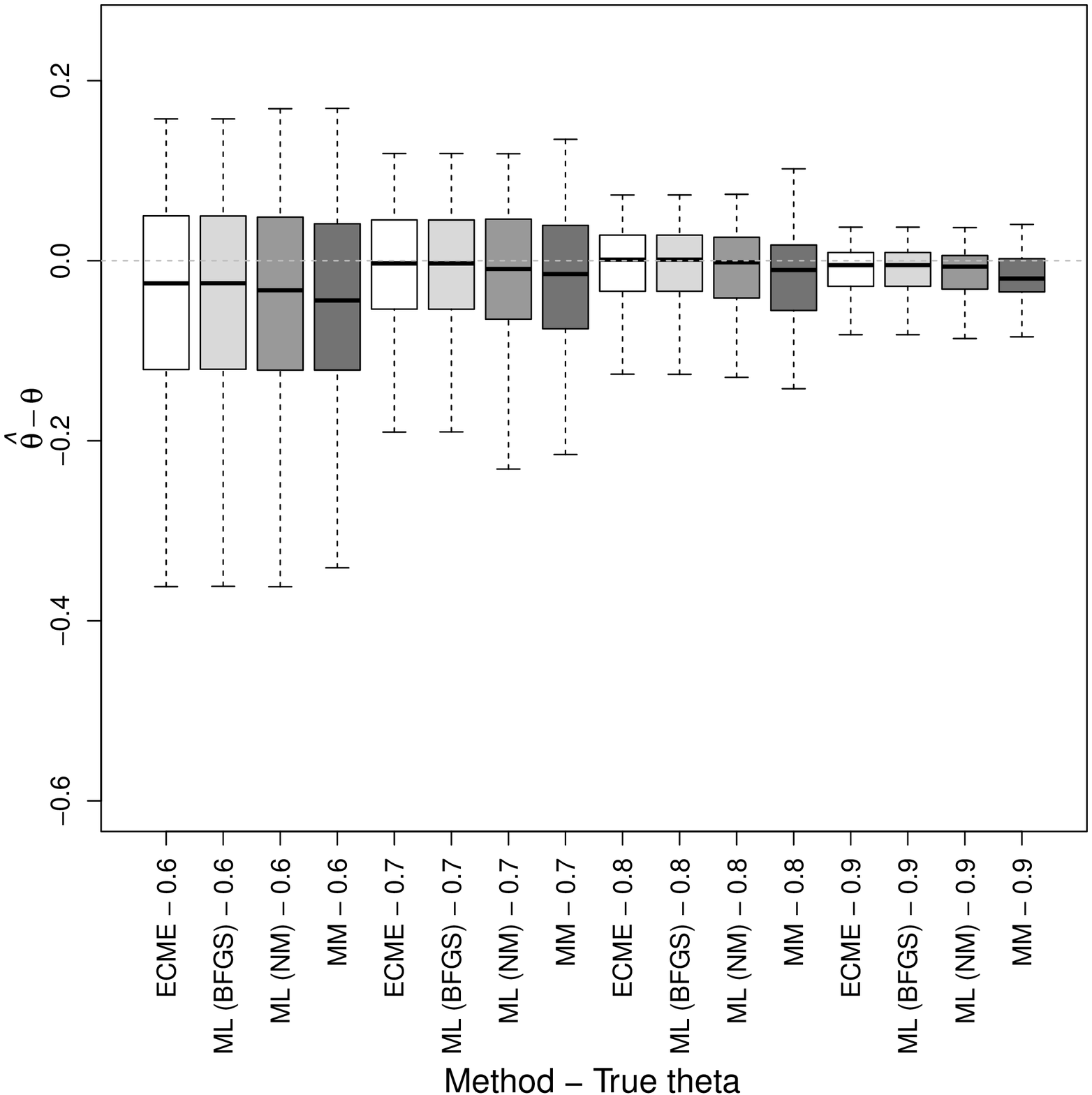}}
\subfigure[$d=3$ and $n=1000$\label{fig:biasd3n1000}]
{\includegraphics[width=0.328\textwidth]{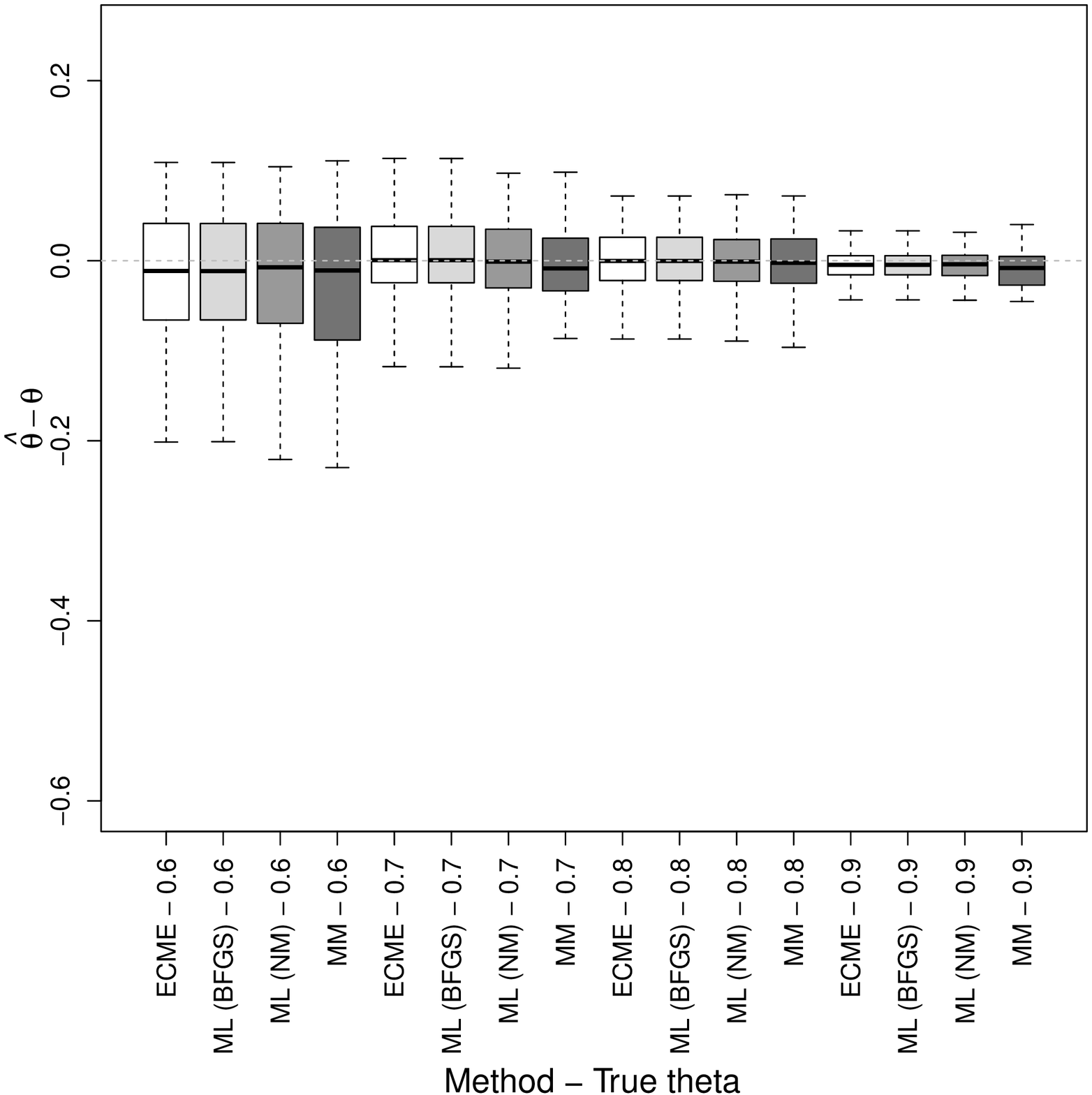}}

\subfigure[$d=5$ and $n=200$ \label{fig:biasd5n200}]
{\includegraphics[width=0.328\textwidth]{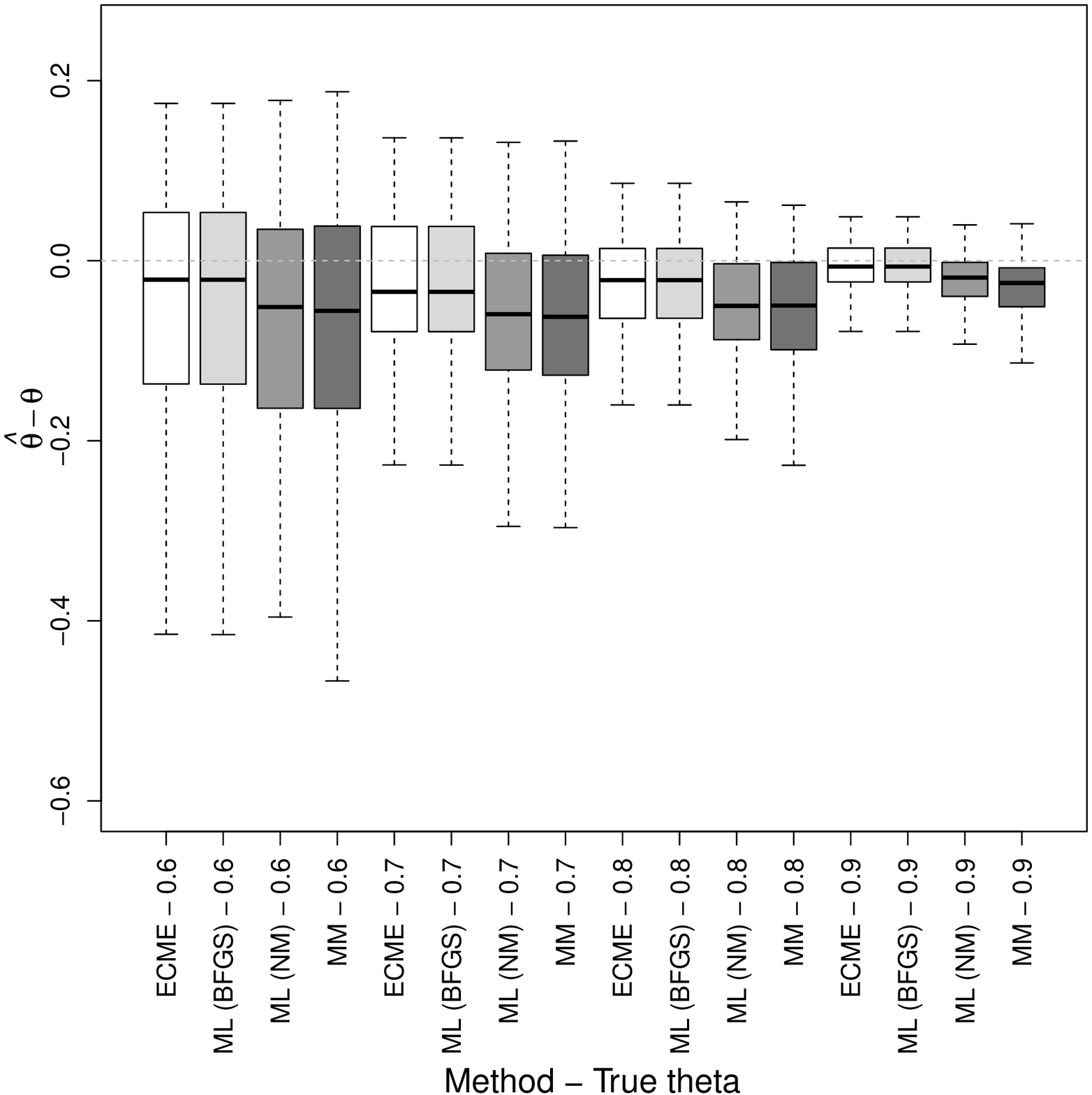}}
\subfigure[$d=5$ and $n=500$\label{fig:biasd5n500}]
{\includegraphics[width=0.328\textwidth]{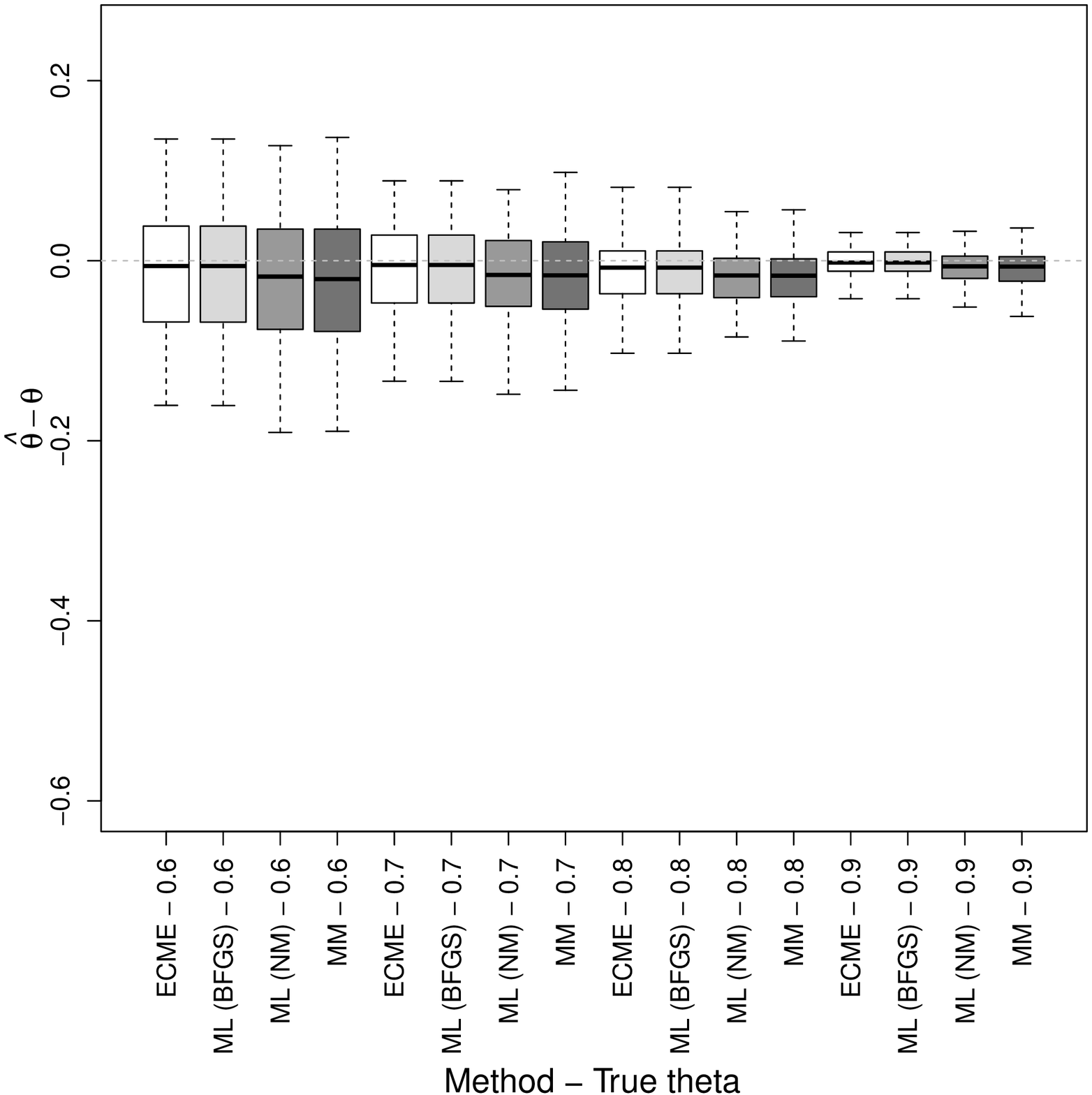}}
\subfigure[$d=5$ and $n=1000$\label{fig:biasd5n1000}]
{\includegraphics[width=0.328\textwidth]{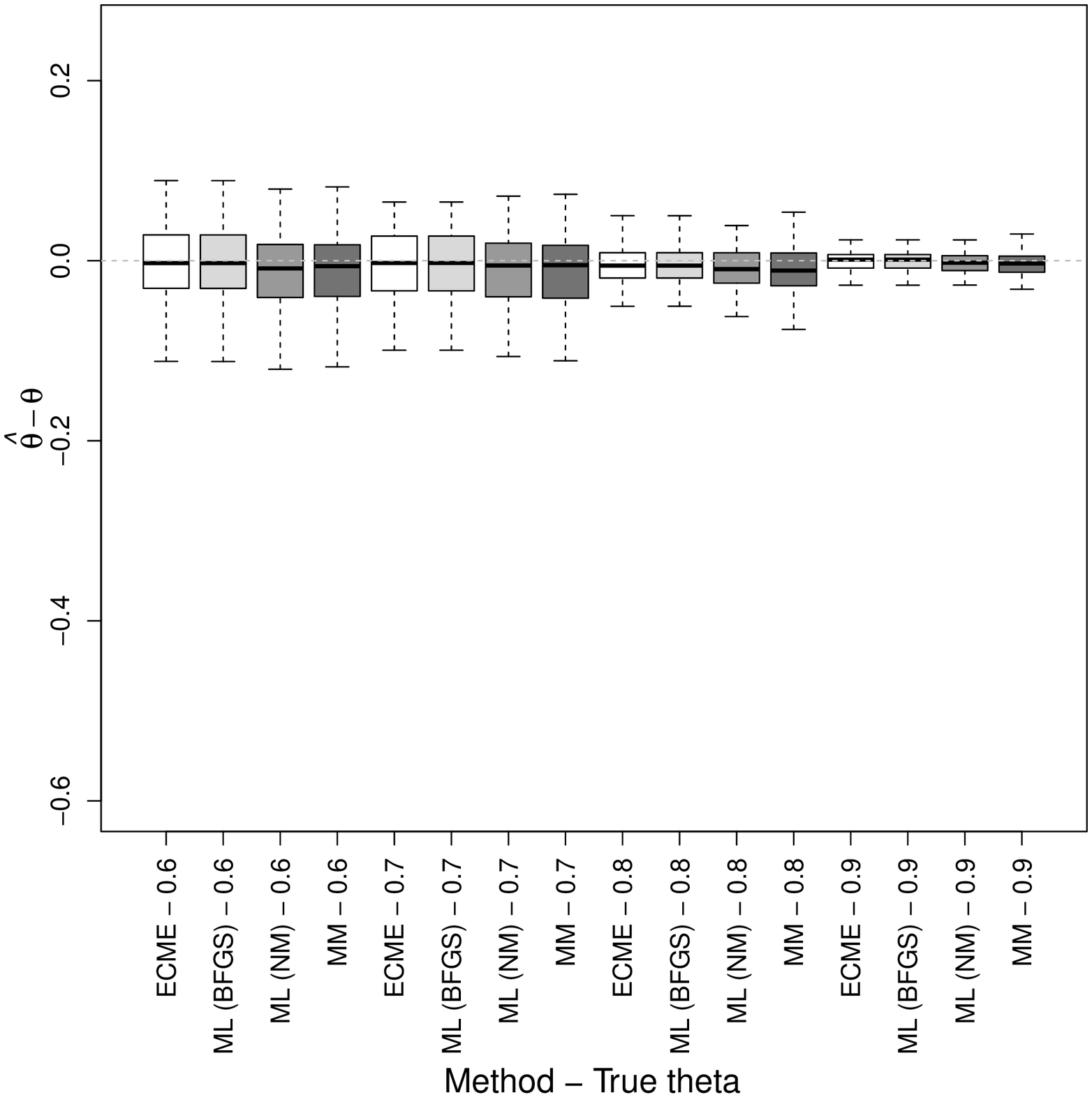}}
\caption{Box-plots of the differences $(\hat{\theta}-\theta)$ for each pair $\left(n,d\right)$.
The box-plot in each subfigure summarize the results over 100 replications and refers to a combination between $\theta$ and the estimation method considered.}
\label{fig:Theta.Bias}
\end{figure}
\begin{figure}[!ht]
\centering
\subfigure[$d=2$ and $n=200$ \label{fig:msed2n200}]
{\includegraphics[width=0.328\textwidth]{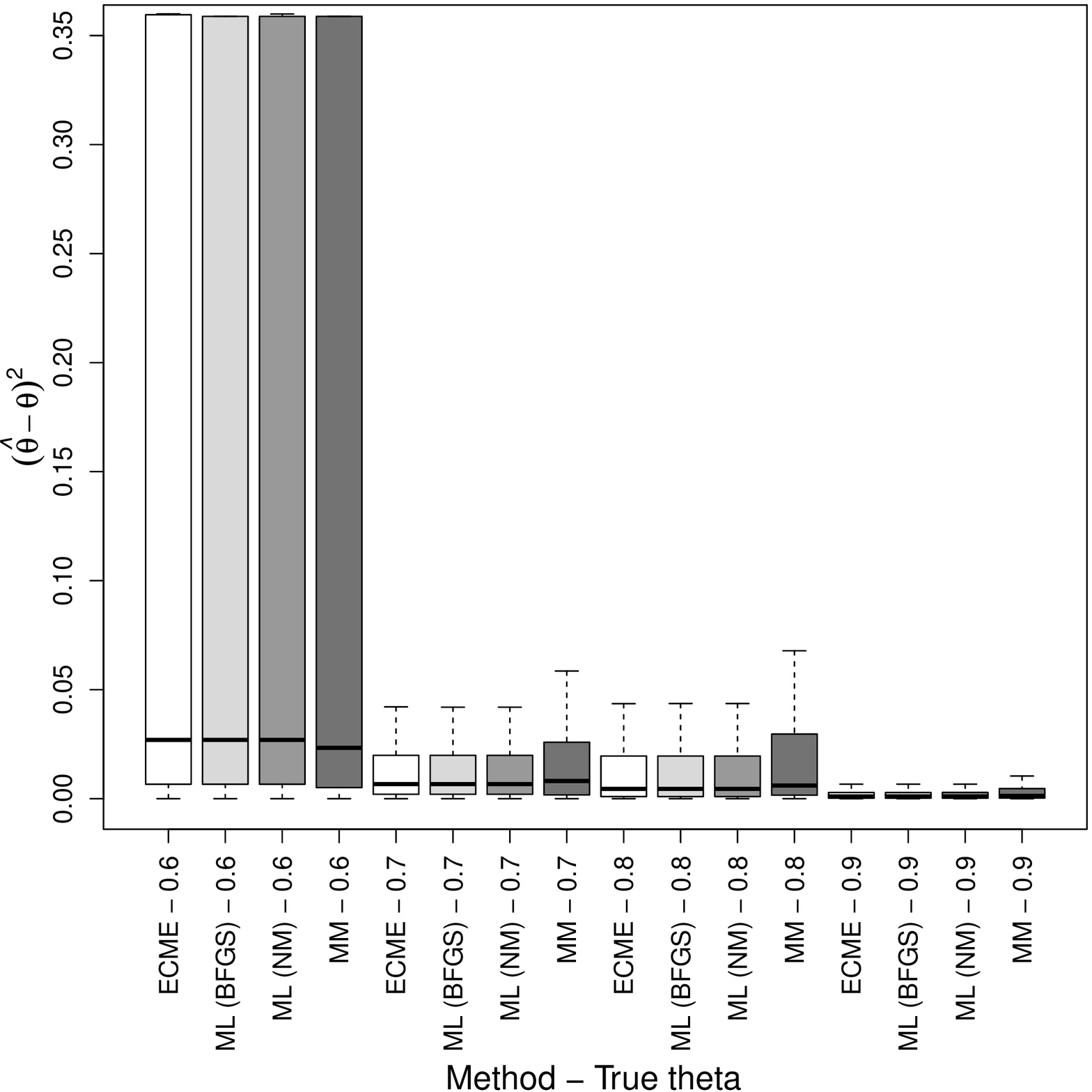}}
\subfigure[$d=2$ and $n=500$\label{fig:msed2n500}]
{\includegraphics[width=0.328\textwidth]{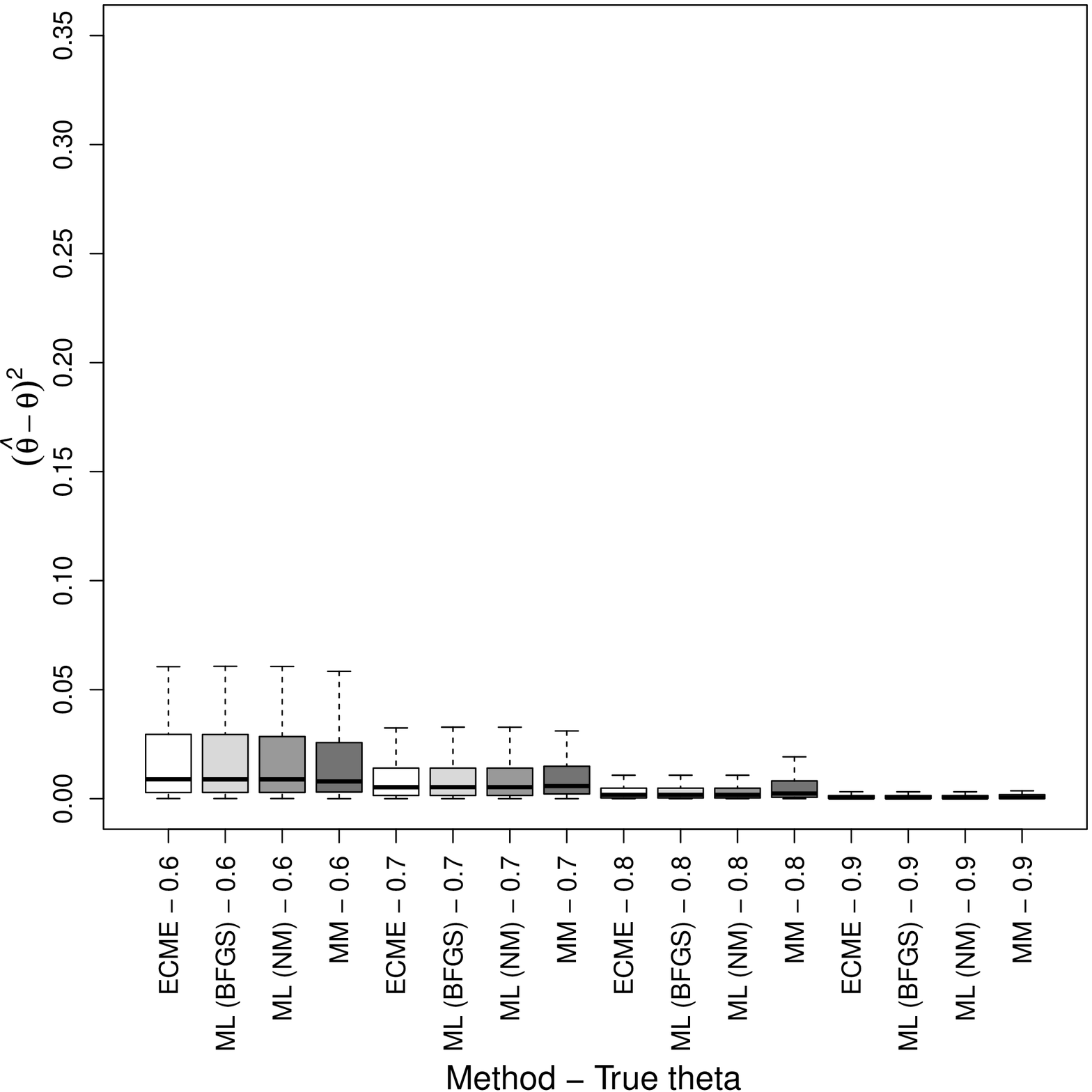}}
\subfigure[$d=2$ and $n=1000$\label{fig:msed2n1000}]
{\includegraphics[width=0.328\textwidth]{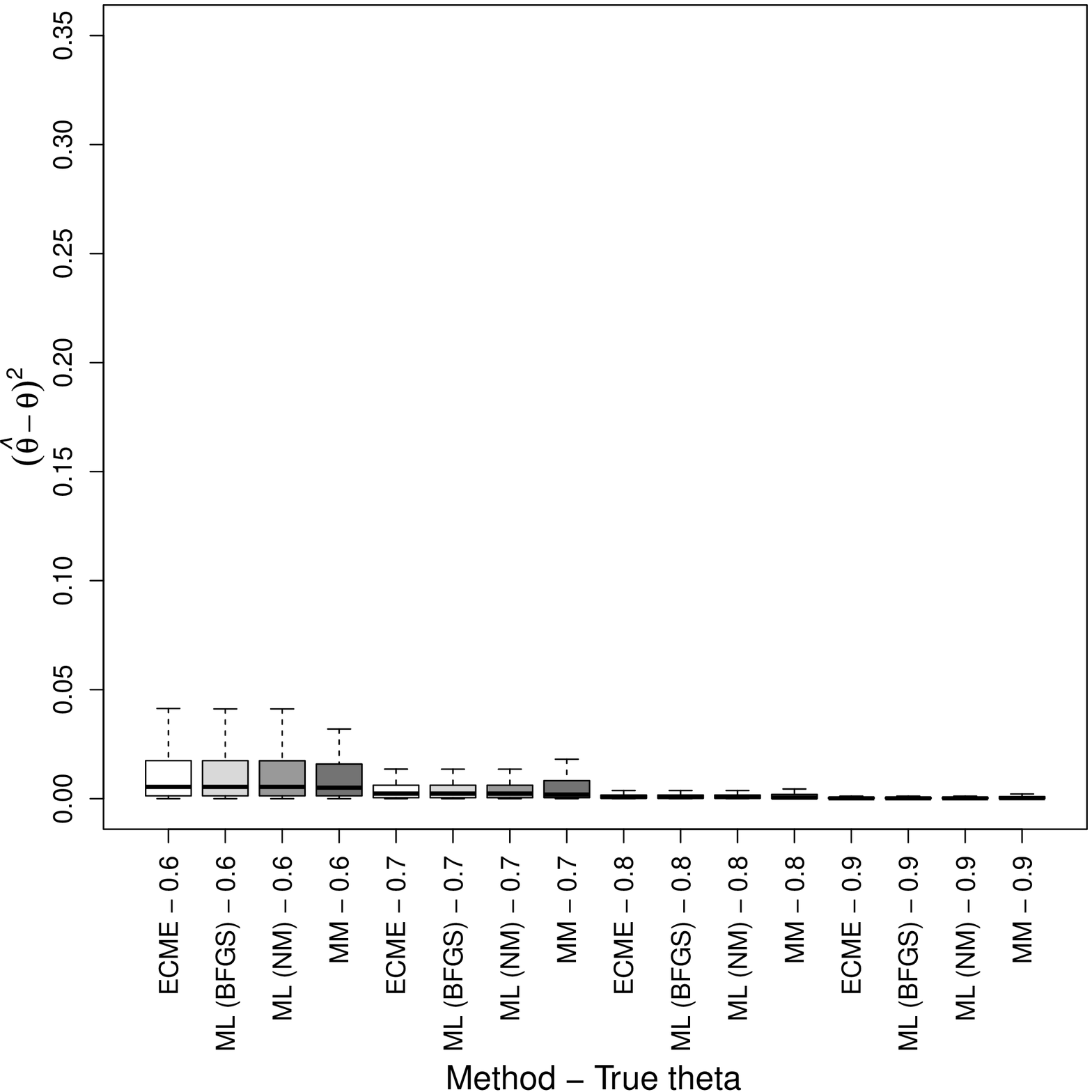}}

\subfigure[$d=3$ and $n=200$ \label{fig:msed3n200}]
{\includegraphics[width=0.328\textwidth]{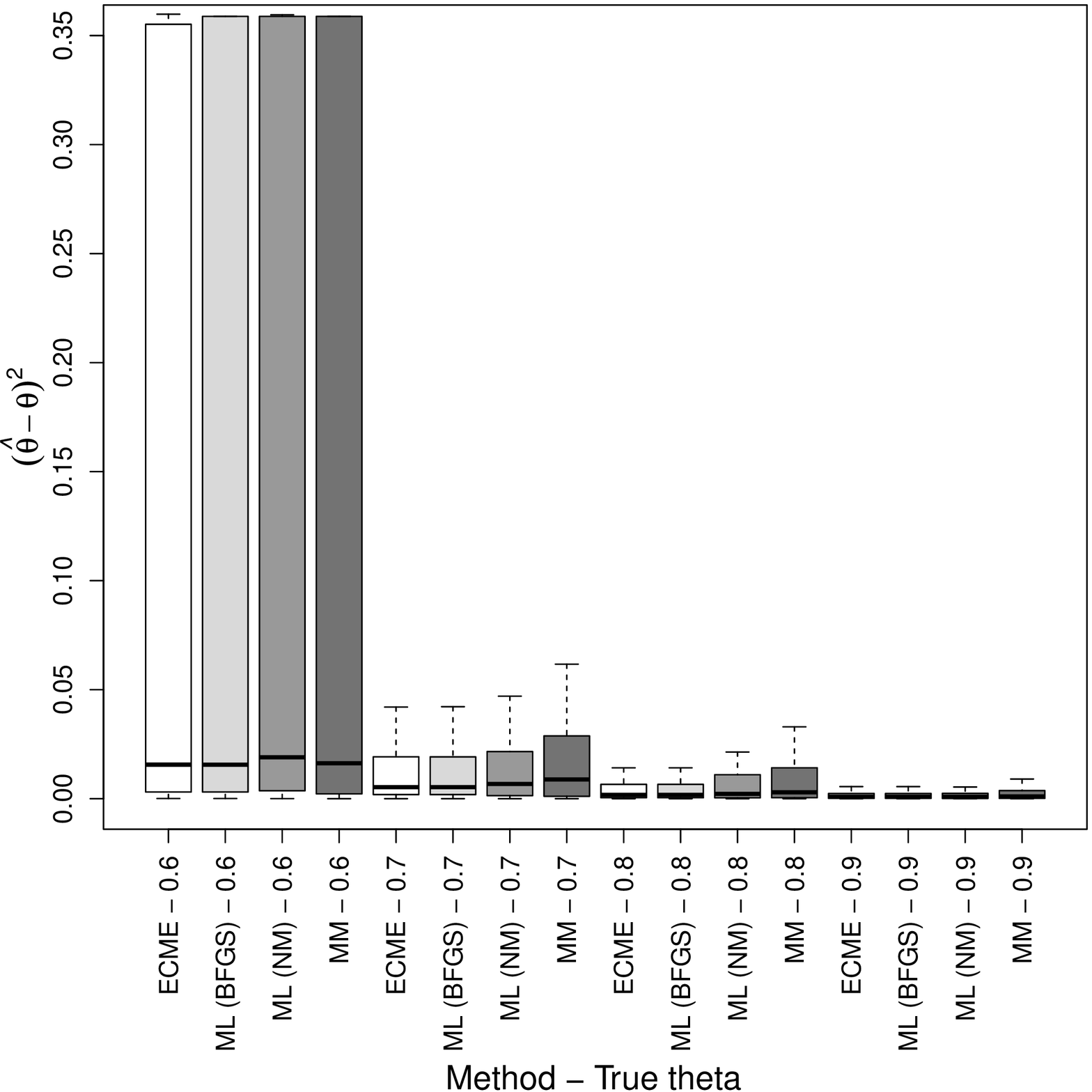}}
\subfigure[$d=3$ and $n=500$\label{fig:msed3n500}]
{\includegraphics[width=0.328\textwidth]{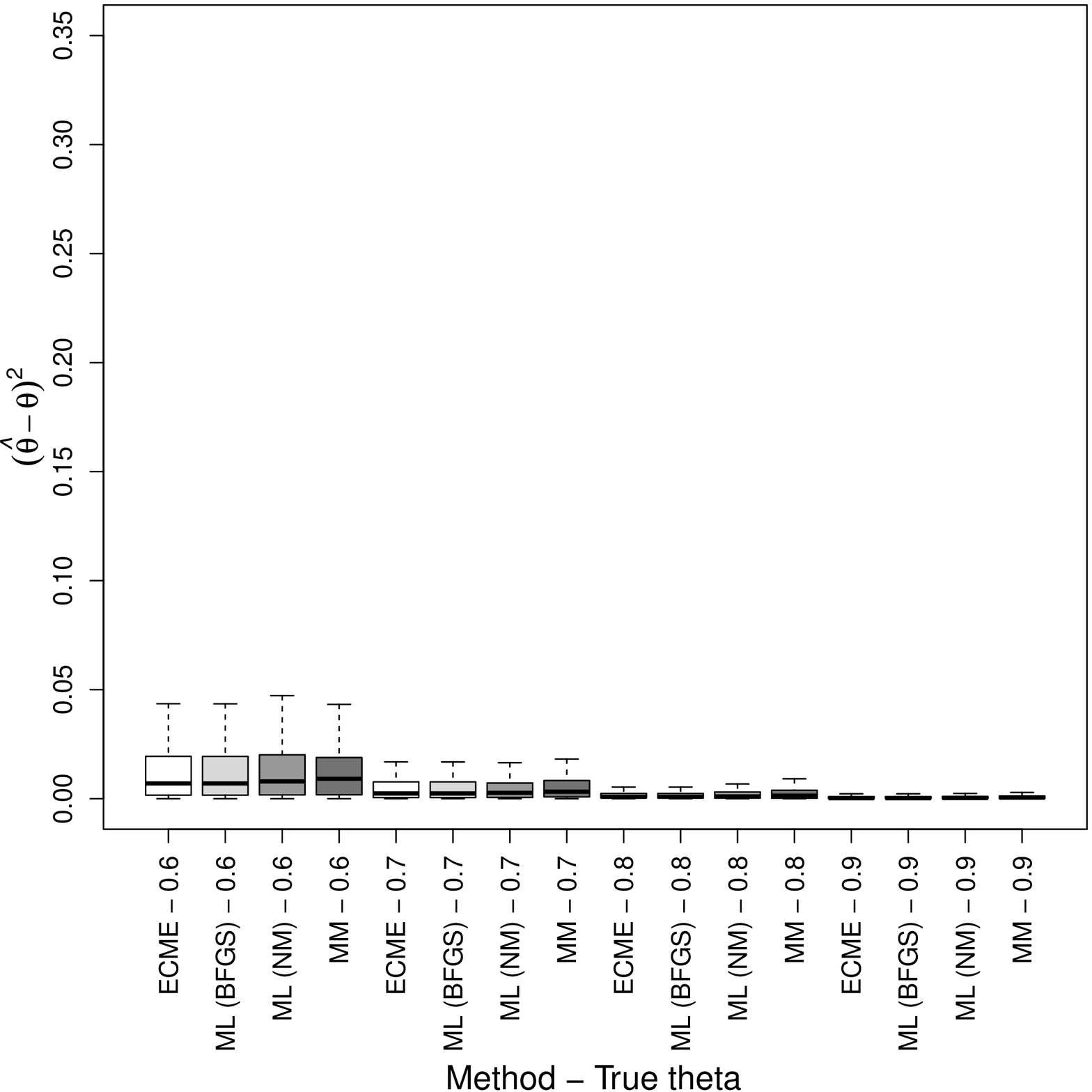}}
\subfigure[$d=3$ and $n=1000$\label{fig:msed3n1000}]
{\includegraphics[width=0.328\textwidth]{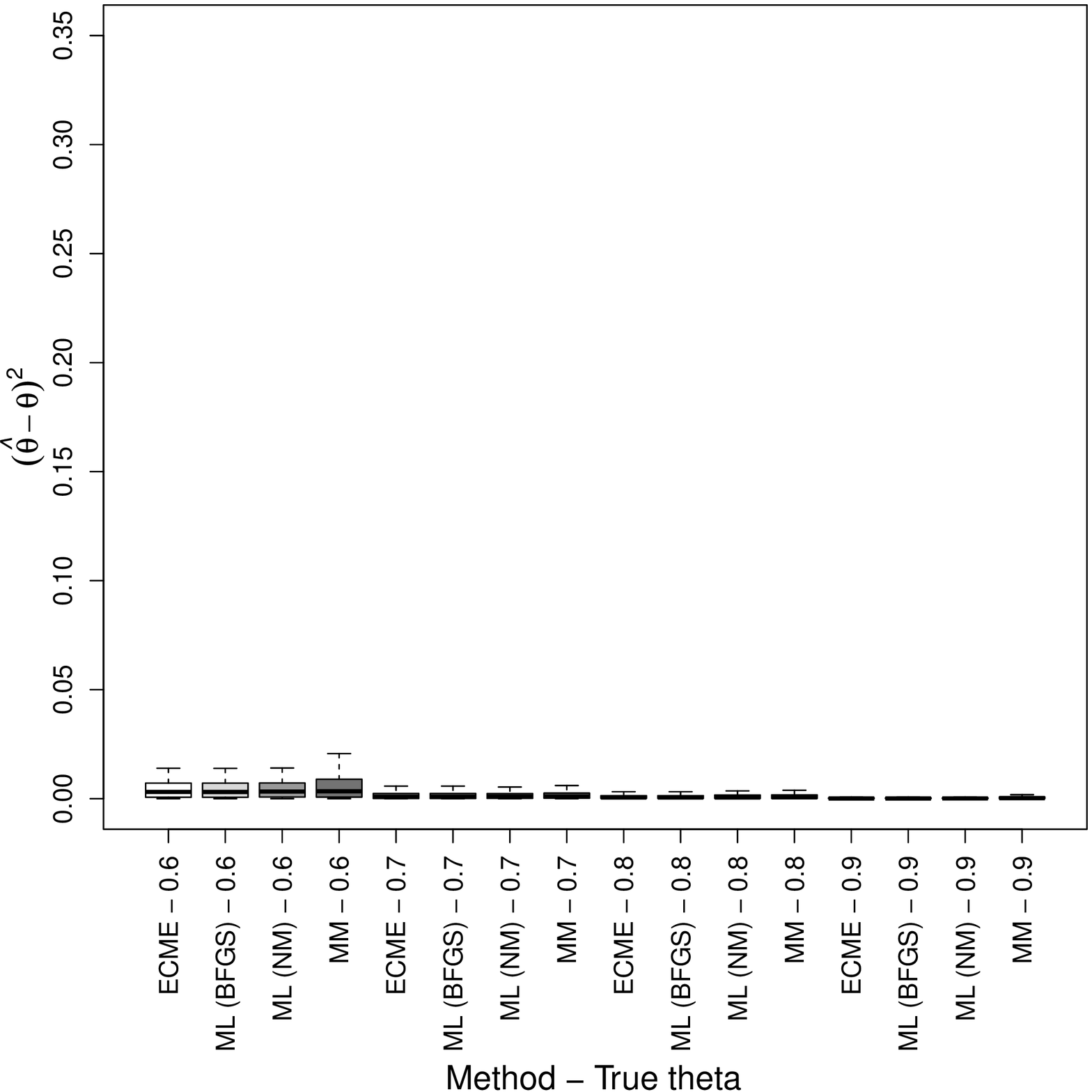}}

\subfigure[$d=5$ and $n=200$ \label{fig:msed5n200}]
{\includegraphics[width=0.328\textwidth]{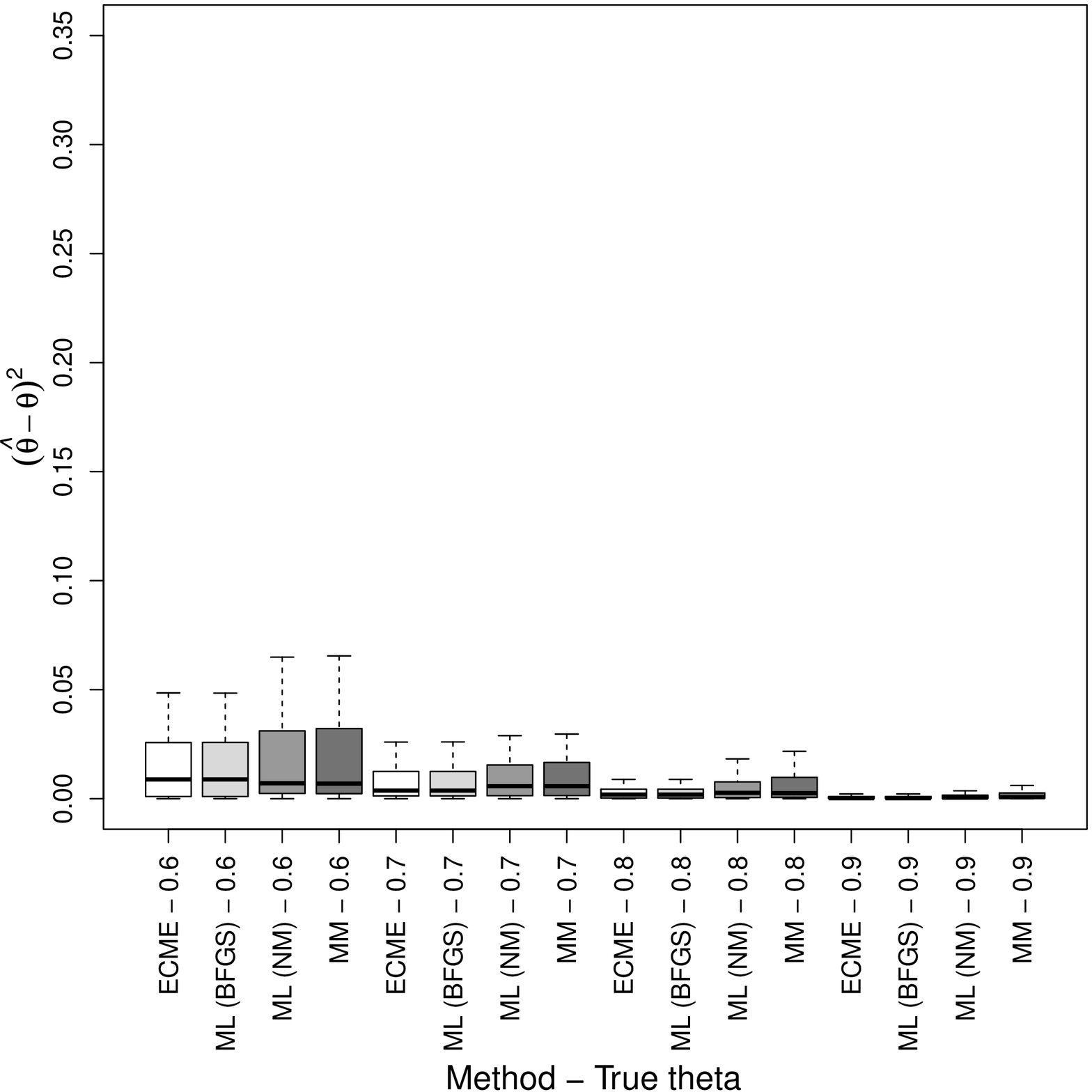}}
\subfigure[$d=5$ and $n=500$\label{fig:msed5n500}]
{\includegraphics[width=0.328\textwidth]{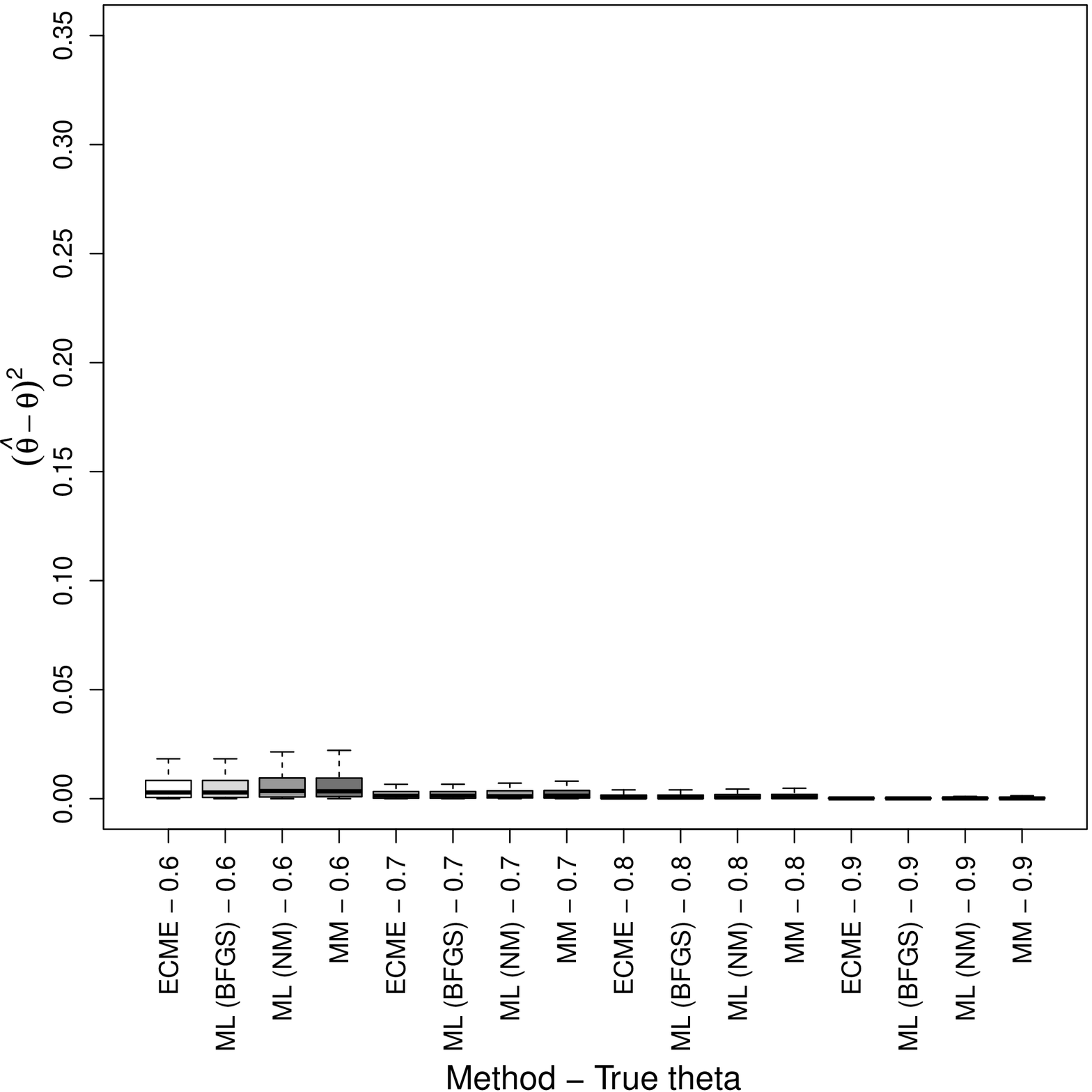}}
\subfigure[$d=5$ and $n=1000$\label{fig:msed5n1000}]
{\includegraphics[width=0.328\textwidth]{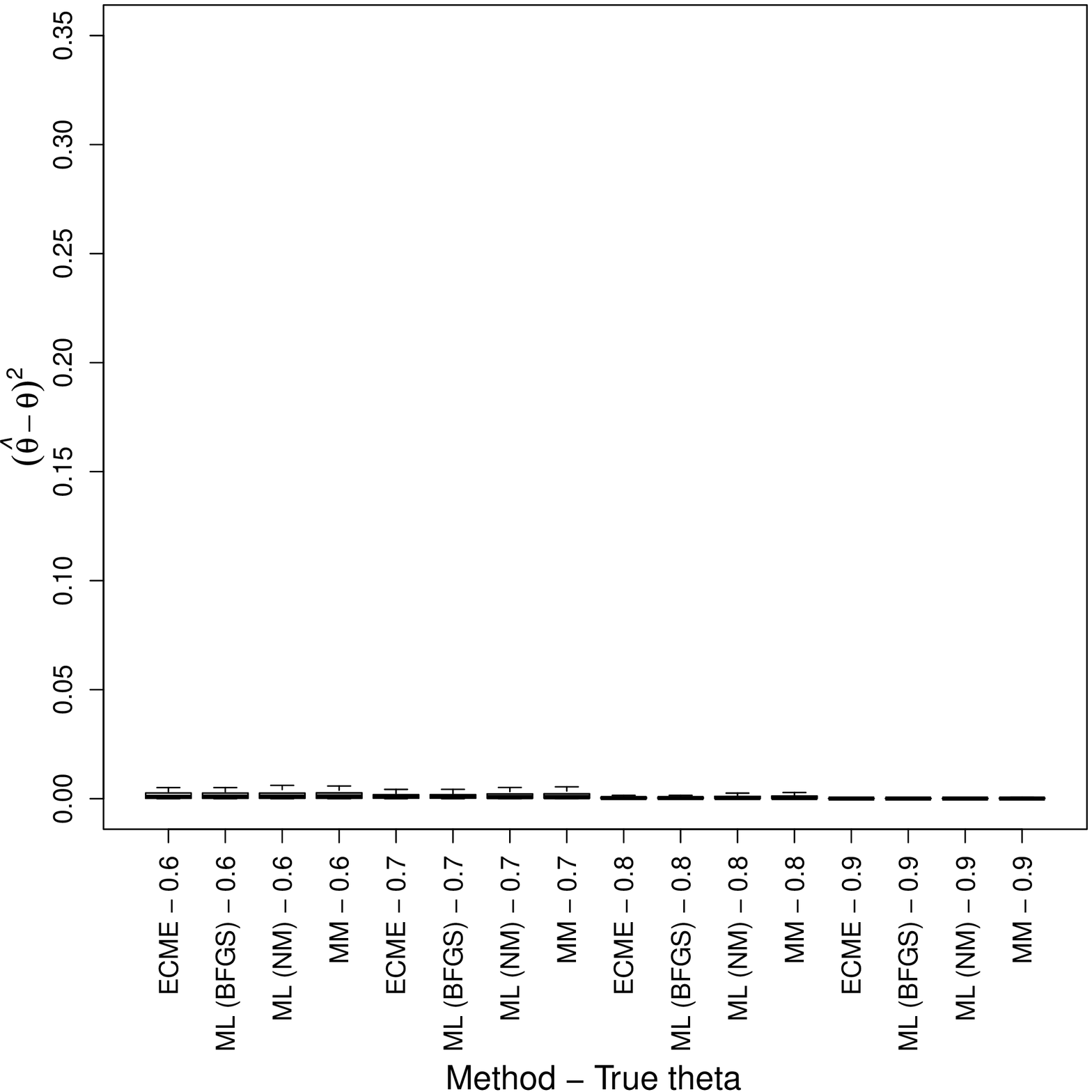}}
\caption{Box-plots of the differences $(\hat{\theta}-\theta)^2$ for each pair $\left(n,d\right)$.
The box-plot in each subfigure summarize the results over 100 replications and refers to a combination between $\theta$ and the estimation method considered.}
\label{fig:Theta.MSE}
\end{figure}
\begin{table}[!ht]
\caption{Log-likelihoods values averaged over 100 replications.}
\label{tab:log-likelihoods}
\centering
%\resizebox*{0.9\textwidth}{!}{
\begin{tabular}{rrr rrr}
\toprule
   	&	    &	          &	       \multicolumn{2}{c}{ML direct}	    &	                \\
			\cline{4-5}
$n$	&	$d$	&	$\theta$	&	 Nelder-Mead 	&	BFGS 	&	ECME 	\\	
\midrule
200	&	2	&	0.6	&	-649.860	&	-649.860	&	-649.860	\\	
	&		&	0.7	&	-674.456	&	-674.456	&	-674.456	\\	
	&		&	0.8	&	-699.929	&	-699.928	&	-699.928	\\	
	&		&	0.9	&	-744.023	&	-744.023	&	-744.023	\\[1mm]	
	&	3	&	0.6	&	-973.686	&	-973.680	&	-973.680	\\	
	&		&	0.7	&	-1009.362	&	-1009.341	&	-1009.341	\\	
	&		&	0.8	&	-1051.155	&	-1051.113	&	-1051.113	\\	
	&		&	0.9	&	-1107.200	&	-1107.120	&	-1107.120	\\[1mm]	
	&	5	&	0.6	&	-1620.323	&	-1620.245	&	-1620.245	\\	
	&		&	0.7	&	-1677.022	&	-1676.873	&	-1676.873	\\	
	&		&	0.8	&	-1744.074	&	-1743.715	&	-1743.715	\\	
	&		&	0.9	&	-1844.226	&	-1843.270	&	-1843.270	\\[1mm]	
500	&	2	&	0.6	&	-1623.352	&	-1623.353	&	-1623.353	\\	
	&		&	0.7	&	-1686.506	&	-1686.505	&	-1686.505	\\	
	&		&	0.8	&	-1758.846	&	-1758.846	&	-1758.846	\\	
	&		&	0.9	&	-1853.362	&	-1853.369	&	-1853.369	\\[1mm]	
	&	3	&	0.6	&	-2433.665	&	-2433.656	&	-2433.656	\\	
	&		&	0.7	&	-2532.655	&	-2532.629	&	-2532.629	\\	
	&		&	0.8	&	-2639.989	&	-2639.930	&	-2639.930	\\	
	&		&	0.9	&	-2785.396	&	-2785.292	&	-2785.292	\\[1mm]	
	&	5	&	0.6	&	-4066.467	&	-4066.397	&	-4066.397	\\	
	&		&	0.7	&	-4209.660	&	-4209.510	&	-4209.510	\\	
	&		&	0.8	&	-4376.644	&	-4376.347	&	-4376.347	\\	
	&		&	0.9	&	-4613.231	&	-4612.413	&	-4612.413	\\[1mm]	
1000	&	2	&	0.6	&	-3257.490	&	-3257.490	&	-3257.490	\\	
	&		&	0.7	&	-3368.835	&	-3368.835	&	-3368.835	\\	
	&		&	0.8	&	-3513.907	&	-3513.906	&	-3513.906	\\	
	&		&	0.9	&	-3722.987	&	-3722.986	&	-3722.986	\\[1mm]	
	&	3	&	0.6	&	-4887.646	&	-4887.634	&	-4887.634	\\	
	&		&	0.7	&	-5050.221	&	-5050.197	&	-5050.197	\\	
	&		&	0.8	&	-5270.342	&	-5270.311	&	-5270.311	\\	
	&		&	0.9	&	-5558.195	&	-5558.122	&	-5558.122	\\[1mm]	
	&	5	&	0.6	&	-8125.940	&	-8125.868	&	-8125.868	\\	
	&		&	0.7	&	-8418.000	&	-8417.825	&	-8417.825	\\	
	&		&	0.8	&	-8759.827	&	-8759.525	&	-8759.525	\\	
	&		&	0.9	&	-9241.587	&	-9240.833	&	-9240.833	\\	
\bottomrule
\end{tabular}
%}
\end{table} 

%%%%%%%%%%%%%%%%%%%%%%%%%
%% Computational times %%
%%%%%%%%%%%%%%%%%%%%%%%%%

As concerns the computational times, for each combination of $\theta$ and the estimation method considered, \figurename~\ref{fig:TIMEmarginal} reports the box-plots of the elapsed times in each of the $3\times 3\times 100=900$ replications obtained marginalizing with respect to $n$ and $d$.
\begin{figure}[!ht]
\centering
  \resizebox{0.6\textwidth}{!}{\includegraphics{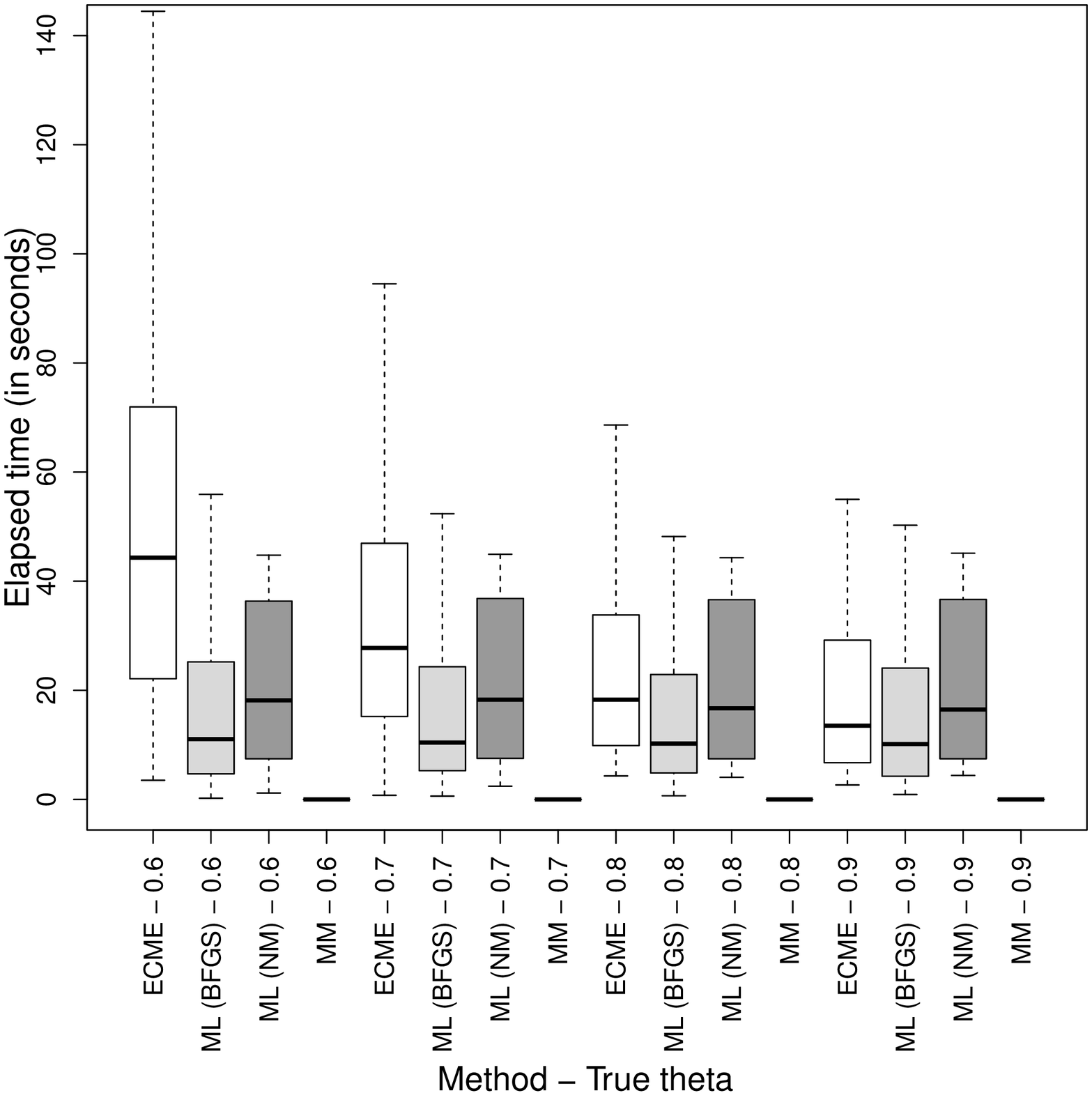}} %
\caption{Elapsed time (in seconds) on each combination between $\theta$ and the estimation method considered.
Each box plot refers to $3\times 3\times 100=900$ values obtained summing out with respect to $d\in \left\{2,3,5\right\}$ and $n\in \left\{200,500,1000\right\}$.}
\label{fig:TIMEmarginal}       
\end{figure}
The MM provides the best (lowest) elapsed times, regardless of the true underlying $\theta$.
Among the algorithms considered to obtain ML estimates, the BFGS is the best one, followed by the Nelder-Mead and ECME algorithms.
The elapsed times from the ECME algorithm seem to be a decreasing function of $\theta$, tying those of the BFGS algorithm when $\theta=0.9$.
Another interesting result can be noted by looking at \figurename~\ref{fig:TIMEd5} where we report, analogously to \figurename~\ref{fig:TIMEmarginal}, the box plots of the elapsed times only in the case $d=5$, when varying the sample size.
\begin{figure}[!ht]
\centering
\subfigure[$d=5$ and $n=200$ \label{fig:d5n200}]
{\includegraphics[width=0.328\textwidth]{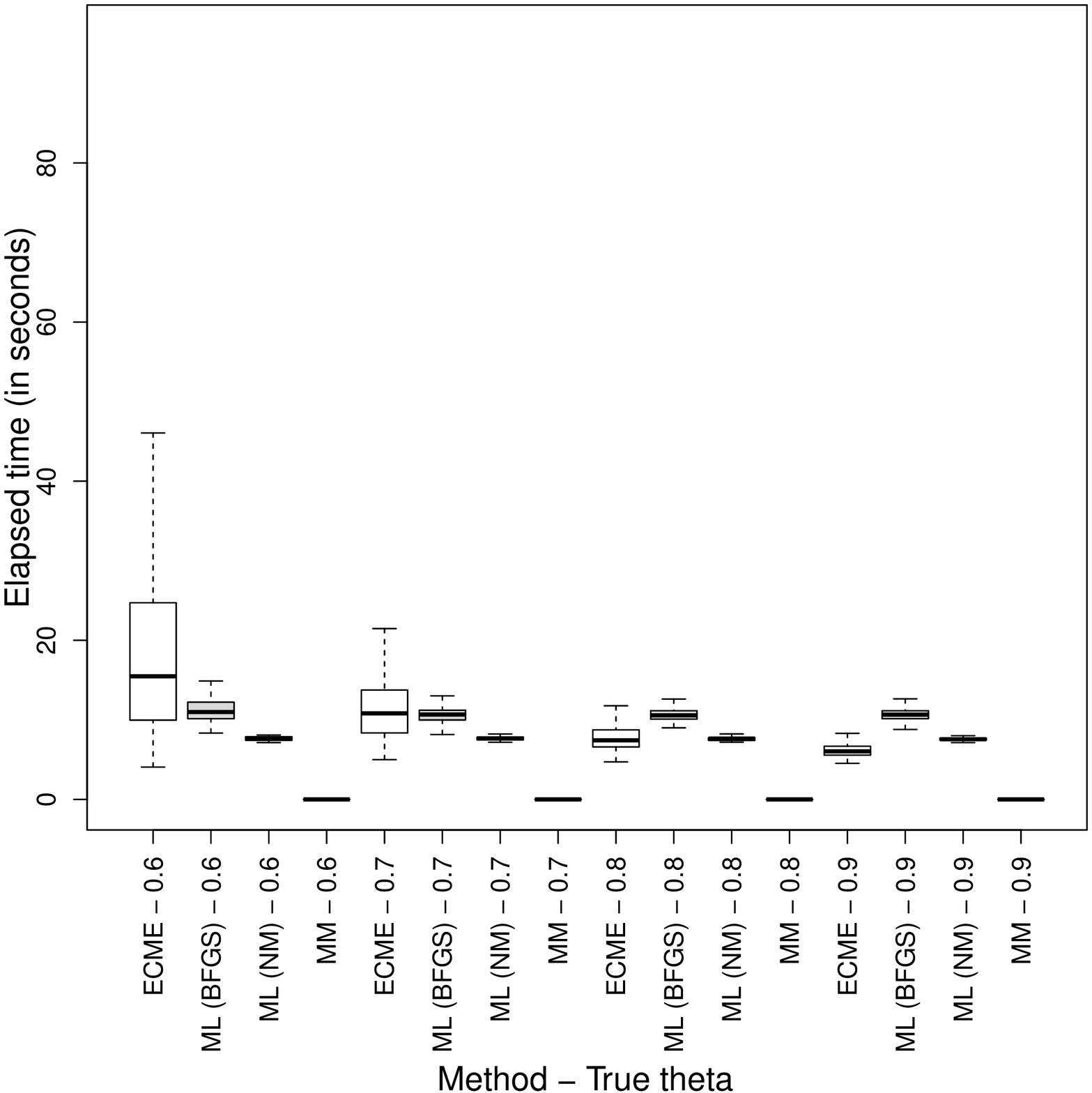}}
\subfigure[$d=5$ and $n=500$\label{fig:d5n500}]
{\includegraphics[width=0.328\textwidth]{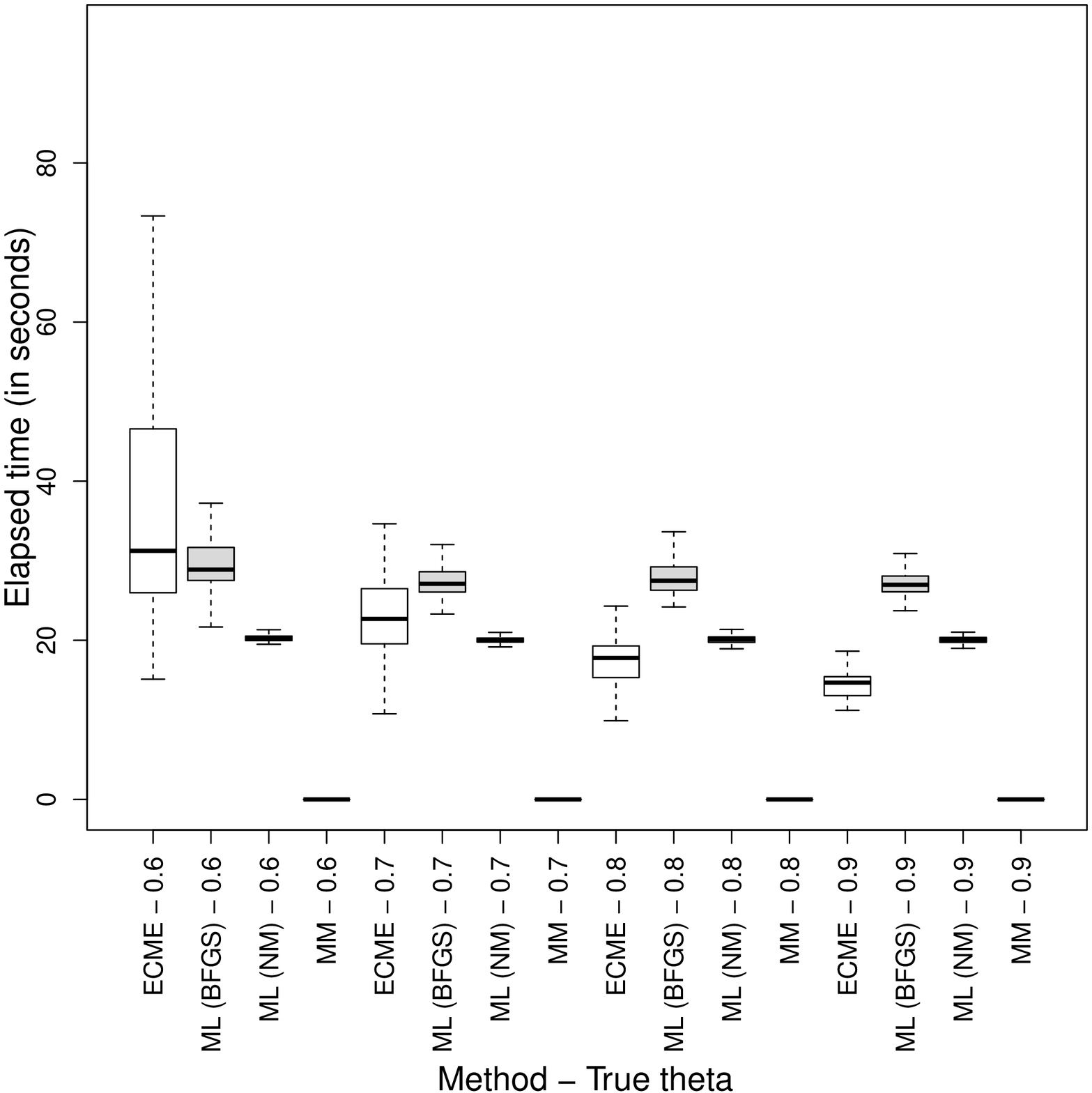}}
\subfigure[$d=5$ and $n=1000$\label{fig:d5n1000}]
{\includegraphics[width=0.328\textwidth]{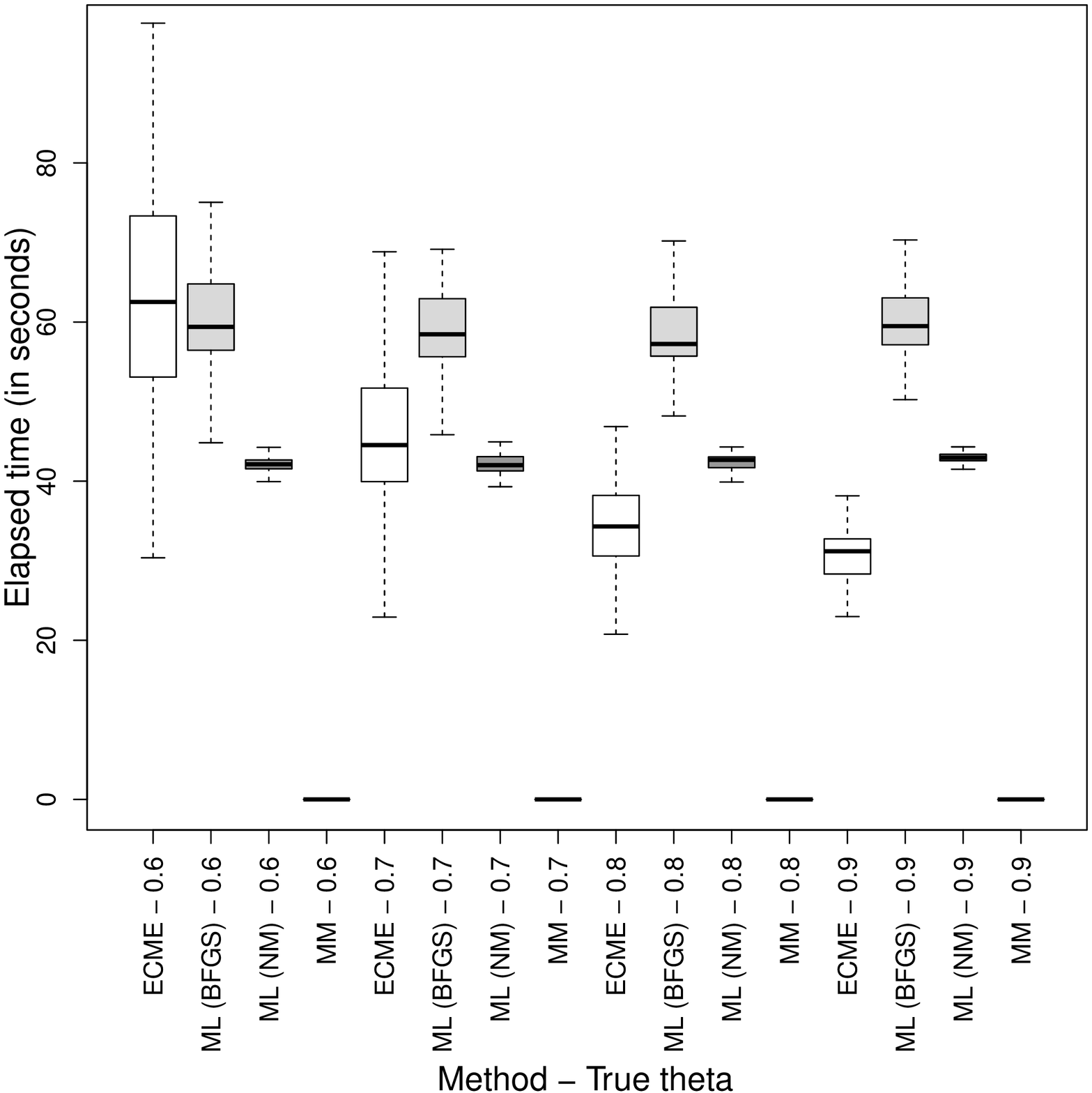}}
\caption{Elapsed time (in seconds) in the case $d=5$ on each combination between $\theta$ and the estimation method considered, when varying $n\in \left\{200,500,1000\right\}$.
Each box plot refers to 100 replications.}
\label{fig:TIMEd5}
\end{figure}
Indeed, the ECME algorithm provides the best elapsed times for high values of $\theta$ (say $\theta>0.7$), regardless of the sample size.
Other simulations with greater values of $d$, not reported here for the sake of space, have shown that the ECME is the fastest algorithm to obtain ML estimates, and this is true regardless of $\theta$.

Summarizing, giving prominence to parameter recovery, we suggest to use the ML approach.
Among the three algorithms considered to obtain ML estimates, we suggest direct ML via the BFGS algorithm for low dimensions (say $d<4$) and the ECME algorithm for high dimensions (say $d>10$).
In the other cases, the choice also depends on the value of the true but unknown inflation parameter.
A possible solution may be to preliminary run the method of moments and, if the MM-estimated value of $\theta$ is high, then using the ECME algorithm, and using the BFGS algorithm otherwise.
%In the remaining cases, since we do not know the true underlying value of $\theta$, we can equally use one of the two algorithms.
%Note that, for the financial data analyzed in Section~\ref{sec:Real data analysis}, where the kurtosis is typically very high, we suggest USARE IL METODO DEI MOMENTI. Per dati finanziari, dove la kurtosi è alta, usare ECME

%\clearpage

\subsection{Model selection: AIC versus BIC}
\label{subsec:Model selection: AIC versus BIC}

%For practical purposes, 
Model selection is usually required to select among a set of candidate models with differing number of parameters.
%, for data at hand.
The classical way to perform this selection is via the computation of a convenient (likelihood-based) model selection criterion.
%: the model associated with the best value of the adopted criterion is selected.
Two widely considered choices, that in our formulation need to be maximized, are the Akaike information criterion (AIC; \citealp{Akai:Anew:1974}) 
\begin{equation*}
\text{AIC} = 2l\left(\widehat{\bPsi}\right)-2\times\#\text{par},
%\label{eq:AIC}
\end{equation*}
and the Bayesian information criterion \citep[BIC;][]{Schw:Esti:1978}
\begin{equation*}
\text{BIC} = 2l\left(\widehat{\bPsi}\right)-\ln n\times \#\text{par},
%\label{eq:BIC}
\end{equation*}
where $\#\text{par}$ is the overall number of free parameters in the model.
AIC and BIC are similar, but with a different penalty for the number of parameters.
BIC penalizes free parameters more strongly than AIC if $n>7$.
%, though it depends on $n$ and relative magnitude of $n$ and $m$.  
%Moreover, in our formulation, they need to be maximized. 
For both, the models being compared need not to be nested, unlike the case when models are being compared using, for example, a likelihood ratio test (\citealp{Gres:Punz:Clos:2013} and \citealp{Punz:Brow:McNi:Hypo:2016}).
%As we can note, AIC and BIC are similar, but with a different penalty for the number of parameters.
A comparison of AIC and BIC is given by \citet[][Section~6.3--6.4]{Burn:Ande:Mode:2013}, with follow-up remarks by \citet{Burn:Ande:Mult:2004}. 
%A point made by several researchers is that AIC and BIC are appropriate for different tasks. 
%In particular, 
BIC is argued to be appropriate for selecting the ``true model'' (i.e.~the process that generated the data) from the set of candidate models.
To be specific, if the ``true model'' is in the set of candidates, then BIC will select it with probability 1, as $n \rightarrow \infty$; in contrast, when selection is done via AIC, the probability can be less than 1 (\citealp[][Section~6.3--6.4]{Burn:Ande:Mode:2013}, \citealp{Vrie:Mode:2012}, and \citealp{Aho:Derr:Pete:Mode:2014}).
Proponents of AIC argue that this issue is negligible, because the ``true model'' is virtually never in the candidate set. 
If the ``true model'' is not in the candidate set, then the most that we can hope to do is select the model that best approximates the ``true model''. 
AIC is appropriate for finding the best approximating model, under certain assumptions.

In this section we compare the performance of AIC and BIC.
We consider two experimental factors: the dimension ($d\in \left\{2,3\right\}$) and the data generating process (DGP).
The sample size is instead fixed to $n=1000$. 
As DGPs, in addition to the MN and MTIN, we consider several multivariate elliptical heavy-tailed distributions used in the literature.
They are the M$t$, MLN, MCN, MSVG, MSH, SNIG, and MSGH distributions given in Section~\ref{subsec:Representations}, with the MLN denoting the multivariate leptokurtic normal distribution introduced by \citet{Bagn:Punz:Zoia:Them:2016}. 
All the DGPs share a zero mean vector ($\bmu=\bzero$) and an identity scale matrix ($\bSigma=\bI$).
The remaining parameter(s) governing the tail-weight of the heavy-tailed competing distributions are: $\nu=4.5$ (degrees of freedom) for the M$t$, $\theta=0.99$ (inflation parameter) for the MTIN, $\beta=6$ (excess kurtosis) for the MLN, $\alpha=0.9$ and $\eta=5$ (proportion of good points and degree of contamination) for the MCN, $\lambda=\psi=3$ for the MSVG, $\chi=0.5$ and $\psi=3$ for the MSH, $\chi=3$ and $\psi=0.2$ for the SNIG, and $\lambda=\chi=\psi=0.1$ for the MSGH; for details about the parameters $\lambda$, $\chi$ and $\psi$ of the distributions within the MSGH family, refer to \appendixname~\ref{app:SGHD}.

For each combination of the simulation factors, one-hundred replications are considered, yielding a total of $9\times 2\times 100=1800$ generated datasets.
For each, we fit all the models in the considered set of DGPs.
Parameters are estimated based on the ML approach.
We adopt the ECME algorithm for the MTIN distribution. 
As concerns the competing models, we use the \texttt{WML.MLN()} function of the code available at \url{http://www.olddei.unict.it/punzo/Rcode.htm} to fit the MLN, the \texttt{CNmixt()} function of the \textbf{ContaminatedMixt} package \citep{Punz:Mazz:McNi:Cont:2018,R:ContaminatedMixt} to fit the MCN, and the \textbf{ghyp} package \citep{ghyp} to fit the remaining models.

\tablename s~\ref{tab:AICBIC_d=2} and \ref{tab:AICBIC_d=3} show the percentage of times AIC and BIC select each model in the cases $d=2$ and $d=3$, respectively.
Results are organized as a contingency table where the true DGP is given by column and the fitted models by row.
Each cell of the contingency table reports two selection counts referred to AIC (on the top) and BIC (on the bottom). 
The cells on the diagonal reports a sort of true positive percentage (TPP), measuring the percentage of times that each criterion is able to discover the true DGP.
We can note how, regardless of $d$, the criteria are able enough to recognize the true underlying DGP being the counts mainly concentrated on the diagonal cells; however, for some DGPs within the MSGH family, such as the MSVG and MSH, their considerable similarity makes discovering the true model more difficult.
%assessment of the true model challenging.
This is probably due to the considered parameterizations under the MSVG and MSH DGPs.
%being the modal choice always toward the true model.
%The last column of these tables helps us to view the overall results by showing
The last column of each table shows the average number of times a model different from the underlying DGP is selected; this can be meant as a sort of false positive percentage (FPP). 
The ideal situation for AIC and BIC should have zero off-diagonal values.
\newcolumntype{C}{>{\centering\arraybackslash}p{10cm}}
\begin{table}[!ht]	
\centering																				
 \resizebox{\textwidth}{!}{
 \begin{tabular}{lc*{10}{>{\centering\arraybackslash}p{.065\linewidth}}}
%{|*{9}{C|}}
%{lc rrrrrrrrr} 
%{lc@{\hskip 0.3in}r@{\hskip 0.3in}rr@{\hskip 0.3in}rrc}	
%{lr*{9}{p{1cm}}}%
\toprule													
	&		&		\multicolumn{10}{c}{True}																																		\\	
\cline{3-11}																			
Fitted	&		&		MN	&		Mt	&		MTIN	&		MLN	&		MCN	&		MSVG	&		MSH	&		SNIG	&		MSGH	&		FPP	\\	
\midrule																																		
MN	&	AIC	&	\cellcolor{lightgray}	91	&		0	&		0	&		0	&		0	&		0	&		0	&		0	&		0	&		0.000	\\	
	&	BIC	&	\cellcolor{lightgray}	100	&		0	&		0	&		0	&		0	&		0	&		0	&		0	&		0	&		0.000	\\	[-3mm]
	&		&			&			&			&			&			&			&			&			&			&			\\	
Mt	&	AIC	&		1	&	\cellcolor{lightgray}	68	&		11	&		0	&		8	&		3	&		4	&		5	&		0	&		4.000	\\	
	&	BIC	&		0	&	\cellcolor{lightgray}	68	&		11	&		0	&		20	&		3	&		4	&		5	&		0	&		5.375	\\	[-3mm]
	&		&			&			&			&			&			&			&			&			&			&			\\	
MTIN	&	AIC	&		1	&		16	&	\cellcolor{lightgray}	88	&		0	&		1	&		7	&		8	&		2	&		0	&		4.375	\\	
	&	BIC	&		0	&		16	&	\cellcolor{lightgray}	88	&		0	&		12	&		7	&		10	&		2	&		0	&		5.875	\\	[-3mm]
	&		&			&			&			&			&			&			&			&			&			&			\\	
MLN	&	AIC	&		6	&		0	&		0	&	\cellcolor{lightgray}	100	&		0	&		12	&		0	&		0	&		0	&		2.250	\\	
	&	BIC	&		0	&		0	&		0	&	\cellcolor{lightgray}	100	&		0	&		14	&		0	&		0	&		0	&		1.750	\\	[-3mm]
	&		&			&			&			&			&			&			&			&			&			&			\\	
MCN	&	AIC	&		0	&		0	&		0	&		0	&	\cellcolor{lightgray}	91	&		6	&		7	&		1	&		0	&		1.750	\\	
	&	BIC	&		0	&		0	&		0	&		0	&	\cellcolor{lightgray}	68	&		0	&		0	&		0	&		0	&		0.000	\\	[-3mm]
	&		&			&			&			&			&			&			&			&			&			&			\\	
MSVG	&	AIC	&		1	&		0	&		0	&		0	&		0	&	\cellcolor{lightgray}	56	&		18	&		3	&		0	&		2.750	\\	
	&	BIC	&		0	&		0	&		0	&		0	&		0	&	\cellcolor{lightgray}	56	&		18	&		4	&		1	&		2.875	\\	[-3mm]
	&		&			&			&			&			&			&			&			&			&			&			\\	
MSH	&	AIC	&		0	&		1	&		0	&		0	&		0	&		10	&	\cellcolor{lightgray}	33	&		0	&		0	&		1.375	\\	
	&	BIC	&		0	&		1	&		0	&		0	&		0	&		14	&	\cellcolor{lightgray}	37	&		0	&		0	&		1.875	\\	[-3mm]
	&		&			&			&			&			&			&			&			&			&			&			\\	
SNIG	&	AIC	&		0	&		15	&		0	&		0	&		0	&		6	&		30	&	\cellcolor{lightgray}	84	&		0	&		6.375	\\	
	&	BIC	&		0	&		15	&		0	&		0	&		0	&		6	&		31	&	\cellcolor{lightgray}	89	&		0	&		6.500	\\	[-3mm]
	&		&			&			&			&			&			&			&			&			&			&			\\	
MSGH	&	AIC	&		0	&		0	&		1	&		0	&		0	&		0	&		0	&		5	&	\cellcolor{lightgray}	100	&		0.750	\\	
	&	BIC	&		0	&		0	&		1	&		0	&		0	&		0	&		0	&		0	&	\cellcolor{lightgray}	99	&		0.125	\\	
\bottomrule
\end{tabular}			
}																		
\caption{Percentage of times AIC and BIC select each model.
The true DGP is shown by column, while the fitted models are given by row.
The last column provides the false positive percentage (FPP).}		
\label{tab:AICBIC_d=2}																			
\end{table}	
\newcolumntype{C}{>{\centering\arraybackslash}p{10cm}}
\begin{table}[!ht]	
\centering																				
\resizebox{\textwidth}{!}{
 \begin{tabular}{lc*{10}{>{\centering\arraybackslash}p{.065\linewidth}}}
%{|*{9}{C|}}
%{lc rrrrrrrrr} 
%{lc@{\hskip 0.3in}r@{\hskip 0.3in}rr@{\hskip 0.3in}rrc}	
%{lr*{9}{p{1cm}}}%
\toprule													
	&		&		\multicolumn{10}{c}{True}																																		\\	
\cline{3-11}																																								
Fitted	&		&		MN	&		M$t$	&		MTIN	&		MLN	&		MCN	&		MSVG	&		MSH	&		SNIG	&		MSGH	&	FPP	\\	
\midrule																																	
MN	&	AIC	&	\cellcolor{lightgray}	96	&		0	&		0	&		0	&		0	&		0	&		0	&		0	&		0	&	0.000	\\	
	&	BIC	&	\cellcolor{lightgray}	100	&		0	&		0	&		0	&		0	&		0	&		0	&		0	&		0	&	0.000	\\	[-3mm]
	&		&			&			&			&			&			&			&			&			&			&			\\
M$t$	&	AIC	&		0	&	\cellcolor{lightgray}	77	&		5	&		0	&		1	&		1	&		0	&		0	&		1	&	1.000	\\	
	&	BIC	&		0	&	\cellcolor{lightgray}	78	&		5	&		0	&		6	&		1	&		0	&		4	&		6	&	2.750	\\	[-3mm]
	&		&			&			&			&			&			&			&			&			&			&			\\
MTIN	&	AIC	&		1	&		8	&	\cellcolor{lightgray}	94	&		0	&		1	&		3	&		0	&		1	&		1	&	1.875	\\	
	&	BIC	&		0	&		8	&	\cellcolor{lightgray}	95	&		0	&		1	&		3	&		0	&		1	&		1	&	1.750	\\	[-3mm]
	&		&			&			&			&			&			&			&			&			&			&			\\
MLN	&	AIC	&		3	&		0	&		0	&	\cellcolor{lightgray}	97	&		0	&		2	&		0	&		0	&		0	&	0.625	\\	
	&	BIC	&		0	&		0	&		0	&	\cellcolor{lightgray}	100	&		0	&		4	&		0	&		0	&		0	&	0.500	\\	[-3mm]
	&		&			&			&			&			&			&			&			&			&			&			\\
MCN	&	AIC	&		0	&		0	&		1	&		3	&	\cellcolor{lightgray}	98	&		7	&		1	&		0	&		0	&	1.500	\\	
	&	BIC	&		0	&		0	&		0	&		0	&	\cellcolor{lightgray}	93	&		1	&		0	&		0	&		0	&	0.125	\\	[-3mm]
	&		&			&			&			&			&			&			&			&			&			&			\\
MSVG	&	AIC	&		0	&		0	&		0	&		0	&		0	&	\cellcolor{lightgray}	56	&		23	&		1	&		0	&	3.000	\\	
	&	BIC	&		0	&		0	&		0	&		0	&		0	&	\cellcolor{lightgray}	56	&		24	&		1	&		0	&	3.125	\\	[-3mm]
	&		&			&			&			&			&			&			&			&			&			&			\\
MSH	&	AIC	&		0	&		0	&		0	&		0	&		0	&		22	&	\cellcolor{lightgray}	41	&		0	&		0	&	2.750	\\	
	&	BIC	&		0	&		0	&		0	&		0	&		0	&		26	&	\cellcolor{lightgray}	42	&		0	&		0	&	3.250	\\	[-3mm]
	&		&			&			&			&			&			&			&			&			&			&			\\
SNIG	&	AIC	&		0	&		13	&		0	&		0	&		0	&		9	&		31	&	\cellcolor{lightgray}	83	&		0	&	6.625	\\	
	&	BIC	&		0	&		14	&		0	&		0	&		0	&		9	&		34	&	\cellcolor{lightgray}	93	&		0	&	0.000	\\	[-3mm]
	&		&			&			&			&			&			&			&			&			&			&			\\
MSGH	&	AIC	&		0	&		2	&		0	&		0	&		0	&		0	&		4	&		15	&	\cellcolor{lightgray}	98	&	2.625	\\	
	&	BIC	&		0	&		0	&		0	&		0	&		0	&		0	&		0	&		1	&	\cellcolor{lightgray}	93	&	0.125	\\	
\bottomrule
\end{tabular}			
}																		
\caption{Percentage of times AIC and BIC select each model.
The true DGP is shown by column, while the fitted models are given by row.
The last column provides the false positive percentage (FPP).}				
\label{tab:AICBIC_d=3}																			
\end{table}

The first aspect we note is that, as expected, the AIC tends to select less parsimonious models.
As an example, consider 
%the case $d=2$ of 
\tablename~\ref{tab:AICBIC_d=2}.
On the one hand, under a MN-DGP, the BIC always discovers the true model, while the AIC selects 8 times models with an additional parameter (M$t$, MTIN, and MLN), and one time a model with two additional parameters (MSGV).
On the other hand, under a MCN-DGP, the truth is discovered 91 times by the AIC and only 68 times by the BIC.
This happens because the BIC has a propensity for more parsimonious models such as the M$t$ and MTIN, having only one additional parameter with respect to the MN.
Regardless of the considered number of dimensions, another aspect we note is that, for some DGPs within the MSGH family, such as the MSVG and MSH, their considerable similarity makes assessment of the true model challenging.
%Regardless of the considered number of dimensions, another aspect we note is that, for some DGPs such as the MSVG and MSH, the considerable similarity between some of the models (MSVG, MSH, and SNIG) nested within the MSGH distribution, makes assessment of the true model challenging.
%This is probably due to the considered parameterizations under the MSVG and MSH DGPs.
In summary, it is difficult to establish the best model selection criterion among those considered; thus, we will use both in the real data analyses presented in Section~\ref{sec:Real data analysis}.

\section{Financial applications}
\label{sec:Real data analysis}

Risk management and portfolio selection are interesting financial areas of application of the proposed MTIN distribution.
In these areas, one of the main problems is the modelling of the joint distribution of stock-prices and asset-returns.
So, the models that are commonly considered are inherently multivariate, as stressed by \citet[][Chapter~3]{McNe:Frey:Embr:Quan:2005}, with the MN distribution playing a special rule (\citealp{Krin:Rach:Hoch:Fabo:Esti:2008} and \citealp[][Chapter~14]{Rach:Hoec:Fabo:Foca:Prob:2010}).
However, many empirical studies show that the MN is not appropriate for the distribution of risk factor returns and stock returns (see, e.g., \citealp{Mand:TheV:1963} and \citealp{Fama:Port:1965}).
Two possible motivations for this inappropriateness concern its elliptical symmetry and thin tails \citep{Bing:Kies:Rudi:Semi:2002}.
While the former assumption is not a limitation, as shown by \citet[][Section~3.3]{McNe:Frey:Embr:Quan:2005} via empirical studies, the latter one is more harmful since the MN-tails are not consistent with the empirical heavy tails of the distribution of returns. 
This has motivated numerous proposals for alternative parametric multivariate elliptical heavy-tailed distributions.
One of the most famous proposals in this direction is represented by the multivariate $\alpha$-stable sub-Gaussian (M$\alpha$SSG) distributions \citep[see, e.g.,][]{Krin:Rach:Hoch:Fabo:Esti:2008}. %, a subclass of the multivariate $\alpha$-stable distributions.
Unfortunately, these distributions are so heavy-tailed that the second moment is infinite, a fact that is inconsistent with empirical findings for most financial returns \citep{Bing:Kies:Rudi:Semi:2002} and with the fact that many theoretical models in finance rely on the existence of this moment \citep[][p.~290]{Rach:Hoec:Fabo:Foca:Prob:2010}.
A further drawback of the M$\alpha$SSG distributions is that, with a few exceptions, they do not possess an analytic expression for the pdf \citep[][p.~112]{Krin:Rach:Hoch:Fabo:Esti:2008}.
This is the reason why other multivariate elliptical heavy-tailed distributions, such as those from the MNSM family (see Section~\ref{subsubsec:Multivariate normal scale mixtures}), are typically preferred for financial modelling.

%\subsection{Operational and computational aspects}
%\label{subsec:Operational}

Motivated by the above considerations, the proposed MTIN model is compared, on real financial data, with the same multivariate elliptical heavy-tailed distributions already considered in Section~\ref{subsec:Model selection: AIC versus BIC}, which are also widely used in the financial literature.
Parameters are estimated based on the ML approach, under the common simplifying assumption that returns form iid samples \citep[see, e.g.,][p.~84]{McNe:Frey:Embr:Quan:2005}.
Computational details have been already discussed in Section~\ref{sec:Simulation studies}.
The comparison is made in terms of AIC, BIC, and appropriateness of the fitted models in reproducing the empirical kurtosis.

%\subsection{Emprirical results}
%\label{subsec:Emprirical}

We fit the competing models to three multivariate time series related to 4 of the 30 large publicly owned companies considered by the Dow Jones index.
The companies are American Express (\textsf{AXP}), Boeing (\textsf{BA}), Intel (\textsf{INTC}), and Microsoft (\textsf{MSFT}).
All the codes needed to replicate the analysis are included in the Supplementary Material.
We consider daily log-returns spanning the period from January 5th, 2015 to December 1th, 2017 ($n = 734$ observations downloadable from \url{http://finance.yahoo.com/}).
The matrix of scatter plots is displayed in \figurename~\ref{fig:d=4}.
\begin{figure}[!ht]
\centering
  \resizebox{0.6\textwidth}{!}{\includegraphics{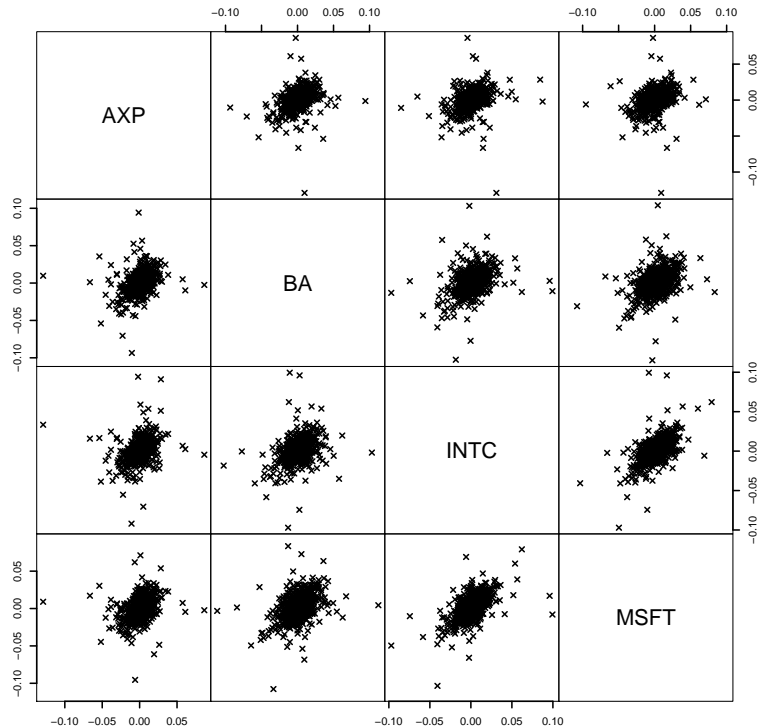}} %
\caption{
Pairwise scatter plots of daily log-returns of the quadruple $(\textsf{AXP},\textsf{BA},\textsf{INTC},\textsf{MSFT})$ spanning the period from January 5th, 2015 to December 1th, 2017.
}
\label{fig:d=4}       
\end{figure}

\tablename~\ref{tab:univariate descriptive} provides some descriptive statistics (and Jarque-Bera normality tests) for the \textsf{AXP}, \textsf{BA}, \textsf{INTC}, and \textsf{MSFT} series, separately considered.
\begin{table}[!ht]
\caption{Descriptive statistics, and Jarque-Bera normality tests (with $p$-values in brackets), for the \textsf{AXP}, \textsf{BA}, \textsf{INTC}, and \textsf{MSFT} series.}
\label{tab:univariate descriptive}
\centering
%\resizebox*{0.9\textwidth}{!}{
\begin{tabular}{l c rrrr}
\toprule
   && Mean & Std. Dev & Excess kurtosis & Jarque-Bera test \\
\midrule
\textsf{AXP}  && 0.000 & 0.013 & 17.648 & 9756.065 ($<$ 0.001)  \\
\textsf{BA}   && 0.001 & 0.014 & 8.116 &  2022.966 ($<$ 0.001)  \\
\textsf{INTC} && 0.000 & 0.014 & 5.574 &  970.972 ($<$ 0.001)  \\
\textsf{MSFT} && 0.000 & 0.014 & 10.964 &  3508.548 ($<$ 0.001)  \\
\bottomrule
\end{tabular}
%}
\end{table}
All the series are leptokurtic and the Jarque-Bera statistic confirms the departure from univariate normality at the 1$\permille$ level.
As concerns the dependence between stocks, \tablename~\ref{tab:correlations} displays the pairwise sample correlations along with $p$-values from the tests of uncorrelation computed via the \texttt{cor.test()} function of the \textbf{stats} package.
\begin{table}[!ht]
\caption{Pairwise correlations between \textsf{AXP}, \textsf{BA}, \textsf{INTC}, and \textsf{MSFT}.
$p$-values from the tests of pairwise uncorrelation are given in brackets.}
\label{tab:correlations}
\centering
%\resizebox*{0.9\textwidth}{!}{
\begin{tabular}{cccc}
\toprule
&   \textsf{BA} & \textsf{INTC} & \textsf{MSFT}\\
\midrule
\textsf{AXP}   &	        0.342  &	0.311 	       &  0.294    \\
               &     ($<$ 0.001) &  ($<$ 0.001)    &  ($<$ 0.001)   \\
\textsf{BA}    &		             &	0.400 	       &  0.362   \\
               &                 &     ($<$ 0.001) &  ($<$ 0.001)   \\
\textsf{INTC}  &		             &	      	       &  0.569   \\
               &                 &                 &  ($<$ 0.001)   \\
\bottomrule
\end{tabular}
%}
\end{table}
As we can see, the four correlations are all significantly different from zero confirming the need of a multivariate model allowing for correlation between series.  
%\begin{table}[!ht]
%\caption{Descriptive statistics, and Jarque-Bera normality tests (with $p$-values in brackets), for the \textsf{AXP}, \textsf{BA}, and \textsf{INTC} series separately considered.}
%\label{tab:univariate descriptive}
%\centering
%\resizebox*{0.9\textwidth}{!}{
%\begin{tabular}{l c rrrrr}
%\toprule
%   && Mean & Std. Dev & Skewness & Excess kurtosis & Jarque-Bera test \\
%\midrule
%\textsf{AXP}  && 0.000 & 0.013 & -1.250 & 17.456 & 9237.7 ($<$ 0.001)  \\
%\textsf{BA} && 0.001 & 0.014 & -0.265 & 8.062 &  1939.1 ($<$ 0.001)  \\
%\textsf{INTC} && 0.000 & 0.014 & -0.424 & 6.671 &  976.9 ($<$ 0.001)  \\
%\bottomrule
%\end{tabular}
%}
%\end{table}
%\textcolor{blue}{\textbf{TOGLIEREI LA COLONNA SKEWNESS e aggiungerei la matrice di correlazione per dimostrare l'esigenza di un modello multivariato}} \textcolor{red}{\textbf{SONO COMPLETAMENTE D'ACCORDO!}}

\subsection{Bivariate example}

In the first example we consider the bivariate ($d=2$) time series of log-returns of \textsf{AXP} and \textsf{BA} stocks.
\tablename~\ref{tab:DowJones} presents a comparison of the models in terms of goodness-of-fit and ability in reproducing the empirical kurtosis, which is 43.04.
\tablename~\ref{tab:DowJones} also gives rankings induced by AIC, BIC, and by the absolute difference between estimated and empirical kurtosis.
The estimated kurtosis is computed analytically, by substituting the estimated parameters in the theoretical formula for the kurtosis.
The theoretical kurtosis is given in \eqref{eq:kurtosis} for the MTIN, in \appendixname~\ref{app:SGHD} for the MSGH models (along with their particular subcases), in \citet{Kotz:Nada:Mult:2004} for the M$t$, and in \citet{Bagn:Punz:Zoia:Them:2016} for MLN and MCN models.
\begin{table}[!ht]	
\centering																				
 \resizebox{0.9\textwidth}{!}{
 \begin{tabular}{lr@{\hskip 0.3in}r@{\hskip 0.3in}rr@{\hskip 0.3in}rr@{\hskip 0.3in}rrr}	
\toprule													
Model	&	$\#$ par.	&	Log-lik.	&	AIC	&	Ranking	&	BIC	&	Ranking	&	Kurtosis	&	Abs.~diff.	&	Ranking	\\
\midrule
%MN	&	5	&	4170.66	&	8331.31	&	9	&	8308.46	&	9	&	8.00	&	34.52	&	7	\\		
%MTIN	&	6	&	4418.14	&	8824.28	&	1	&	8796.86	&	1	&	40.39	&	2.13	&	1	\\		
%M$t$	&	6	&	4415.15	&	8818.30	&	2	&	8790.88	&	2	&	-	&	-	&	-	\\		
%MSVG	&	6	&	4378.53	&	8745.06	&	6	&	8717.64	&	6	&	14.07	&	28.46	&	5	\\		
%MSH	&	6	&	4378.31	&	8744.63	&	7	&	8717.21	&	7	&	13.05	&	29.47	&	6	\\		
%MSGH	&	7	&	4415.15	&	8816.30	&	3	&	8784.31	&	3	&	760.80	&	718.28	&	8	\\		
%SNIG	&	6	&	4403.97	&	8795.94	&	5	&	8768.52	&	4	&	20.19	&	22.33	&	3	\\		
%MLN	&	6	&	4314.31	&	8616.61	&	8	&	8589.20	&	8	&	14.40	&	28.12	&	4	\\		
%MCN	&	7	&	4406.53	&	8799.05	&	4	&	8767.07	&	5	&	27.97	&	14.56	&	2	\\	
MN	&	5	&	4306.900	&	8603.800	&	9	&	8580.808	&	9	&	8	&	35.043	&	7	\\
M$t$	&	6	&	4559.031	&	9106.063	&	2	&	9078.472	&	2	&	-	&	-	&	9	\\
MTIN	&	6	&	4561.830	&	9111.661	&	1	&	9084.070	&	1	&	44.539	&	1.497	&	1	\\
MLN	&	6	&	4453.772	&	8895.545	&	8	&	8867.953	&	8	&	14.400	&	28.643	&	4	\\
MCN	&	7	&	4549.373	&	9084.747	&	4	&	9052.557	&	5	&	27.929	&	15.114	&	2	\\
MSVG	&	7	&	4521.768	&	9031.536	&	6	&	9003.945	&	6	&	14.083	&	28.960	&	5	\\
MSH	&	7	&	4521.450	&	9030.901	&	7	&	9003.310	&	7	&	13.065	&	29.977	&	6	\\
SNIG	&	7	&	4547.695	&	9083.390	&	5	&	9055.798	&	4	&	20.241	&	22.801	&	3	\\
MSGH	&	8	&	4559.031	&	9104.062	&	3	&	9071.873	&	3	&	82.140	&	39.097	&	8	\\
	
\bottomrule
\end{tabular}			
}																		
\caption{Bivariate example: log-likelihood, AIC, BIC and kurtosis for the competing models, along with rankings induced by these criteria.
The column ``Abs.~diff.'' denotes the absolute difference between empirical and estimated kurtoses.
The empirical kurtosis is 43.04. 
}		
\label{tab:DowJones}																			
\end{table}	

As we can see from \tablename~\ref{tab:DowJones}, AIC and BIC provide the same ranking, except for the SNIG and MCN distributions having a reversed position.
These criteria indicate that the MTIN is the best model; its estimated inflation parameter is $\widehat{\theta}=0.993$ (very close to the upper bound for this parameter).
The second and third best models are the M$t$ and MSGH, respectively.
For the M$t$ model, the estimated degrees of freedom are lower than $4$ ($\widehat{\nu}=3.371$) and, as such, the corresponding kurtosis is undefined.
Among the competing models, the MTIN is the best one also in reproducing the empirical kurtosis (the absolute difference is 1.497).
According to this criterion, the second best model is the MCN; interestingly, although the MCN model has an additional parameter with respect to the MTIN, it provides a worse result.
%; it, however, has .
%one providing the lower absolute value difference meaning that it is the closest model in terms of estimated kurtosis.

%\textcolor{blue}{DIRE QUALCOSA SU MSGH?}

\subsection{Trivariate example}

In the second example we consider the trivariate ($d=3$) time series of log-returns of the \textsf{AXP}, \textsf{BA}, and \textsf{INTC} stocks.
AIC and BIC in \tablename~\ref{tab:DowJones2} provide, for the fitted models, the same ranking. 
As for the previous example, considering AIC and BIC, the MTIN, M$t$, and MSGH models occupy the first, second, and third positions, respectively.
The estimated inflation parameter ($\widehat{\theta}=0.990$) is close to the upper bound also in this case.

The estimated degrees of freedom for the M$t$ distribution are lower than $4$ also in this example ($\widehat{\nu}=3.497$); then, we can not evaluate the ability of the M$t$ model in reproducing the empirical kurtosis (63.48).
According to the absolute difference between estimated and empirical kurtosis, the best model is the MTIN while the second best is the MCN.
\begin{table}[!ht]	
\centering																				
 \resizebox{0.9\textwidth}{!}{
\begin{tabular}{lr@{\hskip 0.3in}r@{\hskip 0.3in}rr@{\hskip 0.3in}rr@{\hskip 0.3in}rrr}	
\toprule													
Model	&	$\#$ par.	&	Log-Lik.	&	AIC	&	Ranking	&	BIC	&	Ranking	&	 Kurtosis	&	Abs.~diff. &	Ranking	\\
\midrule
%MN	&	9	&	6298.33	&	12578.65	&	9	&	12537.53	&	9	&	15.00	&	48.13	&	8	\\		
%MTIN	&	10	&	6635.22	&	13250.44	&	1	&	13204.74	&	1	&	63.17	&	0.04	&	1	\\		
%M$t$	&	10	&	6634.82	&	13249.64	&	2	&	13203.95	&	2	&	-	&	-	&	-	\\		
%MSVG	&	10	&	6593.18	&	13166.36	&	6	&	13120.66	&	6	&	26.06	&	37.07	&	6	\\		
%MSH	&	10	&	6583.83	&	13147.66	&	7	&	13101.96	&	7	&	22.15	&	40.98	&	7	\\		
%MSGH	&	11	&	6634.82	&	13247.64	&	3	&	13197.37	&	3	&	59.87	&	3.27	&	2	\\		
%SNIG	&	10	&	6622.33	&	13224.67	&	4	&	13178.97	&	4	&	39.48	&	23.66	&	4	\\		
%MLN	&	10	&	6431.63	&	12843.25	&	8	&	12797.56	&	8	&	27.00	&	36.13	&	5	\\		
%MCN	&	11	&	6613.12	&	13204.24	&	5	&	13153.97	&	5	&	45.55	&	17.58	&	3	\\	
MN	&	9	&	6499.588	&	12981.180	&	9	&	12939.790	&	9	&	15	&	48.488	&	8	\\
M$t$	&	10	&	6843.439	&	13666.880	&	2	&	13620.890	&	2	&	-	&	-	&	9	\\
MTIN	&	10	&	6843.689	&	13667.380	&	1	&	13621.390	&	1	&	68.182	&	4.693	&	1	\\
MLN	&	10	&	6631.859	&	13243.720	&	8	&	13197.730	&	8	&	27.000	&	36.488	&	5	\\
MCN	&	11	&	6820.018	&	13618.040	&	5	&	13567.450	&	5	&	45.330	&	18.158	&	2	\\
MSVG	&	11	&	6802.071	&	13584.140	&	6	&	13538.160	&	6	&	25.872	&	37.616	&	6	\\
MSH	&	11	&	6792.339	&	13564.680	&	7	&	13518.690	&	7	&	22.500	&	40.988	&	7	\\
SNIG	&	11	&	6831.064	&	13642.130	&	4	&	13596.140	&	4	&	35.843	&	27.646	&	3	\\
MSGH	&	12	&	6843.438	&	13664.880	&	3	&	13614.290	&	3	&	94.084	&	30.596	&	4	\\
\bottomrule
\end{tabular}			
}																		
\caption{Trivariate example: log-likelihood, AIC, BIC and kurtosis for the competing models, along with rankings induced by these criteria.
The column ``Abs.~diff.'' denotes the absolute difference between empirical and estimated kurtoses.
The empirical kurtosis is 63.48.
}		
\label{tab:DowJones2}																			
\end{table}	

\subsection{Quadrivariate example}

In the third example we consider the quadrivariate ($d=4$) time series of log-returns of the \textsf{AXP}, \textsf{BA}, \textsf{INTC}, and \textsf{MSFT} stocks.
According to AIC and BIC, which provide the same ranking (see \tablename~\ref{tab:DowJones3}), the MTIN, M$t$, and MSGH models occupy the first, second, and third positions, respectively.
The estimated inflation parameter results $\widehat{\theta}=0.990$, while the estimated degrees of freedom are $\widehat{\nu}=3.453$.
As concerns the kurtosis, the model having the best performance is the MTIN while the second best, as in the previous two examples, is the MCN.
\begin{table}[!ht]	
\centering																				
 \resizebox{0.9\textwidth}{!}{
\begin{tabular}{lr@{\hskip 0.3in}r@{\hskip 0.3in}rr@{\hskip 0.3in}rr@{\hskip 0.3in}rrr}	
\toprule													
Model	&	$\#$ par.	&	Log-Lik.	&	AIC	&	Ranking	&	BIC	&	Ranking	&	 Kurtosis	&	Abs.~diff. &	Ranking	\\
\midrule
MN	&	14	&	8740.774	&	17453.550	&	9	&	17389.170	&	9	&	24	&	74.474	&	7	\\
M$t$	&	15	&	9236.442	&	18442.880	&	2	&	18373.910	&	2	&	-	&	-	&	9	\\
MTIN	&	15	&	9236.793	&	18443.590	&	1	&	18374.610	&	1	&	109.293	&	10.819	&	1	\\
MLN	&	15	&	8934.125	&	17838.250	&	8	&	17769.270	&	8	&	40.000	&	58.475	&	5	\\
MCN	&	16	&	9198.779	&	18365.560	&	5	&	18291.980	&	5	&	70.307	&	28.168	&	2	\\
MSVG	&	16	&	9180.549	&	18331.100	&	6	&	18262.120	&	6	&	41.269	&	57.206	&	4	\\
MSH	&	16	&	9150.832	&	18271.660	&	7	&	18202.690	&	7	&	33.600	&	64.874	&	6	\\
SNIG	&	16	&	9220.903	&	18411.810	&	4	&	18342.830	&	4	&	57.865	&	40.609	&	3	\\
MSGH	&	17	&	9236.443	&	18440.890	&	3	&	18367.310	&	3	&	182.387	&	83.912	&	8	\\
\bottomrule
\end{tabular}			
}																		
\caption{Quadrivariate example: log-likelihood, AIC, BIC and kurtosis for the competing models, along with rankings induced by these criteria.
The column ``Abs.~diff.'' denotes the absolute difference between empirical and estimated kurtoses.
The empirical kurtosis is 98.474.
}		
\label{tab:DowJones3}																			
\end{table}

\section{Discussion}
\label{sec:Discussion}

In this article, the multivariate tail-inflated normal (MTIN) distribution has been introduced.
Compared with other multivariate heavy-tailed elliptical distributions embedding the multivariate normal (MN), the MTIN has the following main characteristics:
\begin{itemize}
	\item a closed-form representation for the probability density function (differently from the $\alpha$-stable sub-Gaussian distributions, with the exception of the Cauchy distribution);
	\item a single parameter governing the tail weight (differently from the contaminated normal distribution and from the majority of distributions nested within the symmetric generalized hyperbolic);
	\item an MN scale mixture representation (differently from the leptokurtic normal distribution);
	\item the covariance matrix and excess kurtosis always exist (differently from $t$ and $\alpha$-stable sub-Gaussian distributions);
	\item any level of excess kurtosis can be reached (differently from the leptokurtic normal distribution).
\end{itemize}
Moreover, being a member of the MN scale mixture family, the MTIN distribution allows to obtain robust estimates of the parameters of the nested MN distribution when the maximum likelihood estimation paradigm is considered, as discussed in Section~\ref{sec:Some notes on robustness}.  
The analysis on financial data of Section~\ref{sec:Real data analysis} has further corroborated how the proposed model represents a valid alternative to most of the distributions cited above in terms of AIC and BIC, but also in reproducing the empirical kurtosis of the analyzed data that, in the considered financial context, is typically very high.

%\bibliographystyle{Chicago}
%%\setlength{\bibsep}{3pt}
%\bibliography{References}

\bibliographystyle{natbib}      % basic style, author-year citations
%\bibliography{References}

\appendix
\numberwithin{equation}{section}

\section*{Appendix}

%\begin{flushleft}
%\textcolor{blue}{\textbf{STRUTTURATA MALE}}
%\end{flushleft}

\section{Alternative formulations of the pdf}
\label{app:Alternative formulation of the pdf}

An analogous formulation of the pdf in \eqref{eq:MTIN pdf} is
\begin{equation}
f_{\text{TIN}}\left(\bx;\bmu,\bSigma,\theta\right)  =
\frac{2\pi^{-\frac{d}{2}}\left|\bSigma \right|^{-\frac{1}{2}}}{\theta \left[\delta\left(\bx;\bmu,\bSigma\right)\right]^{\left(\frac{d}{2}+1\right)}}    
\left[ \gamma\left(\frac{d}{2}+1, \frac{\delta\left(\bx;\bmu,\bSigma\right)}{2} \right) - \gamma\left(\frac{d}{2}+1,(1-\theta)\frac{\delta\left(\bx;\bmu,\bSigma\right)}{2} \right) \right],  
\label{eq:MTIN pdf lower}
\end{equation}
where $\gamma\left(\cdot,\cdot\right)$ is the lower incomplete gamma function.
A further formulation of the pdf in \eqref{eq:MTIN pdf} is  
%(DA DECIDERE DOVE METTERE) Another way to define \eqref{eq:MTIN pdf} is the following:
\begin{align} 
f_{\text{TIN}}\left(\bx;\bmu,\bSigma,\theta\right)=&
\frac{2\pi^{-\frac{d}{2}}\left|\bSigma \right|^{-\frac{1}{2}}}{\theta \left[\delta\left(\bx;\bmu,\bSigma\right)\right]^{\left(\frac{d}{2}+1\right)}}    
\left\{
\sum_{i=0}^{ \left\lfloor \frac{d}{2}\right\rfloor } \left(\frac{d}{2}\right)_i 
\left(\frac{\delta\left(\bx;\bmu,\bSigma\right)}{2}\right)^{\frac{d}{2}-i}
\left[  \exp\left(\frac{(1-\theta) \delta \left(\bx;\bmu,\bSigma\right) }{2}\right)      \left(1-\theta\right)^{\frac{d}{2}-1}-1 \right]  \right.\nonumber\\
&\left.+  \sqrt{\pi} \left(\frac{d}{2}\right)_{\frac{d-1}{2}} \left[\Phi\left(\sqrt{\delta\left(\bx;\bmu,\bSigma\right)}\right)-\Phi\left(\sqrt{(1-\theta)\delta\left(\bx;\bmu,\bSigma\right)}\right)\right]  \indic_{\left\{\frac{1}{2}\right\}}\left(\frac{d}{2} - \left\lfloor \frac{d}{2}\right\rfloor\right)  \right\},
\label{eq:MTIN pdf new}
\end{align}
where $\left\lfloor \cdot \right\rfloor$ is the floor function, $(a)_b$ is the Pochhammer operator \citep{Abra:Steg:Hand:1965}, and $\Phi(\cdot)$ is the distribution function of a standard normal random variable.
The alternative form in \eqref{eq:MTIN pdf new} can be easily obtained by iteratively using the integration by parts to solve the integral in \eqref{eq:withintegral}.
Note that this formulation of the MTIN distribution simplifies if $d$ is even because the second row in \eqref{eq:MTIN pdf new} vanishes.
%In particular, being $d$ (the dimension) an integer, using it is easy to prove that

\section{Proof of Theorem~\ref{theo: MTIN moments}}
\label{app:Teorema sui momenti della MTIN}

Before to sketch the proof of Theorem~\ref{theo: MTIN moments}, we recall Definition~\ref{def:MNSM}, which is discussed in more details in \citet[][Section~3.2.2]{McNe:Frey:Embr:Quan:2005}, and we provide Theorem~\ref{theo:MNSM moments} along its proof.
These preparatory results give us the possibility to demonstrate Theorem~\ref{theo: MTIN moments}.
%We finally provide a proof of Theorem~\ref{theo: MTIN moments}. 
\begin{defi}[MNSM: elliptical representation]\label{def:MNSM}
A random vector $\bX$ is said to have a multivariate normal scale mixture (MNSM) distribution if 
\begin{equation}
\bX = \bmu + \sqrt{V} T \boldsymbol{\Lambda}  \bU,
\label{eq:nvm}
\end{equation}
where $\bmu$, $T$, $\boldsymbol{\Lambda}$, and $\bU$ are defined as in \eqref{eq:ell}, while $V$ is a non-negative random variable which is independent of $T$ and $\bU$.  
\end{defi}
%\noindent For further details about Definition~\ref{def:MNSM}, see \citet[][Section~3.2.2]{McNe:Frey:Embr:Quan:2005}.
\begin{theorem}[MNSM: mean vector, covariance matrix and kurtosis]
\label{theo:MNSM moments}
If $\bX$ has a MNSM distribution, then
\begin{eqnarray}
\mbox{E}\left(\bX\right) &=& \bmu,   \label{eq:appmean}  \\  
\text{Var}\left(\bX\right) &=& E(V) \bSigma, \label{eq:appvariance} \\
%\text{Skew}\left(\bX\right) &=& 0, \label{eq:skewness} \\
\mbox{Kurt}\left(\bX\right) &=& d\left(d+2\right) \frac{E\left(V^2\right)}{\left[E\left(V\right)\right]^2}. \label{eq:appkurtosis2}
\end{eqnarray}
\end{theorem}
\begin{proof}
Results in \eqref{eq:appmean} and \eqref{eq:appvariance} are proved in \citet[][p.~74]{McNe:Frey:Embr:Quan:2005}.
The result in \eqref{eq:appkurtosis2} follows by
\begin{eqnarray*}
  \mbox{Kurt}(\bX)&=& \mbox{E}\left\{\left[ \left(\bX-\bmu\right)' \left[ \mbox{Var}\left(\bX\right)\right]^{-1}\left(\bX-\bmu\right) \right]^2 \right\} \nonumber \\
	&=& \mbox{E}\left\{\left[ \left(\bX-\bmu\right)' \left[  \mbox{E}(V)\bSigma\right]^{-1}\left(\bX-\bmu\right) \right]^2 \right\} \nonumber \\
	&=& \frac{1}{\left[\mbox{E}(V)\right]^2}\mbox{E} \left[ \left( T^2 V \bU ' \boldsymbol{\Lambda}' \bSigma^{-1} \boldsymbol{\Lambda} \bU\right)^2   \right]  \nonumber \\
	&=& \mbox{E}\left(T^4\right)\frac{\mbox{E}\left( V^2 \right)}{\left[\mbox{E}(V)\right]^2} \\
	&=& d(d+2)\frac{\mbox{E}\left( V^2 \right)}{\left[\mbox{E}(V)\right]^2}.
\end{eqnarray*}
\end{proof}

We are now ready to give the proof of Theorem~\ref{theo: MTIN moments}.
\begin{proof}[Proof of Theorem~\ref{theo: MTIN moments}]
The MTIN distribution, having the stochastic representation in \eqref{eq:ell}, is a particular MNSM where
% , as defined in \eqref{eq:nvm}, in which 
$V$ has the inverse uniform distribution over the interval $\left(1,1/\left(1-\theta\right)\right)$.
%; in symbols, $V\sim \mathcal{U}^-\left(1,1/\left(1-\theta\right)\right)$.
%In particular, the inverse uniform random variable $V=1/W$, with $W\sim \mathcal{U}\left(1-\theta,1\right)$,
% having a uniform distribution in $(1-\theta,1)$, 
In such a case, 
%$V$ has the following first two moments 
\begin{equation}
 \mbox{E}(V) = -\frac{\log(1-\theta)}{\theta} \quad \mbox{and} \quad \mbox{E}\left(V^2\right) = \frac{1}{1-\theta}.
\label{eq:mominv}
\end{equation} 
The result in \eqref{eq:mean} is verified by \eqref{eq:appmean}.  
The results in \eqref{eq:variance} and \eqref{eq:kurtosis} are obtained by substituting the moments given in \eqref{eq:mominv} in \eqref{eq:appvariance} and \eqref{eq:appkurtosis2}, respectively.
\end{proof}

\section{Method of moments}
\label{sec:Method of moments}

In the method of moments applied to the estimation of the parameters of the MTIN distribution, we relate the (unknown) population moments in \eqref{eq:mean}, \eqref{eq:variance}, and \eqref{eq:kurtosis} to their sample counterparts 
%the sample mean, covariance matrix and kurtosis
$$
\overline{\bx} = \frac{1}{n} \sum_{i=1}^n \bx_i,
\quad 
\bS =  \frac{1}{n-1}  \sum_{i=1}^n \left(\bx_i - \overline{\bx}\right) \left(\bx_i - \overline{\bx}\right)',
\quad \text{and} \quad 
\widehat{\text{Kurt}}(\bX) = \frac{1}{n} \sum_{i=1}^n \left[
%\delta\left(\bx_i;\overline{\bx},\bS\right)
\left(\bx_i-\overline{\bx} \right)' \bS^{-1}\left(\bx_i - \overline{\bx}  \right)
\right]^2.
$$
%We recall that $\overline{\bx}$ and $\bS$ are consistent estimators of $\bmu$ and $\bSigma$ \citep{Ande:Nonn:1992}.
Solving the system of equations
	$$
	\begin{cases}
	\bmu = \overline{\bx} \\ 
	v\left(\theta\right) \bSigma = \bS \\ 
	k\left(\theta\right) d\left(d+2\right) = \max\left\{\widehat{\text{Kurt}}(\bX),d\left(d+2\right)\right\} 
	\end{cases}
$$
with respect to $\bmu$, $\bSigma$, and $\theta$, gives raise to the estimates $\widehat{\bmu} = \overline{\bx}$ and $\widehat{\bSigma} = \bS/v(\widehat{\theta})$ for $\bmu$ and $\bSigma$, respectively, where $\widehat{\theta}$ is determined by numerically solving the third equation of the system with respect to $\theta$.
To search for this root in the interval $\left(0,1\right)$, we can use the \texttt{uniroot()} function included in the \textbf{stats} package.
%perform this numerical search in the interval $\left(0,1\right)$, we use the \texttt{optimize()} function, included in the \textbf{stats} package
%The function uniroot searches the interval from lower to upper for a root (i.e., zero) of the function f with respect to its first argument.
%\textcolor{red}{\textbf{To perform this numerical search in the interval $\left(0,1\right)$, we use the \texttt{optimize()} function, included in the \textbf{stats} package, which is a combination of golden section search and successive parabolic interpolation.}} 
Finally note that, the third equation in the system involves the maximum because the MTIN distribution can not be platykurtic.
%does not allow for platikurtic shapes.    
	%$$
	%\begin{cases}
	 %\\ 
	%\widehat{\bSigma} = \bS/v\left(\widehat{\theta}\right) \\ 
	%k\left(\theta\right) d\left(d+2\right) = \widehat{\text{Kurt}}(\bX) 
	%\end{cases}
%$$
%which are the sample mean, the sample covariance matrix, and the sample kurtosis \citep{Mard:1970}, respectively.
%which are unbiased and consistent estimators \citep{Ande:Nonn:1992} of $\bmu$ and $\bSigma$, respectively.
%In light of the formula proposed by \citet{Mard:1970}, the parameter $\beta$ can be estimated starting from
%%for \eqref{eq:kurt}
%\begin{equation*}
%\widehat{\mbox{Kurt}}(\bX) = \frac{1}{n}  \sum_{i=1}^n   \left[\left(\bx_i- \overline{\bx} \right)' \boldsymbol{S}^{-1}\left(\boldsymbol{x}_i - \overline{\bx}  \right)\right]^2.
%\label{eq:estkurt}
%\end{equation*}
%In particular, recalling the constraint given for $\beta$ in Corollary~\ref{cor:1}, the estimator of $\beta$ results to be
%$$
%\widehat{\beta} = \max\left\{0,\min\left[ \widehat{\mbox{Kurt}}(\bX)-d\left(d+2\right),4d,4d(d + 2)/5\right]\right\}.
%$$
%Although the maximum likelihood approach, described in Section~\ref{sec:Maximum likelihood}, is typically preferred, the method of moments appears to be the most natural estimation method and this is due to the fact that the MLN parameters are directly related to the theoretical moments.

\section{Problems with the application of the EM algorithm}
\label{app:application of the EM algorithm}

In the EM algorithm applied to the ML estimation of $\bPsi$ for the MTIN distribution, $\bPsi$ is taken as a whole.
On the $(r+1)$th iteration of the EM algorithm, the E-step is the same as given in Section~\eqref{subsec:E-step} for the ECME algorithm; instead, the two CM-steps are replaced by a single M-step which maximizes $Q\left(\bPsi|\bPsi^{(r)}\right)$ in \eqref{eq:Q}.
As the two terms of $Q\left(\bPsi|\bPsi^{(r)}\right)$, namely 
%According to the right-hand side of \eqref{eq:Q}, 
$Q_1\left(\bmu,\bSigma|\bPsi^{(r)}\right)$ and $Q_2\left(\theta|\bPsi^{(r)}\right)$, have zero-cross derivatives, they
% and, as such, 
can be maximized separately with respect to the parameters they involve.
Maximizing $Q_1\left(\bmu,\bSigma|\bPsi^{(r)}\right)$ with respect to $\bmu$ and $\bSigma$ corresponds to the first CM-step of the ECME algorithm (cf.~Section~\ref{subsec:CM-step 1}).
As it can be noted from \eqref{eq:MTIN complete-data log-likelihood theta}, maximizing $Q_2\left(\theta|\bPsi^{(r)}\right)$ with respect to $\theta$ is equivalent to maximize the log-likelihood function of $n$ independent observations $w_1^{(r)},\ldots,w_n^{(r)}$ from $\mathcal{U}\left(1-\theta,1\right)$.
Hence, we do not need to compute $\text{E}_{\bPsi^{(r)}}\left\{\log\left[\indic_{\left(1-\theta,1\right)}\left(W_i|\bX_i=\bx_i\right)\right]\right\}$ in the E-step.
Therefore, from the standard theory about the uniform distribution \citep[see, e.g.,][Chapter~26]{John:Kotz:cont2:1970}, the update for $\theta$ is
\begin{equation}
\theta^{(r+1)} = 1 - w_{(1)}^{(r)},
\label{eq:M-step theta}
\end{equation}
where $w_{(1)}^{(r)}=\min\left(w_1^{(r)},\ldots,w_n^{(r)}\right)$. 
%Unfortunately, because $w_i^{(r)}$, $i=1,\ldots,n$, is an expected value of a distribution with support $\left(1-\theta^{(r)},1\right)$, as it can be seen from \eqref{eq:wi expectation}, $\theta^{(r+1)}$ can only decrease from one iteration to another.In other words,
%%the iterative structure of the EM algorithm by looking at \eqref{eq:M-step theta}, it is clear that $\theta^{(r+1)}$ is a decreasing function of the iteration $r$. 
%%To be more precise, 
%%this means that 
%$$
%\lim_{r \rightarrow \infty}\theta^{(r+1)} = 0
%$$ 
%regardless from the true value of $\theta$, and this implies a failure to converge for the algorithm.
%%So, the algorithm fails to converge.
Unfortunately, according to \eqref{eq:wi expectation}, it happens that $\theta^{(r+1)} < \theta^{(r)}$ because $w_{(1)}^{(r)} \in \left(1-\theta^{(r)},1\right)$.
% as can be noted from \eqref{eq:wi expectation} where it is clear as $w_i^{(r)}\in \left(1-\theta^{(r)},1\right)$, $i=1,\ldots,n$.
%because $w_i^{(r)}\in \left(1-\theta^{(r)},1\right)$, $i=1,\ldots,n$, (note that $w_i^{(r)}$ is an expected value of a distribution with support $\left(1-\theta^{(r)},1\right)$, see \eqref{eq:wi expectation}) $\theta^{(r+1)}$ can only decrease from one iteration to another.
This means that $\theta^{(r)}$ is a decreasing function of $r$, and this yields
$$
\lim_{r \rightarrow \infty}\theta^{(r)} = 0,
$$
regardless of the true but unknown value of $\theta$.
So, the EM algorithm fails to converge.

\section{Kurtosis of the symmetric generalized hyperbolic distribution}
\label{app:SGHD}

As well-documented in \citet[][Section~3.2]{McNe:Frey:Embr:Quan:2005}, the MSGH distribution belongs to the MNSM family if $V$ in \eqref{eq:nvm} has a generalized inverse Gaussian distribution; in symbols, $V\sim \mathcal{N}^{-}\left(\lambda,\chi,\psi\right)$.
In particular, if $\bX$ has a MSGH distribution, in symbols $\bX \sim \mathcal{SGH}_d\left(\bmu,\bSigma,\lambda,\chi,\psi\right)$, then its pdf is
\begin{equation}
f_{\text{MSGH}}\left(\bx ;\bmu,\bSigma,\lambda,\chi,\psi \right) = \frac{\left(\sqrt{\chi \psi}\right)^{-\lambda}\psi^{\frac{d}{2}}}{\left(2\pi\right)^{\frac{d}{2}}\left|\bSigma\right|^{\frac{1}{2}}K_{\lambda}\left(\sqrt{\chi \psi}\right)} 
\frac{K_{\lambda-\frac{d}{2}}\left(\sqrt{\left[\chi+\delta\left(\bx;\bmu,\bSigma\right)\right]\psi}\right)}{\left\{\sqrt{\left[\chi+\delta\left(\bx;\bmu,\bSigma\right)\right]\psi}\right\}^{\frac{d}{2}-\lambda}},
\label{eq:denMSGH}
\end{equation}
where $K_\lambda(\cdot)$ is the modified Bessel function of the third kind with index $\lambda$ \citep[see][for details]{Abra:Steg:Hand:1965}.
The parameters $\lambda$, $\chi$, and $\psi$ satisfy the following conditions: if $\lambda< 0$, then $\chi>0$ and $\psi\geq 0$; if $\lambda = 0$, then $\chi>0$ and $\psi > 0$; if $\lambda > 0$, then $\chi\geq 0$ and $\psi\geq 0$.

In the following we provide the kurtosis for the MSGH distribution (Theorem~\ref{theo:Kurtosis MSGH}) as well as for the nested SNIG, MSH (Corollary~\ref{cor: SNIG and MSH kurtosis}), and MSVG (Corollary~\ref{cor: MSVG kurtosis}) distributions. 
\begin{theorem}[MSGH: kurtosis]\label{theo:Kurtosis MSGH}
If $\bX \sim \mathcal{SGH}_d\left(\bmu,\bSigma,\lambda,\chi,\psi\right)$, then
\begin{equation}
\mbox{Kurt}\left(\bX\right) = d\left(d+2\right) \frac{K_{\lambda+2}\left(\sqrt{\chi\psi}\right) K_{\lambda}\left(\sqrt{\chi\psi}\right)}{K_{\lambda+1}\left(\sqrt{\chi\psi}\right)^2}. \label{eq:MSGH kurtosis}
\end{equation}
\end{theorem}
\begin{proof}
If $\bX \sim \mathcal{SGH}_d\left(\bmu,\bSigma,\lambda,\chi,\psi\right)$, then it can be expressed as in \eqref{eq:nvm} with $V\sim \mathcal{N}^{-}(\lambda,\chi,\psi)$.
For the non-limiting case $\chi>0$ and $\psi>0$, it results \citep[see][Section~A.2.5]{McNe:Frey:Embr:Quan:2005}
\begin{equation}
 \mbox{E} \left(V^\alpha \right) =  \left(\frac{\chi}{\psi}\right)^{\frac{\alpha}{2}} \frac{K_{\lambda+\alpha}\left(\sqrt{\chi\psi}\right)}{K_{\lambda+\alpha}\left(\sqrt{\chi\psi}\right)}. \quad \alpha\in \real.
\label{eq:momGIG}
\end{equation}
Then, setting $\alpha=1$ and $\alpha=2$ in \eqref{eq:momGIG}, the first and second moments of $V$ are respectively obtained.
Once $\mbox{E} \left(V \right)$ and $\mbox{E} \left(V^2 \right)$ are substituted in \eqref{eq:appkurtosis2}, \eqref{eq:MSGH kurtosis} is obtained.
\end{proof}

\begin{cor}[SNIG and MSH: kurtosis]\label{cor: SNIG and MSH kurtosis}
The kurtosis of the SNIG and MSH distributions are
\begin{equation}
\mbox{Kurt}\left(\bX\right) = d\left(d+2\right) \frac{K_{\frac{3}{2}}\left(\sqrt{\chi\psi}\right) K_{-\frac{1}{2}}\left(\sqrt{\chi\psi}\right)}{K_{\frac{1}{2}}\left(\sqrt{\chi\psi}\right)^2}
\label{eq:SNIG kurtosis}
\end{equation}
and
\begin{equation}
\mbox{Kurt}\left(\bX\right) = d\left(d+2\right) \frac{K_{3}\left(\sqrt{\chi\psi}\right) K_{1}\left(\sqrt{\chi\psi}\right)}{K_{2}\left(\sqrt{\chi\psi}\right)^2},
\label{eq:MSH kurtosis}
\end{equation}
respectively.
\begin{proof}
If $\lambda=-0.5$ and $\lambda=1$ are substituted in \eqref{eq:MSGH kurtosis}, then the results in \eqref{eq:SNIG kurtosis} and \eqref{eq:MSH kurtosis} are obtained, respectively.
\end{proof}
\end{cor}

As concerns the kurtosis of the MSVG distribution, we have a particular simplification of \eqref{eq:MSGH kurtosis} which is detailed in the following corollary. 

\begin{cor}[MSVG: kurtosis]\label{cor: MSVG kurtosis}
If $\lambda>0$ and $\chi=0$ in Theorem~\ref{theo:Kurtosis MSGH}, leading \eqref{eq:denMSGH} to be the pdf of the MSVG distribution, then
%leading \eqref{eq:denMSGH} to be the pdf of the MSVG distribution, 
\begin{equation}
\mbox{Kurt}\left(\bX\right) = d\left(d+2\right) \frac{\lambda+1}{\lambda}. 
\label{eq:kurtMSVG}
\end{equation}
\end{cor}
\begin{proof}
If $\lambda>0$ and $\chi=0$, then the random variable $V$ has a gamma distribution with parameters $\alpha=\lambda$ and $\beta=\frac{1}{2}\psi$.
In this case, the MSGH distribution in \eqref{eq:denMSGH} reduces to the MSVG distribution.
Then, by substituting the first two moments 
$$
\mbox{E}(V)=\frac{2\lambda}{\psi}  \quad  \mbox{and}  \quad \mbox{E} \left(V^2 \right)=\frac{4\lambda(\lambda+1)}{\psi^2}
$$
in \eqref{eq:appkurtosis2}, the kurtosis in \eqref{eq:kurtMSVG} is straightforwardly obtained.
\end{proof}

\section{Partial derivatives of $l\left(\bPsi^*\right)$ in \eqref{eq:MTIN observed-data log-likelihood rep}}
\label{app:Partial derivatives}

The first order partial derivatives of $l\left(\bPsi^*\right)$ in \eqref{eq:MTIN observed-data log-likelihood rep}, with respect to $\bmu$, $\bOmega $ and $\gamma$, are respectively given by
\begin{align}
\frac{ \partial\: l\left(\bPsi^*\right)}{\partial \bmu } = & 2 \left(\frac{d}{2}+1\right)\sum_{i=1}^n \frac{1}{\delta\left(\bx_i;\bmu,\bOmega' \bOmega \right)} \bOmega ^{-1}  \left(\bOmega'\right)^{-1}  \left(\bx_i-\bmu  \right) \nonumber \\
& + \sum_{i=1}^n \nu\left(\bx_i,\bmu,\bOmega ,\gamma\right)  \bOmega ^{-1}  \left(\bOmega'\right)^{-1}  \left(\bx_i - \bmu \right) ,
\nonumber
%\label{dermulam}
\\
\frac{ \partial\: l\left(\bPsi^*\right)}{\partial \bOmega } = & - \frac{n}{2} \left(\bOmega'\right)^{-1}+  2 \left(\frac{d}{2}+1\right)\sum_{i=1}^n \frac{1}{\delta\left(\bx_i;\bmu,\bOmega' \bOmega \right)}  \left(\bOmega'\right)^{-1}  \left(\bx_i-  \bmu  \right) \left(\bx_i-  \bmu  \right)'  \bOmega ^{-1}  \left(\bOmega'\right)^{-1}  \nonumber \\
  & +\sum_{i=1}^n \nu\left(\bx_i,\bmu,\bOmega ,\gamma\right) \left(\bOmega'\right)^{-1}  \left(\bx_i-  \bmu  \right) \left(\bx_i-  \bmu  \right)'  \bOmega ^{-1}  \left(\bOmega'\right)^{-1}, \nonumber
	%\label{deromega} 
	\\
	\frac{ \partial\: l\left(\bPsi^*\right)}{\partial \gamma} = & - \frac{n\exp(\gamma)}{1+\exp(\gamma)} + \left[1+\exp(\gamma)\right]^{-\frac{d}{2}-1} 
	\sum_{i=1}^n 
\frac{\left[\frac{\delta\left(\bx_i;\bmu,\bOmega' \bOmega \right)}{2}\right]^{\frac{d}{2}+1}\exp\left\{-\frac{\delta\left(\bx_i;\bmu,\bOmega' \bOmega \right)}{2\left[1+\exp(\gamma)\right]}\right\}}
{\Gamma\left[\frac{d}{2}+1,\frac{\delta\left(\bx_i;\bmu,\bOmega' \bOmega \right)}{2 \left[1+\exp(\gamma)\right]}\right] - \Gamma\left[\frac{d}{2}+1,\frac{\delta\left(\bx_i;\bmu,\bOmega' \bOmega \right)}{2} \right]}, 
\label{dergamma}
\end{align}
where 
$$
\nu\left(\bx_i,\bmu,\bOmega ,\gamma\right) =
\frac{\left[\frac{\delta\left(\bx_i;\bmu,\bOmega' \bOmega \right)}{2}\right]^\frac{d}{2}\exp\left[-\frac{\delta\left(\bx_i;\bmu,\bOmega' \bOmega \right)}{2}\right]\left\{\left[1+\exp(\gamma)\right]^{-\frac{d}{2}-1}\exp\left(\frac{\exp(\gamma)\delta\left(\bx_i;\bmu,\bOmega' \bOmega \right)}{2\left[1+\exp(\gamma)\right]}\right)-1\right\}}
{\Gamma\left[\frac{d}{2}+1,\frac{\delta\left(\bx_i;\bmu,\bOmega' \bOmega \right)}{2 \left[1+\exp(\gamma)\right]}\right] - \Gamma\left[\frac{d}{2}+1,\frac{\delta\left(\bx_i;\bmu,\bOmega' \bOmega \right)}{2} \right] }.
$$

\end{document}